\documentclass[final,1p,times,authoryear]{elsarticle}
\usepackage{amsmath}
\usepackage{amsthm}
\usepackage{amssymb}
\usepackage{amsfonts}
\usepackage{maplestd2e}
\usepackage{multicol}
\usepackage{multirow}
\usepackage{enumitem}
\usepackage{caption}
\usepackage{algorithm,algpseudocode}
\usepackage{algorithmicx}
\usepackage{caption}
\usepackage{stmaryrd}
\usepackage{xcolor}
\usepackage{times}
\usepackage{color}
\usepackage{makecell}
\DefineParaStyle{Maple Bullet Item}
\DefineParaStyle{Maple Heading 1}
\DefineParaStyle{Maple Warning}
\DefineParaStyle{Maple Heading 4}
\DefineParaStyle{Maple Heading 2}
\DefineParaStyle{Maple Heading 3}
\DefineParaStyle{Maple Dash Item}
\DefineParaStyle{Maple Error}
\DefineParaStyle{Maple Title}
\DefineParaStyle{Maple Text Output}
\DefineParaStyle{Maple Normal}
\DefineCharStyle{Maple 2D Output}
\DefineCharStyle{Maple 2D Input}
\DefineCharStyle{Maple Maple Input}
\DefineCharStyle{Maple 2D Math}
\DefineCharStyle{Maple Hyperlink}
\usepackage[
breaklinks=true,colorlinks=true,
linkcolor=red,urlcolor=black,citecolor=blue,% PRINT
bookmarks=true,bookmarksopenlevel=2]{hyperref}
\usepackage{moreverb}
\usepackage{verbatim}
\usepackage{sverb}
\usepackage{geometry}
\DeclareMathOperator{\arctanh}{arctanh}
\DeclareMathOperator{\arcsech}{arcsech}
\DeclareMathOperator{\arccosh}{arccosh}

\DeclareMathOperator{\arcsinh}{arcsinh}
\DeclareMathOperator{\lcm}{lcm}

\newtheorem{theorem}{Theorem}
\newtheorem{lemma}[theorem]{Lemma}

\newtheorem{definition}[theorem]{Definition}

\newtheorem{example}[theorem]{Example}
\newtheorem{remark}[theorem]{Remark}

\newtheorem{acknowledgment}[]{Acknowledgment}
\begin{document}

\begin{frontmatter}

\title{Symbolic computation of hypergeometric type and non-holonomic power series}

\author{Bertrand Teguia Tabuguia}
%\address{University of Kassel, Heinrich-Plett-Str.40. 34132 Kassel, Germany}
\ead{bteguia@mathematik.uni-kassel.de}
\ead[url]{https://www.bertrandteguia.com}

\author{Wolfram Koepf}
\address{University of Kassel, Heinrich-Plett-Str.40. 34132 Kassel, Germany}
\ead{koepf@mathematik.uni-kassel.de}
\ead[url]{http://www.mathematik.uni-kassel.de/~koepf}

\begin{abstract}
A term $a_n$ is $m$-fold hypergeometric, for a given positive integer $m$, if the ratio $a_{n+m}/a_n$ is a rational function over a field $\mathbb{K}$ of characteristic zero. We establish the structure of holonomic recurrence equation, i.e. linear and homogeneous recurrence equations having polynomial coefficients, that have $m$-fold hypergeometric term solutions over $\mathbb{K}$, for any positive integer $m$. Consequently, we describe an algorithm, say mfoldHyper, that extends van Hoeij's algorithm (1998) which computes a basis of the subspace of hypergeometric $(m=1)$ term solutions of holonomic recurrence equations to the more general case of $m$-fold hypergeometric terms.

A Laurent-Puiseux series
\begin{equation}
	\sum_{n=n_0}^{\infty}a_n(z-z_0)^{n/k}~ (a_n\in\mathbb{K}, k\in\mathbb{N}, n_0\in\mathbb{Z}), \label{ref00}
\end{equation}
\noindent where $k$ denotes the corresponding Puiseux number, is mainly characterized by the coefficient $a_n$. We generalize the concept of hypergeometric type power series introduced by Koepf (1992), by considering linear combinations of Laurent-Puiseux series whose coefficients are $m$-fold hypergeometric terms. Such power series could not be computed before due to the lack of an algorithm to find $m$-fold hypergeometric term solutions of holonomic recurrence equation which constitute a key step in Koepf's procedure. Thanks to mfoldHyper, we deduce a complete procedure to compute these power series; indeed, it turns out that every linear combination of power series with $m$-fold hypergeometric term coefficients, for finitely many values of $m$, is detected.

On the other hand, we investigate an algorithm to represent power series of non-holonomic functions like $\tan(z), (1-\tan(z))/(1+\tan(z)),$ $z/(\exp(z)-1), \log(1+\sin(z)), \text{etc}$. The algorithm follows the same steps of Koepf's algorithm, but instead of seeking holonomic differential equations, quadratic differential equations are computed and the Cauchy product rule is used to deduce recurrence equations for the power series coefficients. This algorithm defines a normal function that yields together with enough initial values normal forms for many power series of non-holonomic functions. Therefore, non-trivial identities like
\[
\ln\left(\frac{1+\tan(z)}{1-\tan(z)}\right)=2\arctanh\left(\frac{\sin(2z)}{1+\cos(2z)}\right)
\]
are automatically proved using this approach. This paper is accompanied by implementations in the Computer Algebra Systems (CAS) Maxima 5.44.0 and Maple 2019.
\end{abstract}

\end{frontmatter}

\section{Introduction}\label{sec0}

The applicability of complex analysis is essentially restricted to analytic functions, since it easily allows both differentiation and integration. These functions are represented by power series with positive radius of convergence. Power series are used to represent orthogonal polynomials (\cite{koepf1997representations}); in combinatorics, generating functions are power series (\cite{combinator}); in dynamical systems, algebraic properties of power series involve most of the constructions (see \cite{dynsys}); we can also enumerate commutative algebra and algebraic geometry (\cite{intro1}), \cite[Chapter VII]{gathmann2006tropical}. It is therefore important to know the exact general coefficient or formula of a power series. There is no algorithm which computes the power series of every given analytic function. We classify series with a certain common property, and build an algorithm which will always find the power series representation from such an analytic expression, whenever possible. It is important to notice the word "expression", because we are not considering complex functions as abstract objects defined in a certain domain and its range, but instead as a differentiable object that we can manipulate symbolically to characterize its Taylor coefficients by a certain type of linear recurrence equation. Moreover, by the unique power series characterization, this approach does not only lead to the verification of known identities, but also to the discovery of new ones.

Let $\mathbb{K}$ be an infinite computable field\footnote{Mostly $\mathbb{K}:=\mathbb{Q}(\alpha_1,\ldots,\alpha_N)$ is the field of rational functions in several variables}, and $(a_n)_{n\in\mathbb{Z}}, a_n\in\mathbb{K},$ be an $m$-fold hypergeometric sequence such that
\begin{equation}
	a_{n+m} = r(n) a_n,~\forall n>n_0,~n_0\in\mathbb{Z}, \label{ref1}
\end{equation}

\noindent where $r(n)$ denotes a rational function in $\mathbb{K}(n)$, $m \in \mathbb{N}$, and $n_0$ is the first non-zero term index. $m$-fold hypergeometric sequences are very useful in summation theory (\cite{WolfBook}). Our first interest is to describe an algorithm which computes power series (Puiseux series) of the form

\begin{equation}
	\sum_{n=n_0}^{\infty}a_n(z-z_0)^{n/k}~ (a_n\in\mathbb{K}, k\in\mathbb{N}, n_0\in\mathbb{Z}), \label{ref0}
\end{equation}

\noindent such that $a_n$ is an $m$-fold hypergeometric term.

In 1992, Koepf published an algorithmic approach for computing power series (see \cite{Koepf1992}). The algorithm was implemented in the computer algebra systems (CAS) Maple (\cite{heck1993introduction}) and Mathematica (\cite{wolfram1999mathematica}). In his original approach, Koepf considered three types of functions: \textit{two-term recurrence relation} type which corresponds to expressions leading to a linear recurrence equation equivalent to $(\ref{ref1})$. That is
\begin{equation}
	Q_n a_{n+m} + P_n a_n = 0,~n\in\mathbb{Z}, \label{ref2}
\end{equation}
where $Q_n,P_n$ are polynomials in $\mathbb{K}[n]$. The second type called \textit{exp-like}, corresponding to expressions leading to linear recurrence equations with constant coefficients in $\mathbb{K}$. And the third type with a completely different approach based on partial fraction decomposition (over $\mathbb{C}$) corresponding to rational functions in $\mathbb{C}(z)$. All gathered in the Maple and Mathematica packages FPS could already recover the power series formulas of a large family of analytic functions.

Note that in the rational function case, the algorithm can still find a linear recurrence equation satisfied by the general coefficient sought, but the issue was in solving that equation. Furthermore, it turns out that the general coefficient found for each type used in Koepf's approach is always a linear combination of $m$-fold hypergeometric terms. Therefore, if we could find all $m$-fold hypergeometric term solutions of a linear homogeneous recurrence equation, then we could considerably increase the family of power series computed automatically.

Marko Petkov{\v{s}}ek later published an algorithm which finds all hypergeometric ($m=1$) term solutions of linear recurrences (\cite{petkovvsek1992hypergeometric}). This algorithm was implemented in Maple by Koepf and  in Mathematica by Petkov{\v{s}}ek. Petkov{\v{s}}ek brilliantly used tools involved in Gosper's algorithm (see Chapter 5 in ~\cite{WolfBook}) in his approach. However, the complexity of Petkov{\v{s}}ek's algorithm can be very high depending on the degree of polynomial coefficients of the equation. 

In 1999, Mark van Hoeij used a different approach and got a much more efficient algorithm for the same purpose. Indeed, he considered the local behavior of solution terms, which naturally decreases the complexity by reducing the number of candidates since hypergeometric term solutions are built from some factors of the leading and the trailing polynomial coefficients (\cite{van1999finite}). Van Hoeij implemented his algorithm in Maple as \textit{LREtools[hypergeomsols]}. An equivalent algorithm is part of our main approach, this is described in (\cite{BThyper}).

Note, however, that the Petkov{\v{s}}ek and van Hoeij algorithms might only find hypergeometric term solutions in an extension field of $\mathbb{Q}$, which in certain cases, for $m>1$, can be equivalent to $m$-fold hypergeometric term solutions in $\mathbb{Q}$. Indeed, the algorithm is implemented to find all hypergeometric term solutions in $\mathbb{Q}(\alpha)$, where $\alpha \in \mathbb{C}\setminus \mathbb{Q}$; since $\alpha$ is not always explicitly known in advance, we will often replace extension fields of $\mathbb{Q}$ by $\mathbb{C}$. But this has some disadvantages of simplicity. If we consider the power series of the cosine function at $z_0=0$ given by
\begin{equation}
	\cos(z)=\sum_{n=0}^{\infty}\dfrac{(-1)^n}{(2n)!}z^{2n} \label{ref3},
\end{equation}
then we observe that its general coefficient satisfies the recurrence equation
\begin{equation}
	\left( 1+n\right) \left( 2+n\right) {{a}_{n+2}}+{{a}_{n}}=0. \label{ref4}
\end{equation}

Using Koepf's algorithm, the type $m=2$ is detected and the formula $(\ref{ref3})$ is obtained as a linear combination of the two 2-fold hypergeometric series involved, provided the initial values $a_0=1$ and $a_1=0$.

Using van Hoeij's algorithm implemented in Maple\footnote{We used Maple 2019 for this paper}, with the same initial values, we find the hypergeometric solution
\begin{equation}
	\frac{\text{i}^n}{2\Gamma(n+1)} + \frac{(-\text{i})^n}{2\Gamma(n+1)}, ~ \text{i}\in\mathbb{C},\text{i}^2=-1 \label{ref5}
\end{equation}
which gives the expansion
\begin{equation}
	\cos(z)=\sum_{n=0}^{\infty}\dfrac{(-\text{i})^n+\text{i}^n}{2\Gamma(n+1)}z^n \label{ref6}.
\end{equation}

Therefore thanks to Koepf's algorithm, Maple treats the cosine case well in $\mathbb{Q}$ since the recurrence equation obtained is a two-term recurrence relation. In general, an issue occurs with unnecessary algebraic extensions of $\mathbb{Q}$ when van Hoeij's algorithm is used, because it only looks for hypergeometric term solutions ($m$=1). For example, any linear combination of $\cos(z)$ or $\sin(z)$ with an expression having a hypergeometric general coefficient will have a formula involving $(\ref{ref5})$. 
\begin{maplegroup}
	\begin{mapleinput}
		\mapleinline{active}{1d}{convert(cos(z)+exp(z),FPS);
		}{}
	\end{mapleinput}
	\mapleresult
	\begin{maplelatex}
		\[\displaystyle 
		\sum_{k=0}^{\infty}\dfrac{\left(1+\frac{\text{i}^k}{2}+\frac{(-\text{i})^k}{2}\right) z^k}{k!}\]
	\end{maplelatex}
\end{maplegroup}
\begin{maplegroup}
	\begin{mapleinput}
		\mapleinline{active}{1d}{convert(log(1+z)+sin(z),FPS);
		}{}
	\end{mapleinput}
	\mapleresult
	\begin{maplelatex}
		\[\displaystyle \sum_{k=0}^{\infty} \left(-\frac{(-1)^{k+1}}{k+1}-\frac{\text{i}\cdot\text{i}^{k+1}}{2(k+1)!} + \frac{\text{i}\cdot(-\text{i})^{k+1}}{2(k+1)!}\right) z^{k+1} \]
	\end{maplelatex}
\end{maplegroup}

Note, however, that this paper does not aim to find the power series formula with a simple hypergeometric general coefficient, but to find the formula with the simplest $m$-fold hypergeometric general coefficients. Simple here means that the coefficients are not taken in an extension field of $\mathbb{Q}$ whenever there exists an $m$-fold equivalent over $\mathbb{Q}$. We should highlight $m$-fold hypergeometric, because up to now there is no implemented algorithm able to find such solutions of a linear recurrence equation. And it is worth to have such an algorithm since in many cases, Maple's \textit{convert}\footnote{Maple's convert command uses Koepf's original approach followed by an invocation of van Hoeij's algorithm.} command fails to find power series of this type. 
\begin{maplegroup}
	\begin{mapleinput}
		\mapleinline{active}{1d}{convert(arcsin(z)+cos(z),FPS);
		}{}
	\end{mapleinput}
	\mapleresult
	\begin{maplelatex}
		\[\displaystyle \arcsin(z) + \cos(z)\]
	\end{maplelatex}
\end{maplegroup}
\begin{maplegroup}
	\begin{mapleinput}
		\mapleinline{active}{1d}{convert(exp(z\symbol{94}2)+log(1+z\symbol{94}3),FPS);
		}{}
	\end{mapleinput}
	\mapleresult
	\begin{maplelatex}
		\[\displaystyle e^{z^2} + \ln(1+z^{3})\]
	\end{maplelatex}
\end{maplegroup}

The above Maple failures rely on the incapacity of van Hoeij's algorithm to detect $m$-fold ($m>1$) hypergeometric term solutions of so-called holonomic recurrence equations, that is homogeneous linear recurrence equations with polynomial coefficients. Indeed, by using the Maple package \textit{FormalPowerSeries} we get the following holonomic recurrence equations.

\begin{maplegroup}
	\begin{mapleinput}
		\mapleinline{active}{1d}{bind(FormalPowerSeries)
		}{}
	\end{mapleinput}
\end{maplegroup}
\begin{maplegroup}
	\begin{mapleinput}
		\mapleinline{active}{1d}{RE1:=SimpleRE(arcsin(z)+cos(z),z,a(n));
		}{}
	\end{mapleinput}
	\mapleresult
	\begin{maplelatex}
		\vspace{-0.5cm}
		\begin{multline*}
			\text{RE1 :=}-n(n^3-10n^2+21n-22)a(n)+(n-4)^2a(n-4)+(n-2)(n^3-11n^2+39n-41)a(n-2)\\
			+2(n+1)(n+2)(n^2+4n-1)a(n+2)-2(n+1)(n+2)(n+3)(n+4)a(n+4) = 0
		\end{multline*}
	\end{maplelatex}
\end{maplegroup}
\begin{maplegroup}
	\begin{mapleinput}
		\mapleinline{active}{1d}{RE2:=SimpleRE(exp(z\symbol{94}2)+log(1+z\symbol{94}3),z,a(n));
		}{}
	\end{mapleinput}
	\mapleresult
	\begin{maplelatex}
		\begin{multline*}
			\text{RE2 :=} -4(n-9)^2a(n-9)+2(n-13)(n-7)^2a(n-7)-4(n-6)(2n-15)a(n-6)\\
			+2(n-7)(n-5)^2a(n-5)+2(n-4)(2n^2-28n+107)a(n-4)-4(n-3)(n-6)a(n-3)\\
			+(n-2)(n-4)(n-17)a(n-2)+2(n-1)(n-4)^2a(n-1)-(n-1)(n-2)(n+1)a(n+1) = 0
		\end{multline*}
	\end{maplelatex}
\end{maplegroup}
Applying van Hoeij's algorithm to these two recurrence equations yields
\begin{maplegroup}
	\begin{mapleinput}
		\mapleinline{active}{1d}{LREtools[hypergeomsols](RE1,a(n),\{\},output=basis);}{}
	\end{mapleinput}
	\mapleresult
	\begin{maplelatex}
		\[\displaystyle \left[\dfrac{(-\text{i})^n}{\Gamma(n+1)},~\dfrac{\text{i}^n}{\Gamma(n+1)}\right] \]
	\end{maplelatex}
\end{maplegroup}
\begin{maplegroup}
	\begin{mapleinput}
		\mapleinline{active}{1d}{LREtools[hypergeomsols](RE2,a(n),\{\},output=basis);}{}
	\end{mapleinput}
	\mapleresult
	\begin{maplelatex}
		\[\displaystyle \left[\dfrac{(-1)^n}{n},~\dfrac{\left(\frac{1}{2}-\frac{\text{i}\sqrt{3}}{2}\right)^n}{n},~\dfrac{\left(\frac{1}{2}+\frac{\text{i}\sqrt{3}}{2}\right)^n}{n}\right]\]
	\end{maplelatex}
\end{maplegroup}
\noindent showing that the general coefficients of $\arcsin(z)$ in RE1 and the one of $\exp(z^2)$ in RE2 are missed.

Although some algorithms for computing $m$-fold hypergeometric term solutions of holonomic recurrence equations have been described, none of them is implemented. For example, in (\cite{cluzeau2006computing}) and (\cite{van1999finite}) an algorithm using linear operators is developed, but the described approach needs non-commutative factorization for its implementation. In our approach however, non-commutative algebra is not needed. We will use a different view on holonomic recurrence equations and develop a new algorithm to detect all their $m$-fold hypergeometric term solutions. Thus with our Maple and Maxima\footnote{We used Maxima 5.44.0 for this paper.} implementations, the issue with $m$-fold hypergeometric term solutions of holonomic recurrence equations is solved as demonstrated in Maxima below (see \cite{BTphd}).

\noindent
\begin{minipage}[t]{8ex}\color{red}\bf
	\begin{verbatim}
		(%i1) 
	\end{verbatim}
\end{minipage}
\begin{minipage}[t]{\textwidth}\color{blue}
	\begin{verbatim}
		RE1:FindRE(asin(z)+cos(z),z,a[n]);
	\end{verbatim}
\end{minipage}
\definecolor{labelcolor}{RGB}{100,0,0}
\begin{multline*}
	\displaystyle
	\parbox{10ex}{$\color{labelcolor}\mathrm{\tt (\%o1) }\quad $}
	-2\cdot \left( 1+n\right) \cdot \left( 2+n\right) \cdot \left( 3+n\right) \cdot \left( 4+n\right) \cdot {{a}_{n+4}}+2\cdot \left( 1+n\right) \cdot \left( 2+n\right) \cdot \\
	\left( -1+4\cdot n+{{n}^{2}}\right) \cdot {{a}_{n+2}}-n\cdot \left( -22+21\cdot n-10\cdot {{n}^{2}}+{{n}^{3}}\right) \cdot {{a}_{n}}
	\\+\left( n-2\right) \cdot \left( -41+39\cdot n-11\cdot {{n}^{2}}+{{n}^{3}}\right) \cdot {{a}_{n-2}}+{{\left( n-4\right) }^{2}}\cdot {{a}_{n-4}}=0\mbox{}
\end{multline*}

\noindent
\begin{minipage}[t]{8ex}\color{red}\bf
	\begin{verbatim}
		(%i2) 
	\end{verbatim}
\end{minipage}
\begin{minipage}[t]{\textwidth}\color{blue}
	\begin{verbatim}
		mfoldHyper(RE1,a[n]);
	\end{verbatim}
\end{minipage}
\definecolor{labelcolor}{RGB}{100,0,0}
\[\displaystyle
\parbox{10ex}{$\color{labelcolor}\mathrm{\tt (\%o2) }\quad $}
\left[\left[2,\left\{ \frac{{{\left( -1\right) }^{n}}}{\left( 2\cdot n\right) !},\frac{{{4}^{n}}\cdot {{n!}^{2}}}{{{n}^{2}}\cdot \left( 2\cdot n\right) !} \right\}\right]\right]\mbox{}
\]

\noindent
\begin{minipage}[t]{8ex}\color{red}\bf
	\begin{verbatim}
		(%i3) 
	\end{verbatim}
\end{minipage}
\begin{minipage}[t]{\textwidth}\color{blue}
	\begin{verbatim}
		RE2:FindRE(exp(z^2)+log(1+z^3),z,a[n]);
	\end{verbatim}
\end{minipage}
\definecolor{labelcolor}{RGB}{100,0,0}
\begin{multline*}
	\displaystyle
	\parbox{10ex}{$\color{labelcolor}\mathrm{\tt (\%o3) }\quad $}
	-\left( n-2\right) \cdot \left( n-1\right) \cdot \left( 1+n\right) \cdot {{a}_{n+1}}+2\cdot {{\left( n-4\right) }^{2}}\cdot \left( n-1\right) \cdot {{a}_{n-1}}\\
	+\left( n-17\right) \cdot \left( n-4\right) \cdot \left( n-2\right) \cdot {{a}_{n-2}}-4\cdot \left( n-6\right) \cdot \left( n-3\right) \cdot {{a}_{n-3}}\\
	+2\cdot \left( n-4\right) \cdot \left( 107-28\cdot n+2\cdot {{n}^{2}}\right) \cdot {{a}_{n-4}}+2\cdot \left( n-7\right) \cdot {{\left( n-5\right) }^{2}}\cdot {{a}_{n-5}}\\
	-4\cdot \left( n-6\right) \cdot \left( 2\cdot n-15\right) \cdot {{a}_{n-6}}+2\cdot \left( n-13\right) \cdot {{\left( n-7\right) }^{2}}\cdot {{a}_{n-7}}\\
	-4\cdot {{\left( n-9\right) }^{2}}\cdot {{a}_{n-9}}=0\mbox{}
\end{multline*}

\noindent
\begin{minipage}[t]{8ex}\color{red}\bf
	\begin{verbatim}
		(%i4) 
	\end{verbatim}
\end{minipage}
\begin{minipage}[t]{\textwidth}\color{blue}
	\begin{verbatim}
		mfoldHyper(RE2,a[n]);
	\end{verbatim}
\end{minipage}
\definecolor{labelcolor}{RGB}{100,0,0}
\[\displaystyle
\parbox{10ex}{$\color{labelcolor}\mathrm{\tt (\%o4) }\quad $}
\left[\left[1,\left\{\frac{{{\left( -1\right) }^{n}}}{n}\right\}\right],\left[2,\left\{\frac{1}{n!}\right\}\right],\left[3,\left\{\frac{{{\left( -1\right) }^{n}}}{n}\right\}\right]\right]\mbox{}
\]

\noindent
\begin{minipage}[t]{8ex}\color{red}\bf
	\begin{verbatim}
		(%i5) 
	\end{verbatim}
\end{minipage}
\begin{minipage}[t]{\textwidth}\color{blue}
	\begin{verbatim}
		FPS(asin(z)+cos(z),z,n);
	\end{verbatim}
\end{minipage}
\definecolor{labelcolor}{RGB}{100,0,0}
\[\displaystyle
\parbox{10ex}{$\color{labelcolor}\mathrm{\tt (\%o5) }\quad $}
\left( \sum_{n=0}^{\infty }\frac{\left( 2\cdot n\right) !\cdot {{z}^{1+2\cdot n}}}{\left( 2\cdot n+1\right) \cdot {{4}^{n}}\cdot {{n!}^{2}}}\right) +\sum_{n=0}^{\infty }\frac{{{\left( -1\right) }^{n}}\cdot {{z}^{2\cdot n}}}{\left( 2\cdot n\right) !}\mbox{}
\]

\noindent
\begin{minipage}[t]{8ex}\color{red}\bf
	\begin{verbatim}
		(%i6) 
	\end{verbatim}
\end{minipage}
\begin{minipage}[t]{\textwidth}\color{blue}
	\begin{verbatim}
		FPS(exp(z^2)+log(1+z^3),z,n);
	\end{verbatim}
\end{minipage}
\definecolor{labelcolor}{RGB}{100,0,0}
\[\displaystyle
\parbox{10ex}{$\color{labelcolor}\mathrm{\tt (\%o6) }\quad $}
\left( \sum_{n=0}^{\infty }\frac{{{\left( -1\right) }^{n}}\cdot {{z}^{3\cdot \left( 1+n\right) }}}{n+1}\right) +\sum_{n=0}^{\infty }\frac{{{z}^{2\cdot n}}}{n!}\mbox{}
\]

Another important issue that we solve is the step which consists in deducing, when it exists, the correct linear combination needed to find the hypergeometric type representation sought. Let $P_0(z), P_1(z),\ldots,P_d(z)$ be $d+1$ polynomials in $\mathbb{K}(z)$, and $f_1(z),\ldots, f_d(z)$ some analytic expressions that have $m$-fold hypergeometric term coefficients in their power series expansions. More generally, our algorithm handles formal series formulas of expressions of the form
\begin{equation}
	P_0(z) + \sum_{j=1}^{d} P_j(z)f_j(z). \label{geninput}
\end{equation}

The output of such an input is of course a linear combination of hypergeometric type series, plus a polynomial which might be zero. If the correct linear combination of $m$-fold hypergeometric term solutions of the corresponding linear recurrence equation is not found, then the output might be missed. This happens sometimes with Maple for the hypergeometric ($m=1$) case. For example, Maple gives

\begin{maplegroup}
	\begin{mapleinput}
		\mapleinline{active}{1d}{convert((z+z\symbol{94}2+1)*exp(z)+(z\symbol{94}3+3)*log(1+z),FPS);
		}{}
	\end{mapleinput}
	\mapleresult
	\begin{maplelatex}
		\[\displaystyle (z+z^2+1)e^z + (z^3+3)\ln(z+1)\]
	\end{maplelatex}
\end{maplegroup}
\begin{maplegroup}
	\begin{mapleinput}
		\mapleinline{active}{1d}{}{}
	\end{mapleinput}
\end{maplegroup}
\vspace{-0.7cm}
\begin{maplegroup}
	\begin{mapleinput}
		\mapleinline{active}{1d}{convert(1+z+z\symbol{94}2+z\symbol{94}3*arctan(z),FPS);
		}{}
	\end{mapleinput}
	\mapleresult
	\begin{maplelatex}
		\[\displaystyle  1+z+z^2+z^3\cdot\arctan(z)\]
	\end{maplelatex}
\end{maplegroup}
\begin{maplegroup}
	\begin{mapleinput}
		\mapleinline{active}{1d}{}{}
	\end{mapleinput}
\end{maplegroup}
\vspace{-0.5cm}
\noindent whereas our algorithm yields correctly

\noindent
\begin{minipage}[t]{8ex}\color{red}\bf
	\begin{verbatim}
		(%i7) 
	\end{verbatim}
\end{minipage}
\begin{minipage}[t]{\textwidth}\color{blue}
	\begin{verbatim}
		FPS((z+z^2+1)*exp(z)+(z^3+3)*log(1+z),z,n);
	\end{verbatim}
\end{minipage}
\definecolor{labelcolor}{RGB}{100,0,0}
\vspace{-0.5cm}
\begin{small}
	\begin{multline}
		\parbox{10ex}{$\color{labelcolor}\mathrm{\tt (\%o7) }\quad $} \frac{8\cdot {{z}^{3}}}{3}+{{z}^{2}}+5\cdot z+1 \\
		+\left( \sum_{n=0}^{\infty }-\frac{\left( -68-117\cdot n-61\cdot {{n}^{2}}-13\cdot {{n}^{3}}-{{n}^{4}}-{{\left( -1\right) }^{n}}\cdot \left( 4+n\right) !+2\cdot n\cdot {{\left( -1\right) }^{n}}\cdot \left( 4+n\right) !\right) \cdot {{z}^{4+n}}}{\left( n+1\right) \cdot \left( n+4\right) \cdot \left( 4+n\right) !}\right)\mbox{} \label{polypart1}
	\end{multline}
\end{small}
\begin{minipage}[t]{8ex}\color{red}\bf
	\begin{verbatim}
		(%i8) 
	\end{verbatim}
\end{minipage}
\begin{minipage}[t]{\textwidth}\color{blue}
	\begin{verbatim}
		FPS(1+z+z^2+z^3*atan(z),z,n);
	\end{verbatim}
\end{minipage}
\definecolor{labelcolor}{RGB}{100,0,0}
\[\displaystyle
\parbox{10ex}{$\color{labelcolor}\mathrm{\tt (\%o8) }\quad $}
\left( \sum_{n=0}^{\infty }\frac{{{\left( -1\right) }^{n}}\cdot {{z}^{2\cdot n}}}{2\cdot n-3}\right) +z+\frac{4}{3}\mbox{}
\]

Note that for these two latter examples van Hoeij's algorithm finds the corresponding hypergeometric terms

\begin{maplegroup}
	\begin{mapleinput}
		\mapleinline{active}{1d}{LREtools[hypergeomsols](SimpleRE((z\symbol{94}2+z+1)*exp(z)+(z\symbol{94}3+3)
			*log(1+z),z,a(n)),a(n),\{\},output=basis);}{}
	\end{mapleinput}
	\mapleresult
	\begin{maplelatex}
		\[\displaystyle \left[\dfrac{(-1)^n(2n-9)}{(n-3)n},~\dfrac{(n^2+1)}{\Gamma(n+1)}\right]\]
	\end{maplelatex}
\end{maplegroup}
\begin{maplegroup}
	\begin{mapleinput}
		\mapleinline{active}{1d}{LREtools[hypergeomsols](SimpleRE(1+z+z\symbol{94}2+z\symbol{94}3*arctan(z),
			z,a(n)),a(n),\{\},output=basis);}{}
	\end{mapleinput}
	\mapleresult
	\begin{maplelatex}
		\[\displaystyle \left[\dfrac{i^n}{n-3},~\dfrac{(-1)^n}{n-3}\right]\]
	\end{maplelatex}
\end{maplegroup}
\noindent but the power series terms are missed by the Maple command \textit{convert}. We mention that this issue is not related to an argument of \textit{convert} which has to be specified, in particular the order of the differential equations involved in the computations. Indeed the default value used for the upper bound of the differential equations sought for power series computations is $4$. However, using our Maxima procedure \textit{HolonomicDE} which implements a variant of Koepf's algorithm to compute holonomic differential equations, one finds the following differential equations of order less than $4$.

\noindent
\begin{minipage}[t]{8ex}\color{red}\bf
	\begin{verbatim}
		(%i9) 
	\end{verbatim}
\end{minipage}
\begin{minipage}[t]{\textwidth}\color{blue}
	\begin{verbatim}
		HolonomicDE((z+z^2+1)*exp(z)+(z^3+3)*log(1+z),F(z));
	\end{verbatim}
\end{minipage}
\definecolor{labelcolor}{RGB}{100,0,0}
\begin{small}
	\begin{multline*}\displaystyle
		\parbox{10ex}{$\color{labelcolor}\mathrm{\tt (\%o9) }\quad $}
		\left( 1+z\right) \cdot \left( 63+99\cdot z-18\cdot {{z}^{2}}-84\cdot {{z}^{3}}-9\cdot {{z}^{4}}+33\cdot {{z}^{5}}+4\cdot {{z}^{6}}-2\cdot {{z}^{7}}+{{z}^{8}}+{{z}^{9}}\right)\\
		\cdot \left( \frac{{{d}^{3}}}{d\,{{z}^{3}}}\cdot \mathrm{F}\left( z\right) \right)-\big( 36+27\cdot z-171\cdot {{z}^{2}}-222\cdot {{z}^{3}}+54\cdot {{z}^{4}}+147\cdot {{z}^{5}}+39\cdot {{z}^{6}}+14\cdot {{z}^{8}}+9\cdot {{z}^{9}}\\
		+{{z}^{10}}\big)\cdot \left( \frac{{{d}^{2}}}{d\,{{z}^{2}}}\cdot \mathrm{F}\left( z\right) \right)+\big( -162-441\cdot z-441\cdot {{z}^{2}}-15\cdot {{z}^{3}}+186\cdot {{z}^{4}}+42\cdot {{z}^{5}}-8\cdot {{z}^{6}}\\
		+51\cdot {{z}^{7}}+35\cdot {{z}^{8}}+5\cdot {{z}^{9}}\big)\cdot \left( \frac{d}{d\,z}\cdot \mathrm{F}\left( z\right) \right)-3\cdot \left( -1-z+{{z}^{2}}\right) \cdot \big( -42-42\cdot z\\
		+18\cdot {{z}^{2}}+55\cdot {{z}^{3}}+52\cdot {{z}^{4}}+22\cdot {{z}^{5}}+3\cdot {{z}^{6}}\big)\cdot \mathrm{F}\left( z\right) =0\mbox{}
	\end{multline*}
\end{small}
\noindent
\clearpage
\begin{minipage}[t]{8ex}\color{red}\bf
	\begin{verbatim}
		(%i10) 
	\end{verbatim}
\end{minipage}
\begin{minipage}[t]{\textwidth}\color{blue}
	\begin{verbatim}
		HolonomicDE(1+z+z^2+z^3*atan(z),F(z));
	\end{verbatim}
\end{minipage}
\definecolor{labelcolor}{RGB}{100,0,0}
\begin{multline*}
	\displaystyle
	\parbox{10ex}{$\color{labelcolor}\mathrm{\tt (\%o10) }\quad $}
	z\cdot \left( 1+{{z}^{2}}\right) \cdot \left( 3+2\cdot z+4\cdot {{z}^{2}}+2\cdot {{z}^{3}}\right) \cdot \left( \frac{{{d}^{2}}}{d\,{{z}^{2}}}\cdot \mathrm{F}\left( z\right) \right) \\
	-2\cdot \left( 3+3\cdot z+8\cdot {{z}^{2}}+6\cdot {{z}^{3}}+4\cdot {{z}^{4}}+3\cdot {{z}^{5}}\right) \cdot \left( \frac{d}{d\,z}\cdot \mathrm{F}\left( z\right) \right)\\
	+6\cdot \left( 1+z+2\cdot {{z}^{2}}+{{z}^{4}}\right) \cdot \mathrm{F}\left( z\right) =0\mbox{}
\end{multline*}

Therefore we can deduce that the reason why Maple's command \textit{convert} cannot find the power series formulas of $(1+z+z^2)\exp(z)+(z^3+3)\log(1+z)$ and $1+z+z^2+z^3\arctan(z)$ is that the linear combinations of hypergeometric term solutions of the corresponding holonomic recurrence equations are missed.

As observed with the previous computations, our implementation is written in the CAS Maxima whose internal command \textit{powerseries} dedicated to power series computations is rather limited. Indeed, this command is based on a pattern matching instead of algorithmic model. The syntax is \textit{powerseries(expr,$z$,$z_0$)} that calculates the power series formula of \textit{expr} with respect to the variable $z$ at the point of development $z_0$. Below are some examples showing certain arising issues with the command \textit{powerseries} that are solved by our implementation.

\begin{itemize}
	\item Power series written as a square of a power series.
	
	\noindent
	\begin{minipage}[t]{8ex}\color{red}\bf
		\begin{verbatim}
			(%i11) 
		\end{verbatim}
	\end{minipage}
	\begin{minipage}[t]{\textwidth}\color{blue}
		\begin{verbatim}
			powerseries(asin(z)^2,z,0);
		\end{verbatim}
	\end{minipage}
	\definecolor{labelcolor}{RGB}{100,0,0}
	\[\displaystyle
	\parbox{10ex}{$\color{labelcolor}\mathrm{\tt (\%o11) }\quad $}
	{{\left( \sum_{\mathit{i1}=0}^{\infty }\frac{\mathrm{genfact}\left( 2\cdot \mathit{i1}-1,\mathit{i1},2\right) \cdot {{z}^{1+2\cdot \mathit{i1}}}}{\left( 2\cdot \mathit{i1}+1\right) \cdot \mathrm{genfact}\left( 2\cdot \mathit{i1},\mathit{i1},2\right) }\right) }^{2}}\mbox{}
	\]
	
	\noindent
	\begin{minipage}[t]{8ex}\color{red}\bf
		\begin{verbatim}
			(%i12) 
		\end{verbatim}
	\end{minipage}
	\begin{minipage}[t]{\textwidth}\color{blue}
		\begin{verbatim}
			FPS(asin(z)^2,z,n);
		\end{verbatim}
	\end{minipage}
	\definecolor{labelcolor}{RGB}{100,0,0}
	\[\displaystyle
	\parbox{10ex}{$\color{labelcolor}\mathrm{\tt (\%o12) }\quad $}
	\sum_{n=0}^{\infty }\frac{{{4}^{n}}\cdot {{n!}^{2}}\cdot {{z}^{2+2\cdot n}}}{\left( n+1\right) \cdot \left( 1+2\cdot n\right) !}\mbox{}
	\]

	\item Non-classical power series not detected.
	
	\noindent
	\begin{minipage}[t]{8ex}\color{red}\bf
		\begin{verbatim}
			(%i13) 
		\end{verbatim}
	\end{minipage}
	\begin{minipage}[t]{\textwidth}\color{blue}
		\begin{verbatim}
			powerseries((1-sqrt(1-4*z))/2,z,0);
		\end{verbatim}
	\end{minipage}
	\definecolor{labelcolor}{RGB}{100,0,0}
	\[\displaystyle
	\parbox{10ex}{$\color{labelcolor}\mathrm{\tt (\%o13) }\quad $}
	\mathrm{powerseries}\left( \frac{1-\sqrt{1-4\cdot z}}{2},z,0\right) \mbox{}
	\]
	
	\noindent
	\begin{minipage}[t]{8ex}\color{red}\bf
		\begin{verbatim}
			(%i14) 
		\end{verbatim}
	\end{minipage}
	\begin{minipage}[t]{\textwidth}\color{blue}
		\begin{verbatim}
			FPS((1-sqrt(1-4*z))/2,z,n);
		\end{verbatim}
	\end{minipage}
	\definecolor{labelcolor}{RGB}{100,0,0}
	\[\displaystyle
	\parbox{10ex}{$\color{labelcolor}\mathrm{\tt (\%o14) }\quad $}
	\sum_{n=0}^{\infty }\frac{\left( 2\cdot n\right) !\cdot {{z}^{1+n}}}{\left( n+1\right) \cdot {{n!}^{2}}}\mbox{}
	\]

	\noindent
	\begin{minipage}[t]{8ex}\color{red}\bf
		\begin{verbatim}
			(%i15) 
		\end{verbatim}
	\end{minipage}
	\begin{minipage}[t]{\textwidth}\color{blue}
		\begin{verbatim}
			powerseries(asech(z),z,0);
		\end{verbatim}
	\end{minipage}
	\definecolor{labelcolor}{RGB}{100,0,0}
	\[\displaystyle
	\parbox{10ex}{$\color{labelcolor}\mathrm{\tt (\%o15) }\quad $}
	\sum_{\mathit{i1}=0}^{\infty }\frac{{{z}^{\mathit{i1}}}\cdot \left( \left. \frac{{{d}^{\mathit{i1}}}}{d\,{{z}^{\mathit{i1}}}}\cdot \mathrm{asech}\left( z\right) \right|_{z=0}\right) }{\mathit{i1}!}\mbox{}
	\]
	
	\noindent
	\begin{minipage}[t]{8ex}\color{red}\bf
		\begin{verbatim}
			(%i16) 
		\end{verbatim}
	\end{minipage}
	\begin{minipage}[t]{\textwidth}\color{blue}
		\begin{verbatim}
			FPS(asech(z),z,n);
		\end{verbatim}
	\end{minipage}
	\definecolor{labelcolor}{RGB}{100,0,0}
	\[\displaystyle
	\parbox{10ex}{$\color{labelcolor}\mathrm{\tt (\%o16) }\quad $}
	\left( \sum_{n=0}^{\infty }-\frac{{{4}^{-1-n}}\cdot \left( 1+2\cdot n\right) !\cdot {{z}^{2+2\cdot n}}}{{{\left( 1+n\right) }^{2}}\cdot {{n!}^{2}}}\right) -\mathrm{log}\left( z\right) +\mathrm{log}\left( 2\right) \mbox{}
	\]
	
	Observe that despite the general rule used for this latter example, the output given by \textit{powerseries} is wrong since the logarithmic term $\log(z)$ does not allow the computations of derivatives at $0$.
	
	\item Power series written as multiplication of two power series.
	
	\noindent
	\begin{minipage}[t]{8ex}\color{red}\bf
		\begin{verbatim}
			(%i17) 
		\end{verbatim}
	\end{minipage}
	\begin{minipage}[t]{\textwidth}\color{blue}
		\begin{verbatim}
			powerseries(exp(z)*cos(z),z,0);
		\end{verbatim}
	\end{minipage}
	\definecolor{labelcolor}{RGB}{100,0,0}
	\[\displaystyle
	\parbox{10ex}{$\color{labelcolor}\mathrm{\tt (\%o17) }\quad $}
	\left( \sum_{\mathit{i4}=0}^{\infty }\frac{{{z}^{\mathit{i4}}}}{\mathit{i4}!}\right) \cdot \sum_{\mathit{i4}=0}^{\infty }\frac{{{\left( -1\right) }^{\mathit{i4}}}\cdot {{z}^{2\cdot \mathit{i4}}}}{\left( 2\cdot \mathit{i4}\right) !}\mbox{}
	\]
	
	\noindent
	\begin{minipage}[t]{8ex}\color{red}\bf
		\begin{verbatim}
			(%i18) 
		\end{verbatim}
	\end{minipage}
	\begin{minipage}[t]{\textwidth}\color{blue}
		\begin{verbatim}
			FPS(exp(z)*cos(z),z,n);
		\end{verbatim}
	\end{minipage}
	\definecolor{labelcolor}{RGB}{100,0,0}
	\begin{multline*}
		\displaystyle
		\parbox{10ex}{$\color{labelcolor}\mathrm{\tt (\%o18) }\quad $}
		\left( \sum_{n=0}^{\infty }-\frac{{{\left( -1\right) }^{n}}\cdot {{4}^{n}}\cdot {{z}^{3+4\cdot n}}}{{{\left( \frac{1}{4}\right) }_{n}}\cdot {{\left( \frac{3}{4}\right) }_{n}}\cdot \left( 2\cdot n+1\right) \cdot \left( 4\cdot n+1\right) \cdot \left( 4\cdot n+3\right) \cdot {{64}^{n}}\cdot \left( 2\cdot n\right) !}\right)\\
		+\left( \sum_{n=0}^{\infty }\frac{{{\left( -1\right) }^{n}}\cdot {{4}^{n}}\cdot {{z}^{1+4\cdot n}}}{{{\left( \frac{1}{4}\right) }_{n}}\cdot {{\left( \frac{3}{4}\right) }_{n}}\cdot \left( 4\cdot n+1\right) \cdot {{64}^{n}}\cdot \left( 2\cdot n\right) !}\right)\\
		+\sum_{n=0}^{\infty }\frac{{{\left( -1\right) }^{n}}\cdot {{4}^{n}}\cdot {{z}^{4\cdot n}}}{{{\left( \frac{1}{4}\right) }_{n}}\cdot {{\left( \frac{3}{4}\right) }_{n}}\cdot {{64}^{n}}\cdot \left( 2\cdot n\right) !}\mbox{}
	\end{multline*}
	
	\item A bug due to the involvement of complex numbers in the expansion.
	
	\noindent
	\begin{minipage}[t]{8ex}\color{red}\bf
		\begin{verbatim}
			(%i19) 
		\end{verbatim}
	\end{minipage}
	\begin{minipage}[t]{\textwidth}\color{blue}
		\begin{verbatim}
			powerseries(log(1+z+z^2),z,0);
		\end{verbatim}
	\end{minipage}
	\definecolor{labelcolor}{RGB}{100,0,0}
	
	sign: argument cannot be imaginary; found \%\text{i}
	
	\textcolor{red}{-- an error. To debug this try: debugmode(true);}
	
	\noindent
	\begin{minipage}[t]{8ex}\color{red}\bf
		\begin{verbatim}
			(%i20) 
		\end{verbatim}
	\end{minipage}
	\begin{minipage}[t]{\textwidth}\color{blue}
		\begin{verbatim}
			FPS(log(1+z+z^2),z,n);
		\end{verbatim}
	\end{minipage}
	\definecolor{labelcolor}{RGB}{100,0,0}
	\[\displaystyle
	\parbox{10ex}{$\color{labelcolor}\mathrm{\tt (\%o20) }\quad $}
	\sum_{n=0}^{\infty }-\frac{2\cdot \mathrm{cos}\left( \frac{2\cdot \pi \cdot \left( 1+n\right) }{3}\right) \cdot {{z}^{1+n}}}{n+1}\mbox{}
	\]
	
	In this example the general coefficient is deduced as the real part of a hypergeometric term solution in $\mathbb{C}$ (extension field of $\mathbb{Q}$ involving $i$ and some irrational numbers) of the corresponding linear recurrence equation.
\end{itemize}

On the other hand, some expressions like $\tan(z), \sec(z), \csc(z), etc.$ do not lead to linear recurrence equations, although they are analytic in certain domains. Therefore, we should investigate their power series computation. For that purpose, in this paper we consider two approaches.

Our second approach is to follow the same procedure as Koepf, but this time, instead of looking for a linear differential equation, we look for quadratic ones. For example, for the tangent function, one can find the homogeneous differential equation

\begin{equation}
	\frac{{{d}^{2}}}{d\,{{z}^{2}}}\cdot \mathrm{F}\left( z\right) -2\cdot \mathrm{F}\left( z\right) \cdot \left( \frac{d}{d\,z}\cdot \mathrm{F}\left( z\right) \right) =0, \label{ref7}
\end{equation}

\noindent which after the use of the Cauchy product rule, will lead to the recurrence equation

\begin{equation}
	\left( 1+n\right) \cdot \left( 2+n\right) \cdot {{a}_{n+2}}-2\cdot \sum_{\mathit{k}=0}^{n}\left( \mathit{k}+1\right) \cdot {{a}_{\mathit{k}+1}}\cdot {{a}_{n-\mathit{k}}}=0 \label{ref8}
\end{equation}

\noindent for the corresponding Taylor coefficients. And finally, the power series can be given by a formula depending on two initial values. 

Note, however, that this approach often gives rather complicated outputs. For example in the $\sec(z)$ case, we will find the recurrence equation 
\begin{multline}
	-\sum_{k=0}^{n}\left( \left( 2-2\cdot {{k}^{2}}\right) \cdot {{a}_{k+1}}+\left( 2\cdot k+2\right) \cdot {{a}_{k+1}}\cdot n\right) \cdot {{a}_{n-k+1}}\\
	+\left( {{a}_{k}}+\left( -{{k}^{2}}-3\cdot k-2\right) \cdot {{a}_{k+2}}\right) \cdot {{a}_{n-k}}=0
\end{multline}

The best thing to do would definitely be to "solve" the recurrence equation, but despite the fact that solutions can still be unpractical for computing power series, we intend to algorithmically find simple recursive formulas for the general coefficient. Observe that the formulas

\begin{equation}
	\tan(z) =  \sum_{n \mathop = 1}^\infty \frac {\left({-1}\right)^{n - 1} 2^{2 n} \left({2^{2 n} - 1}\right) B_{2 n} \, z^{2 n - 1} } {\left({2 n}\right)!},\label{ref9}
\end{equation} 
\begin{equation}
	\sec(z) = \sum_{n=0}^{\infty}\frac{{{\left( -1\right) }^{\mathit{n}}} E_{2 n} {{z}^{2 \mathit{n}}}}{\left( 2 \mathit{n}\right) !}\label{ref10},
\end{equation}

\noindent are not explicit because of the unknowns $B_n$ and $E_n$  which represent, respectively, Bernoulli and Euler numbers. Those numbers themselves satisfy rather complicated non-holonomic recurrence equations.

In our third approach, we extend our algorithm of hypergeometric type series. Here we consider reciprocals of formal power series and build an algorithm which can compute reciprocals of power series of some analytic expressions. Using Cauchy's product rule, some other power series are also deduced. This is only available in our Maxima package.

\noindent
\begin{minipage}[t]{8ex}\color{red}\bf
	\begin{verbatim}
		(%i21) 
	\end{verbatim}
\end{minipage}
\begin{minipage}[t]{\textwidth}\color{blue}
	\begin{verbatim}
		FPS(tan(z),z,n);
	\end{verbatim}
\end{minipage}
\definecolor{labelcolor}{RGB}{100,0,0}
\[\displaystyle
\parbox{10ex}{$\color{labelcolor}\mathrm{\tt (\%o21) }\quad $}
\left[\sum_{n=0}^{\infty }\left( \sum_{k=0}^{n}\frac{{{A}_{k}}\cdot {{\left( -1\right) }^{n-k}}}{\left( 1-2\cdot k+2\cdot n\right) !}\right) \cdot {{z}^{1+2\cdot n}},{{A}_{k}}=\sum_{j=1}^{k}-\frac{{{\left( -1\right) }^{j}}\cdot {{A}_{k-j}}}{\left( 2\cdot j\right) !},{{A}_{0}}=1\right]\mbox{}
\]

\noindent
\begin{minipage}[t]{8ex}\color{red}\bf
	\begin{verbatim}
		(%i22) 
	\end{verbatim}
\end{minipage}
\begin{minipage}[t]{\textwidth}\color{blue}
	\begin{verbatim}
		FPS(sec(z),z,n);
	\end{verbatim}
\end{minipage}
\definecolor{labelcolor}{RGB}{100,0,0}
\[\displaystyle
\parbox{10ex}{$\color{labelcolor}\mathrm{\tt (\%o22) }\quad $}
\left[\sum_{n=0}^{\infty }{{A}_{n}}\cdot {{z}^{2\cdot n}},{{A}_{n}}=\sum_{k=1}^{n}-\frac{{{\left( -1\right) }^{k}}\cdot {{A}_{n-k}}}{\left( 2\cdot k\right) !},{{A}_{0}}=1\right]\mbox{}
\]

Furthermore, besides our main results, there are some other interesting ones involved in this work. Indeed, we have got some improvement towards the decision making on the equality of two analytic functions in a certain neighborhood, and the importance of such a study is well-known in computer algebra (\cite{petkovvsek1996b}). Using our approach based on finding quadratic differential equations to represent the power series of non-holonomic functions, we are able to automatically prove identities like
\begin{equation}
	\dfrac{1+\tan(z)}{1-\tan(z)} = \exp\left( 2\cdot \arctanh\left(\dfrac{\sin(2z)}{1+\cos(2z)}\right)\right),~ |z|<1, \label{identity1}
\end{equation}
\noindent which cannot be recognized without using non-trivial transformations (see \cite[Chapter 9]{koepf2006computeralgebra}). With our Maxima implementations, computing quadratic differential equations for both sides yields two compatible\footnote{Two differential equations are said to be compatible if every solution of the lower order DE is solution of the other.} differential equations as shown below.

\noindent
\begin{minipage}[t]{8ex}\color{red}\bf
	\begin{verbatim}
		(%i23) 
	\end{verbatim}
\end{minipage}
\begin{minipage}[t]{\textwidth}\color{blue}
	\begin{verbatim}
		DE1:QDE((1+tan(z))/(1-tan(z)),F(z),Inhomogeneous);
	\end{verbatim}
\end{minipage}
\definecolor{labelcolor}{RGB}{100,0,0}
\[\displaystyle
\parbox{10ex}{$\color{labelcolor}\mathrm{\tt (\%o23) }\quad $}
\frac{d}{d\,z}\cdot \mathrm{F}\left( z\right) -{{\mathrm{F}\left( z\right) }^{2}}-1=0\mbox{}
\]

\noindent
\begin{minipage}[t]{8ex}\color{red}\bf
	\begin{verbatim}
		(%i24) 
	\end{verbatim}
\end{minipage}
\begin{minipage}[t]{\textwidth}\color{blue}
	\begin{verbatim}
		DE2:QDE(exp(2*atanh(sin(2*z)/(1+cos(2*z)))),F(z));
	\end{verbatim}
\end{minipage}
\definecolor{labelcolor}{RGB}{100,0,0}
\[\displaystyle
\parbox{10ex}{$\color{labelcolor}\mathrm{\tt (\%o24) }\quad $}
\mathrm{F}\left( z\right) \cdot \left( \frac{{{d}^{3}}}{d\,{{z}^{3}}}\cdot \mathrm{F}\left( z\right) \right) -3\cdot \left( \frac{d}{d\,z}\cdot \mathrm{F}\left( z\right) \right) \cdot \left( \frac{{{d}^{2}}}{d\,{{z}^{2}}}\cdot \mathrm{F}\left( z\right) \right) +4\cdot \mathrm{F}\left( z\right) \cdot \left( \frac{d}{d\,z}\cdot \mathrm{F}\left( z\right) \right) =0\mbox{}
\]

\noindent
\begin{minipage}[t]{8ex}\color{red}\bf
	\begin{verbatim}
		(%i25) 
	\end{verbatim}
\end{minipage}
\begin{minipage}[t]{\textwidth}\color{blue}
	\begin{verbatim}
		CompatibleDE(DE1,DE2,F(z));
	\end{verbatim}
\end{minipage}
\definecolor{labelcolor}{RGB}{100,0,0}
\[\displaystyle \hspace{-1cm}\textit{The two differential equations are compatible}\]
\[\displaystyle
\parbox{10ex}{$\color{labelcolor}\mathrm{\tt (\%o25) }\quad $}
\textbf{true}\mbox{}
\]

Moreover, our FPS algorithm simplifies the difference to zero in a neighborhood of $0$. 

\noindent
\begin{minipage}[t]{8ex}\color{red}\bf
	\begin{verbatim}
		(%i26) 
	\end{verbatim}
\end{minipage}
\begin{minipage}[t]{\textwidth}\color{blue}
	\begin{verbatim}
		FPS((1+tan(z))/(1-tan(z))
		-exp(2*atanh(sin(2*z)/(1+cos(2*z)))),z,n);
	\end{verbatim}
\end{minipage}
\definecolor{labelcolor}{RGB}{100,0,0}
\[\displaystyle
\parbox{10ex}{$\color{labelcolor}\mathrm{\tt (\%o26) }\quad $}
0\mbox{}
\]
Things are a little simpler with Maple, but this is because Maple automatically applies some simplification on the right-hand side so that the same differential equation is found.
\begin{maplegroup}
	\begin{mapleinput}
		\mapleinline{active}{1d}{FPS[QDE](exp(2*arctanh(sin(2*z)/(1+cos(2*z)))),F(z))
		}{}
	\end{mapleinput}
	\mapleresult
	\begin{maplelatex}
		\[\displaystyle -2\, \left( {\frac {\rm d}{{\rm d}z}}F \left( z \right)  \right) F \left( z \right) +{\frac {{\rm d}^{2}}{{\rm d}{z}^{2}}}F \left( z \right) =0\]
	\end{maplelatex}
\end{maplegroup}

We have also obtained an algorithm for asymptotically fast computation of Taylor expansions of large order for holonomic functions. This is a result already observed in \cite[Section 10.27]{koepf2006computeralgebra}. We have implemented a Maxima function named \textit{Taylor} with the same syntax as \textit{taylor(f,z,$z_0$,d)} which computes the Taylor expansion of order $d$ of $f(z)$. And it turns out as expected that our \textit{Taylor} command is clearly asymptotically faster than \textit{taylor} for holonomic functions. As an example we have:

\noindent
\begin{minipage}[t]{8ex}\color{red}\bf
	\begin{verbatim}
		(%i27) 
	\end{verbatim}
\end{minipage}
\begin{minipage}[t]{\textwidth}\color{blue}
	\begin{verbatim}
		taylor(sin(z)^2,z,0,10);
	\end{verbatim}
\end{minipage}
\definecolor{labelcolor}{RGB}{100,0,0}
\[\displaystyle
\parbox{10ex}{$\color{labelcolor}\mathrm{\tt (\%o27)/T/ }\quad $}
{{z}^{2}}-\frac{{{z}^{4}}}{3}+\frac{2\cdot {{z}^{6}}}{45}-\frac{{{z}^{8}}}{315}+\frac{2\cdot {{z}^{10}}}{14175}+\cdots\mbox{}
\]

\noindent
\begin{minipage}[t]{8ex}\color{red}\bf
	\begin{verbatim}
		(%i28) 
	\end{verbatim}
\end{minipage}
\begin{minipage}[t]{\textwidth}\color{blue}
	\begin{verbatim}
		Taylor(sin(z)^2,z,0,10);
	\end{verbatim}
\end{minipage}
\definecolor{labelcolor}{RGB}{100,0,0}
\[\displaystyle
\parbox{10ex}{$\color{labelcolor}\mathrm{\tt (\%o28) }\quad $}
\frac{2\cdot {{z}^{10}}}{14175}-\frac{{{z}^{8}}}{315}+\frac{2\cdot {{z}^{6}}}{45}-\frac{{{z}^{4}}}{3}+{{z}^{2}}\mbox{}
\]
that illustrates the coincidence between both outputs. Testing the efficiency for large order gives:

\noindent
\begin{minipage}[t]{8ex}\color{red}\bf
	\begin{verbatim}
		(%i29) 
	\end{verbatim}
\end{minipage}
\begin{minipage}[t]{\textwidth}\color{blue}
	\begin{verbatim}
		taylor(sin(z)^2,z,0,1000)$
	\end{verbatim}
\end{minipage}
\definecolor{labelcolor}{RGB}{100,0,0}
\mbox{}\\\mbox{Evaluation took 15.8500 seconds (19.6100 elapsed)}

\noindent
\begin{minipage}[t]{8ex}\color{red}\bf
	\begin{verbatim}
		(%i30) 
	\end{verbatim}
\end{minipage}
\begin{minipage}[t]{\textwidth}\color{blue}
	\begin{verbatim}
		Taylor(sin(z)^2,z,0,1000)$
	\end{verbatim}
\end{minipage}
\definecolor{labelcolor}{RGB}{100,0,0}
\mbox{}\\\mbox{Evaluation took 1.8300 seconds (1.8900 elapsed)}

\noindent which shows that, asymptotically, our \textit{Taylor} command takes just about a fraction of Maxima's internal \textit{taylor} computation timing for $\sin(z)^2$.

As already observed, most of our computations will be presented with the CAS Maxima which was the main system used when our results were being developed (see \cite{BTphd}). The further sections are organized as follows.

The next section recalls the two first steps in Koepf's algorithm: computing holonomic differential equations and holonomic recurrence equations. In this section, we add some linear algebra tricks in order to gain more efficiency in the process of getting holonomic differential equations for hypergeometric type functions. This section ends with the description of our asymptotically fast algorithm for computing Taylor expansions of holonomic functions.

Section \ref{sec2} is devoted to our most important result, which is to present a complete algorithm to find all $m$-fold hypergeometric term solutions of linear recurrence equations with polynomial coefficients.

In Section \ref{sec3} we complete Koepf's algorithm with mfoldHyper. We will see in this section how our algorithm handles the starting points, the polynomial part, and the Puiseux numbers involved in a given hypergeometric type series expansion. 

Section \ref{sec4} describes our approach based on the computation of quadratic differential equations for representing non-holonomic power series. This part is another important contribution of our work. 

%Further accompanying results are presented in Section \ref{sec5}. We mention a deal with asymptotic series which is unfortunately limited by the capabilities of Maxima in computing limits. Nevertheless, some known examples are well computed. We also try to deduce a lazy approach for series representations of holonomic expressions that are not of hypergeometric type. This is the situation for some reciprocals or products of hypergeometric type power series.

\section{Computing holonomic equations}\label{sec1}

This section is about computing holonomic differential equations from given expressions and use them to deduce holonomic recurrence equations of their Taylor or power series coefficients prior to computing $m$-fold hypergeometric term solutions or Taylor expansions. The first interest of this section is our strategy to implement the given algorithm in (\cite{Koepf1992}) for computing holonomic differential equations, which is slightly more efficient than its original version. Moreover, we give an algorithm to check whether two holonomic differential equations are compatible. Furthermore, we propose a formal algorithm that performs fast computation of larger order Taylor expansions for holonomic functions as discussed in \cite[Section 10.27]{koepf2006computeralgebra}. The algorithm to deduce the recurrence equations from holonomic differential equations will be given just because it contains a rewrite rule that plays an asset role in Section \ref{sec4}.

We recall Koepf's original approach from an introductory example. Let $f(z)=\exp(z)+\cos(z)$. We want to find a holonomic differential equation (DE) with coefficients in $\mathbb{Q}[z]$ satisfied by $f(z)$ and deduce a holonomic recurrence equation (RE) with coefficients in $\mathbb{Q}[n]$ satisfied by the Taylor coefficients $a_n$ of $f$.

\textbf{Search of a holonomic DE:} $f'(z)=\exp(z)-\sin(z)$, and therefore there is no $A_0(z)\in \mathbb{Q}(z)$ such that $f'(z)+A_0(z)f(z)=0$ because $A_0(z)$ should be $-(\exp(z)-\sin(z))/(\exp(z)+\cos(z))$ which is not rational. Therefore we move to the second order.
We search for $A_0(z),A_1(z)\in\mathbb{Q}(z)$ such that
\begin{equation*}
	f''(z)+A_1(z)f'(z)+A_0(z)f(z)=0.
\end{equation*}
Note that at this step $f'(z)$ is computed again. We write the sum in terms of linearly independent parts and we obtain
\begin{equation*}
	\left(1+A_0(z)+A_1(z)\right) \exp(z) + \left(A_0(z)-1\right)\cos(z)-A_1(z)\sin(z)=0,
\end{equation*}
and we get the linear system
\begin{equation*}
	\begin{cases}
		A_0(z)-1=0\\
		A_1(z)=0\\
		A_0(z)+A_1(z)+1=0
	\end{cases},
\end{equation*}
which has no solution. However, for the third order, the relation 
\begin{equation*}
	f^{(3)}(z)+A_2(z)f^{(2)}(z)+A_1(z)f^{(1)}(z)+A_0(z)f(z)=0,
\end{equation*}
with $f^{(3)}(z)=\exp(z)+\sin(z)$, leads to a solvable system. By writing the sum in terms of linearly independent parts, we get
\begin{equation}
	\left(1+A_0(z)+A_1(z)+A_2(z)\right)\exp(z)+\left(A_0(z)-A_2(z)\right)\cos(z)+\left(1-A_1(z)\right)\sin(z)=0, \label{holoFirstcase1}
\end{equation}
so 
\begin{equation}
	\begin{cases}1-A_1(z)=0\\A_0(z)-A_2(z)=0\\1+A_0(z)+A_1(z)+A_2(z)=0\end{cases} \iff \begin{cases}A_0(z)=A_2(z)=-1\\A_1(z)=1\end{cases}\label{coeffholo}.
\end{equation} 
Hence the corresponding holonomic DE 
\begin{equation}
	\frac{{{d}^{3}}}{d\,{{z}^{3}}}\cdot \mathrm{F}\left( z\right) -\frac{{{d}^{2}}}{d\,{{z}^{2}}}\cdot \mathrm{F}\left( z\right) +\frac{d}{d\,z}\cdot \mathrm{F}\left( z\right) -\mathrm{F}\left( z\right) =0. \label{holoFirstcase}
\end{equation}
for $f(z)$ is valid.

Notice that if some of the coefficients found in $(\ref{coeffholo})$ did have polynomials different from $1$ as their denominators, then a further step would be the multiplication of the resulting holonomic DE with the least common multiple of the denominators.

\textbf{Conversion of $(\ref{holoFirstcase})$ into its corresponding RE:}  We set 
\begin{equation*}
	f(z)=\sum_{n=0}^{\infty}a_nz^n,
\end{equation*}
so
\begin{align}
	&f'(z)=\sum_{n=1}^{\infty}na_nz^{n-1}=\sum_{n=0}^{\infty}(n+1)a_{n+1}z^n,\\
	&f''(z)=\sum_{n=2}^{\infty}n(n-1)a_nz^{n-2}=\sum_{n=0}^{\infty}(n+2)(n+1)a_{n+2}z^n,\\
	&f^{(3)}(z)= \sum_{n=3}^{\infty}n(n-1)(n-2)a_nz^{n-3}=\sum_{n=0}^{\infty}(n+3)(n+2)(n+1)a_{n+3}z^n.
\end{align}
By substitution of these identities in $(\ref{holoFirstcase})$ for $f$, we get
\begin{eqnarray}
	0 &=& f^{(3)}(z)-f''(z)+f'(z)+f(z) \nonumber\\
	&=& \sum_{n=0}^{\infty}(n+3)(n+2)(n+1)a_{n+3}z^n-\sum_{n=0}^{\infty}(n+2)(n+1)a_{n+2}z^n \nonumber\\
	&\phantom{=}& \phantom{\sum_{n=0}^{\infty}(n+3)(n+2)(n+1)a_{n+3}z^n}+\sum_{n=0}^{\infty}(n+1)a_{n+1}z^n-\sum_{n=0}^{\infty}a_nz^n\nonumber\\
	&=& \sum_{n=0}^{\infty}\left[(n+3)(n+2)(n+1)a_{n+3}-(n+2)(n+1)a_{n+2}+(n+1)a_{n+1}-a_n\right]z^n,\nonumber
\end{eqnarray}
hence by equating the coefficients we find the holonomic RE
\begin{equation}
	(n+3)(n+2)(n+1)a_{n+3}-(n+2)(n+1)a_{n+2}+(n+1)a_{n+1}-a_n = 0,~n=0,1,2,\ldots~. \label{firsexRE}
\end{equation}
for $a_n$. For more details about this procedure see (\cite{Koepf1992,SC93}).

\subsection{Computing holonomic differential equations}\label{mymethodDE}

Note that the algorithm we consider here often finds the holonomic DE of lowest order. In this paragraph we look at the algorithm in a different way. Indeed, if we consider $f(z)=\cos(z)+\sin(z)$, it is clear that the differential equation that the algorithm will find is a null linear combination of the derivatives of $f(z)$ expanded in the basis $\left(\cos(z),\sin(z)\right)$. Thus, we can save the time spent by computing all the derivatives at each iteration of the original algorithm, by trying to write each derivative in the same basis. Therefore the computations of the original algorithm that are done for all the derivatives at each iteration are reduced to only one derivative.

\subsubsection{Idea of the method}
$\mathbb{K}$ denotes a field of characteristic zero. Let $\left(A_0,A_1,\ldots,A_{N-1}\right)\in\mathbb{K}(z)^N,$ $N\in \mathbb{N}$ such that an analytic expression $f$ satisfies
\begin{equation}
	\mathcal{F}\left(f,f',\ldots,f^{(N-1)},f^{(N)}\right) = f^{(N)} + A_{N-1}\cdot f^{(N-1)} + \cdots + A_1\cdot f + A_0 f = 0.
\end{equation}

We consider a basis $(e_1,e_2,\ldots,e_l)$ of the linear span of all linearly independent summands  over $\mathbb{K}(z)$ that appear in the complete expansions of the derivatives $f,f',\ldots, f^{(N)}$. For example, assume for $0\leqslant i\neq j \leqslant N$, that
\begin{eqnarray}
	f^{(i)} &=& e_{i,1} + \cdots + e_{i,k_i},\nonumber\\
	f^{(j)} &=& e_{j,1} + \cdots + e_{j,k_j},\nonumber
\end{eqnarray}
for some positive integers $k_i$ and $k_j$, such that $e_{i,u}/e_{i,v}\notin \mathbb{K}(z)$ for all $u,v\in\llbracket1,k_i\rrbracket$ \footnote{For $l,k\in\mathbb{Z}, l<k$ we define $\llbracket l, k \rrbracket := \{l, l+1,\ldots, k\}$} and $e_{j,u}/e_{j,v}\notin \mathbb{K}(z)$ for all $u,v\in\llbracket1,k_j\rrbracket$. Then for $f^{(i)}$ and $f^{(j)}$ we consider a basis of the linear span of $\{e_{i,1},\ldots,e_{i,k_i},e_{j,1},\ldots,e_{j,k_j}\}$ which may have less elements since some $e_{i,u},$ $u\in\llbracket1,k_i\rrbracket$ and $e_{j,v},$ $v\in\llbracket1,k_j\rrbracket$ can be linearly dependent.

Thus each derivative $f^{(j)}, j\in\mathbb{N}_{\geqslant 0}$ $\left(f^{(0)}=f\right)$ can be seen as a vector in the linear space $\langle e_1,e_2,\ldots,e_l\rangle$.

\noindent Since
\begin{equation}
	\mathcal{F}\left(f,f',\ldots,f^{(N-1)},f^{(N)}\right)=0 \iff -f^{(N)} = A_0\cdot f + A_1\cdot f' + \cdots + A_{N-1}\cdot f^{(N-1)},
\end{equation}
we can write in a matrix representation
\begin{equation}
	-f^{(N)} = \left[f,f',\ldots, f^{(N-1)}\right]_{\left(e_1,e_2,\ldots,e_l\right)} \left(A_0,A_1,\ldots,A_{N-1}\right)^T.
\end{equation}

Therefore, one sees that seeking for a holonomic DE of order $N$ satisfied by a given expression $f(z)$ is equivalent to find a basis in a $\mathbb{K}(z)$-linear space where the system
$$\left(f^{(N)}(z),f^{(N-1)}(z),\ldots,f'(z),f(z)\right)$$ 
is linearly dependent. The idea of the method described in this section is to construct such a basis while computing each derivative of $f(z)$ and their components. Thus, in each iteration $N$, if all the $N+1$ derivatives are expanded in the same basis, then we try to solve the resulting linear system.

\subsubsection{Description of the method}
Now let us present the algorithm in general. Consider an expression $f(z)$ which is not identically zero with $l_0$ linearly independent sub-terms over $\mathbb{K}(z)$. Then we can write
\begin{equation}
	f(z) = f_1(z) + f_2(z) + \cdots + f_{l_0}(z) \label{lind}
\end{equation}	
with $f_i(z)/f_j(z) \notin \mathbb{K}(z),~1\leqslant i\neq j\leqslant l_0$. $f(z)$ is seen as a vector in the basis $E_0=\left(e_1,e_2,\ldots,e_{l_0}\right)$ where $e_i=f_i$. Then we compute the first derivative of $f(z)$, and we get the following two possibilities:
\begin{itemize}
	\item either $f'(z)$ is expressed in $E_0$, which means that there exist $\alpha_{1,i} = \alpha_{1,i}(z)\in\mathbb{K}(z), i=1,\ldots,l_0$ such that
	\begin{equation}
		f'(z) = \alpha_{1,1} e_1 + \alpha_{1,2} e_2 + \ldots + \alpha_{1,l_0} e_{l_0}. \label{fprimbasis}
	\end{equation}
	Here in the worst case, $f'(z)$ and $f(z)$ are linearly independent, but then we know that all derivatives can be expanded in $E_0$.
	\item Or $f'(z)$ is not expanded in $E_0$, which means that $E_0$ has to be augmented and there exist $\alpha_{1,i}\in\mathbb{K}(z), i=1,\ldots,l_0$ and an integer $l_1 > l_0$ such that
	\begin{equation}
		f'(z) = \alpha_{1,1} e_1 + \alpha_{1,2} e_2 + \ldots + \alpha_{1,l_0} e_{l_0}+ e_{l_0+1}+\ldots + e_{l_1}. \label{fprimbasis1}
	\end{equation}
	Observe here that the new basis is $E_1=\left(e_1,\ldots,e_{l_1}\right)$ with $e_{l_0+1},\ldots,e_{l_1}$ corresponding to independent terms brought by $f'(z)$. And also $\alpha_{1,i}$, $i\leqslant l_0$ could be zero.
\end{itemize}

Actually in the first case we may find the DE sought, but in order to present a general view of the algorithm, let us assume that $f(z)$ satisfies a DE of order $N\geqslant 1$. It follows that the process will lead to the representation
\begin{align}
	&f(z)= e_1+\ldots+e_{l_0}\\
	&f'(z)= \alpha_{1,1} e_1 + \cdots + \alpha_{1,l_0} e_{l_0}+ e_{l_0+1}+\cdots + e_{l_1}\\
	&f''(z)= \alpha_{2,1} e_1 + \cdots + \alpha_{2,l_0} e_{l_0}+ \alpha_{2,l_0+1} e_{l_0+1}+\cdots + \alpha_{2,l_1}e_{l_1}+e_{l_1+1}+\ldots+e_{l_2}\\
	&\cdots\\
	&f^{(N-1)}(z)=\alpha_{N-1,1} e_1 + \cdots + \alpha_{N-1,l_{N-2}} e_{l_{N-2}}+ e_{l_{N-2}+1}+\cdots + e_{l_{N-1}}\\
	&f^{(N)} = \alpha_{N,1} e_1 +\cdots + \alpha_{N,l_{N-1}} e_{l_{N-1}},
\end{align}
with positive integers $l_0\leqslant l_1 \leqslant \ldots \leqslant l_{N-1} $,  and $\alpha_{i,j}\in\mathbb{K}(z),~i=1,\ldots,N,~j=1,\ldots,l_{i-1}.$

Note, however, that in this step $N$ only $f^{(N)}(z)$ is computed by differentiating $f^{(N-1)}(z)$. In each step, the algorithm keeps the coefficients $\alpha_{N,i}$, the augmented basis and the current derivative.

It is straightforward to see that the final basis considered is $E_{N-1}=\left(e_1,\ldots,e_{l_{N-1}}\right).$ The algorithm keeps information in a matrix form, say $H$, and at the $N^{\text{th}}$ iteration we have

\begin{small}
	\begin{align}
		&H=
		\begin{bmatrix}
			1&\cdots&1&0&\cdots&0&0&\cdots&\cdots&0&\cdots&0\\
			\alpha_{1,1}&\cdots&\alpha_{1,l_0}&1&\cdots&1&0&\cdots&\cdots&0&\cdots&0\\
			\alpha_{2,1}&\cdots&\alpha_{2,l_0}&\cdots&\cdots&\alpha_{2,l_1}&1&\cdots&\cdots&0&\cdots&0\\
			\vdots&\vdots&\vdots&\vdots&\vdots&\vdots&\vdots&\vdots&\vdots&\vdots&\vdots&\vdots\\
			\alpha_{N-1,1}&\cdots&\cdots&\cdots&\cdots&\cdots&\cdots&\cdots&\alpha_{N-1,l_{N-2}}&1&\cdots&1\\
			\alpha_{N,1}&\cdots&\cdots&\cdots&\cdots&\cdots&\cdots&\cdots&\alpha_{N,l_{N-2}}&\alpha_{N,l_{N-2}+1}&\cdots&\alpha_{N,l_{N-1}}
		\end{bmatrix}. \nonumber\\
		\label{holoMat}
	\end{align}	
\end{small}
\noindent $H$ is a $(N+1)\times l_{N-1}$ matrix in $\mathbb{K}(z)$, and it contains all information that we need to find the holonomic DE sought. Indeed, one can easily show that the coefficients $A_i(z)\in\mathbb{K}(z),i=0,\ldots,N-1$ computed in the original approach constitute the rational components of the unique vector solution of the matrix system
\begin{equation}
	A\cdot v=b,
\end{equation}
with
\begin{align}
	&A= \begin{bmatrix}
		1     &\alpha_{1,1}   &\alpha_{2,1}    &\ldots&\alpha_{N-1,1}\\
		\vdots&\vdots         &\vdots	        &\vdots&\vdots\\
		1	   &\alpha_{1,l_0} &\alpha_{2,l_0}  &\ldots&\alpha_{N-1,l_0}\\
		0	   &1			   &\alpha_{2,l_0+1}&\ldots&\alpha_{N-1,l_0+1}\\
		\vdots&\vdots         &\vdots	        &\vdots&\vdots\\
		0     &0              &0               &\ldots&1
	\end{bmatrix}\\
	&\text{and}\nonumber\\
	&b=-\begin{bmatrix}\alpha_{N,1}\\\alpha_{N,2}\\\vdots\\\alpha_{N,l_{N-1}}\end{bmatrix}.
\end{align}		
Observe that $b$ is the negative (note the minus in front) of the transpose of the last row of $H$, and $A$ is the transpose of $H$ deprived of its last row. The above linear system has $l_{N-1}$ linear equations and $N$ unknowns. 

Let us see how this algorithm works on some examples.
\begin{example}\item
	
	\begin{itemize}
		\item $f(z)=\sin(z)+z\cos(z)$. We have two linearly independent terms over $\mathbb{Q}(z)$, and we can write
		$$f(z)= e_1 + e_2,$$
		with $e_1=\sin(z)$ and $e_2=z\cos(z)$. Computing the first derivative, we get
		$$f'(z)=-z\sin(z)+2\cos(z)=-z\cdot e_1+\dfrac{2}{z}\cdot e_2.$$
		At this step we have 
		$$H=\begin{bmatrix}
			1&1\\
			-z&\frac{2}{z}
		\end{bmatrix},$$
		and we get the system $\begin{bmatrix}1\\1\end{bmatrix}v = \begin{bmatrix}z\\-\frac{2}{z}\end{bmatrix}$, which has no solution $v\in\mathbb{Q}(z)$ (seen as a one-dimensional vector space). Now we compute the second derivative, and we get
		$$f''(z)=-3\sin(z)-z\cos(z)=-3\cdot e_1-e_2.$$
		$H$ becomes
		$$H=\begin{bmatrix}
			1&1\\
			-z&\frac{2}{z}\\
			-3&-1
		\end{bmatrix},$$
		which gives the system 
		$$
		\begin{bmatrix}
			1&-z\\
			1&\frac{2}{z}
		\end{bmatrix} v = \begin{bmatrix}3\\1\end{bmatrix},~v\in\mathbb{Q}(z)^2,
		$$
		and we get the solution
		\begin{equation}
			\left\lbrace\left(\dfrac{6+z^2}{2+z^2},\dfrac{-2z}{2+z^2}\right)\right\rbrace \label{sol1}.
		\end{equation}
		The differential equation sought is therefore
		\begin{equation}
			(2+z^2)f''(z) - 2zf'(z) + (6+z^2)f=0.
		\end{equation}
		\item $f(z)=\arctan(z)$. We have only one term so $e_1=\arctan(z)$. For the first derivative
		$$f'(z)=\dfrac{1}{1+z^2}=0\cdot e_1+e_2,$$
		where $e_2=1/(1+z^2)$. Since the basis has been augmented there is no system to be solved, and at this step we have
		$$
		H=\begin{bmatrix}
			1&0\\0&1
		\end{bmatrix}.
		$$
		The second derivative gives
		$$f''(z)=-\dfrac{2z}{(1+z^2)^2}=0\cdot e_1-\dfrac{2z}{1+z^2}\cdot e_2,$$
		and we get
		$$H=\begin{bmatrix}
			1&0\\0&1\\0&-\frac{2z}{1+z^2}
		\end{bmatrix}$$	
		which produces the system 
		$$
		\begin{bmatrix}
			1&0\\0&1
		\end{bmatrix} v = \begin{bmatrix}0\\\frac{2z}{1+z^2}\end{bmatrix},~v\in\mathbb{Q}(z)^2.
		$$
		We get $v=\left(0,2z/(1+z^2)\right)$, hence the holonomic DE
		\begin{equation}
			(1+z^2)f''(z) + 2zf'(z)=0.
		\end{equation}
		\item $f(z)=\exp(z)+\log(1+z)=e_1 + e_2$, with $e_1=\exp(z)$ and $e_2=\log(1+z)$. The first derivative yields
		$$f'(z)=\exp(z)+\dfrac{1}{1+z}=e_1+0\cdot e_2 + e_3,$$
		with $e_3=1/(1+z)$. Since a new term is added to the basis, the next step is to compute the second derivative
		$$f''(z) = \exp(z)-\dfrac{1}{(1+z)^2}=e_1 + 0\cdot e_2 -\dfrac{1}{(1+z)}\cdot e_3.$$
		No term is added to the basis. We try to solve the resulting system. At this stage
		$$H=\begin{bmatrix}
			1&1&0\\1&0&1\\1&0&-\frac{1}{1+z}
		\end{bmatrix},$$
		and we get the system
		$$
		\begin{bmatrix}
			1&1\\1&0\\0&1
		\end{bmatrix} v = \begin{bmatrix}-1\\0\\\frac{1}{1+z}\end{bmatrix},~v\in\mathbb{Q}(z)^2,
		$$
		which has no solution. We move on and compute the third derivative
		$$f^{(3)}(z)=\exp(z)+\dfrac{2}{(1+z)^3}=e_1+0\cdot e_2 + \dfrac{2}{(1+z)^2}\cdot e_3.$$
		Thus
		$$H=\begin{bmatrix}
			1&1&0\\1&0&1\\1&0&-\frac{1}{1+z}\\1&0&\frac{2}{(1+z)^2}
		\end{bmatrix},$$
		and we obtain the system
		$$
		\begin{bmatrix}
			1&1&1\\1&0&0\\0&1&-\frac{1}{1+z}
		\end{bmatrix} v = \begin{bmatrix}-1\\0\\-\frac{2}{(1+z)^2}\end{bmatrix},
		$$
		whose solution in $\mathbb{Q}(z)^3$ is
		\begin{equation}
			\left\lbrace\left(0,-\dfrac{3+z}{(1+z)(2+z)},-\dfrac{-1+2z+z^2}{(1+z)(2+z)}\right)\right\rbrace \label{sol2}.
		\end{equation}
		Therefore we get the holonomic DE
		\begin{equation}
			(1+z)(2+z)f^{(3)}(z) - (-1+2z+z^2) f^{''}(z) - (3+z)f'(z)=0.
		\end{equation}
	\end{itemize}
\end{example}

We implemented this algorithm in Maple and Maxima as \textit{HolonomicDE(f,F(z))} to compute a holonomic DE with the indeterminate \textit{F(z)} for an expression $f$ of the variable $z$. In Maxima the package contains a global variable \textit{Nmax} which can be changed in order to look for higher order differential equations. Here are some examples.

\noindent
\begin{minipage}[t]{8ex}\color{red}\bf
	\begin{verbatim}
		(%i1) 
	\end{verbatim}
\end{minipage}
\begin{minipage}[t]{\textwidth}\color{blue}
	\begin{verbatim}
		HolonomicDE(asin(z),F(z));
	\end{verbatim}
\end{minipage}
\definecolor{labelcolor}{RGB}{100,0,0}
\[\displaystyle
\parbox{10ex}{$\color{labelcolor}\mathrm{\tt (\%o1) }\quad $}
\left( z-1\right) \cdot \left( 1+z\right) \cdot \left( \frac{{{d}^{2}}}{d\,{{z}^{2}}}\cdot \mathrm{F}\left( z\right) \right) +z\cdot \left( \frac{d}{d\,z}\cdot \mathrm{F}\left( z\right) \right) =0\mbox{}
\]

\noindent
\begin{minipage}[t]{8ex}\color{red}\bf
	\begin{verbatim}
		(%i2) 
	\end{verbatim}
\end{minipage}
\begin{minipage}[t]{\textwidth}\color{blue}
	\begin{verbatim}
		HolonomicDE(cos(z)*log(1+z),F(z));
	\end{verbatim}
\end{minipage}
\definecolor{labelcolor}{RGB}{100,0,0}
\begin{multline*}
	\parbox{10ex}{$\color{labelcolor}\mathrm{\tt (\%o2) }\quad $}
	{{\left( 1+z\right) }^{2}}\cdot \left( 1+2\cdot z\right) \cdot \left( 3+2\cdot z\right) \cdot \left( \frac{{{d}^{4}}}{d\,{{z}^{4}}}\cdot \mathrm{F}\left( z\right) \right) \\
	+4\cdot \left( 1+z\right) \cdot \left( 1+4\cdot z+2\cdot {{z}^{2}}\right) \cdot \left( \frac{{{d}^{3}}}{d\,{{z}^{3}}}\cdot \mathrm{F}\left( z\right) \right)+2\cdot z\cdot \left( 2+z\right) \cdot \left( 5+8\cdot z+4\cdot {{z}^{2}}\right) \cdot \left( \frac{{{d}^{2}}}{d\,{{z}^{2}}}\cdot \mathrm{F}\left( z\right) \right) \\+4\cdot \left( 1+z\right) \cdot \left( 1+4\cdot z+2\cdot {{z}^{2}}\right) \cdot \left( \frac{d}{d\,z}\cdot \mathrm{F}\left( z\right) \right) +\left( -3+6\cdot z+19\cdot {{z}^{2}}+16\cdot {{z}^{3}}+4\cdot {{z}^{4}}\right) \cdot \mathrm{F}\left( z\right) =0\mbox{}
\end{multline*}

\noindent
\begin{minipage}[t]{8ex}\color{red}\bf
	\begin{verbatim}
		(%i3) 
	\end{verbatim}
\end{minipage}
\begin{minipage}[t]{\textwidth}\color{blue}
	\begin{verbatim}
		HolonomicDE(sin(z)^4*asin(z),F(z));
	\end{verbatim}
\end{minipage}
\definecolor{labelcolor}{RGB}{100,0,0}
\[\displaystyle
\parbox{10ex}{$\color{labelcolor}\mathrm{\tt (\%o7) }\quad $}
\text{false}\mbox{}
\]

\noindent
\begin{minipage}[t]{8ex}\color{red}\bf
	\begin{verbatim}
		(%i4) 
	\end{verbatim}
\end{minipage}
\begin{minipage}[t]{\textwidth}\color{blue}
	\begin{verbatim}
		Nmax:10$
	\end{verbatim}
\end{minipage}
\noindent
\begin{minipage}[t]{8ex}\color{red}\bf
	\begin{verbatim}
		(%i5) 
	\end{verbatim}
\end{minipage}
\begin{minipage}[t]{\textwidth}\color{blue}
	\begin{verbatim}
		HolonomicDE(sin(z)^4*asin(z),F(z))$
	\end{verbatim}
\end{minipage}
\definecolor{labelcolor}{RGB}{100,0,0}
\[\displaystyle \text{Evaluation took 1.7500 seconds (1.7570 elapsed) using 858.342 MB.}\mbox{}
\]

The latter is a big differential equation of order $10$ (default value) $>$ \textit{Nmax}, that is why the value of \textit{Nmax} was changed to $10$. The following table shows the efficiency gain of this method on the original one implemented in Maple under the commands \textit{DEtools[FindODE]} (Approach 2) and \textit{FormalPowerSeries[HolonomicDE]} (Approach 3). The difference between these two Maple commands is that \textit{DEtools[FindODE]} uses some simplification in order to consider some special functions. Our Maple implementation is accessible through \textit{FPS[HolonomicDE]} (Approach 1).

\begin{table}[h!]
	\caption{Comparison of timings for computing holonomic DEs}\label{compareHDE}
	\begin{center}
			\begin{tabular}{|l|l|l|l|}
			\hline
			\multirow{2}{*}{$f(z)$}    & \multicolumn{3}{l|}{CPU time}
			\\ \cline{2-4} 
			& Approach 1 & Approach 2 & Approach 3 \\ \hline
			$\sin(z)^4\arcsin(z)$                         & $0.703$   & $1.578$   & $0.734$     \\ \hline
			$\sin(z)^6 \arcsin(z)^3$                      & $116.406$ & $259.250$ & $254.046$   \\ \hline
			$\arctan(z)^2+\sin(z)^4+\log(1+z)^5$          & $1.422$   & $11.859$  & $11.594$    \\ \hline
			\begin{tabular}[c]{@{}l@{}}$\exp(z^{11}+3) \cos(z)+\log(1+z^7)$\\$+\cos(z)^3 \sinh(z)^5$\end{tabular}                      & $189.968$ & $1466.234$& $1373.515$  \\ \hline
			$\arccos(z)^{13}+\sinh(z)^{19}$               & $97.125$  & $318.563$ & $418.859$   \\ \hline
			\begin{tabular}[c]{@{}l@{}}$(3z+5z^7+11z^{13})\log(1+z^3+z^7)$\\$+\arctanh(z)\cos(z)^5$\end{tabular}                                                      & $37.171$  & $332.953$ & $365.922$   \\ \hline
		\end{tabular}
	\end{center}
\end{table}
	
As expected, we realized that our approach gives better timings for expressions whose derivatives yield many linearly independent sub-expressions. This is an important aspect for the computation of hypergeometric type power series since we will consider linear combinations of expressions. In some other cases the timings of both algorithms get closer as the order of derivatives increases because of the use of memory in our approach, but this rather rarely happens.

\subsection{Computing holonomic recurrence equations}

After expanding a holonomic differential equation, one easily establishes the following rewrite rule needed to find the corresponding recurrence equation of the underlying power series coefficients (see \cite{Koepf1992}).
\begin{equation}
		z^l\frac{d^j}{dz}f    \longrightarrow  (n+1-l)_j\cdot a_{n+j-l}.\label{Algo3corr}
\end{equation}
We used the so-called Pochhammer symbol or shifted factorial defined as $(x)_0=1$ and $(x)_n=x\cdot(x+1)\cdots(x+n-1)$ for a constant $x$ and a non-negative integer $n$. The final holonomic RE sought is obtained after collecting similar terms. To get simpler results, we finally factorize the coefficients.

Our Maxima package contains the function \textit{DEtoRE(DE,F(z),a[n])} which converts the holonomic differential equation DE depending on the variable \textit{z} into its corresponding recurrence equation for the coefficients \textit{a[n]}.
\begin{example}\label{eg42}\item
	
	\noindent
	\begin{minipage}[t]{8ex}\color{red}\bf
		\begin{verbatim}
			(%i1) 
		\end{verbatim}
	\end{minipage}
	\begin{minipage}[t]{\textwidth}\color{blue}
		\begin{verbatim}
			DE:HolonomicDE(asin(z),F(z))$
		\end{verbatim}
	\end{minipage}
	
	\noindent
	\begin{minipage}[t]{8ex}\color{red}\bf
		\begin{verbatim}
			(%i2) 
		\end{verbatim}
	\end{minipage}
	\begin{minipage}[t]{\textwidth}\color{blue}
		\begin{verbatim}
			DEtoRE(DE,F(z),a[n]);
		\end{verbatim}
	\end{minipage}
	\definecolor{labelcolor}{RGB}{100,0,0}
	\[\displaystyle
	\parbox{10ex}{$\color{labelcolor}\mathrm{\tt (\%o2) }\quad $}
	{{n}^{2}}\cdot {{a}_{n}}-\left( 1+n\right) \cdot \left( 2+n\right) \cdot {{a}_{n+2}}=0\mbox{}
	\]

    \textnormal{One can also directly compute these recurrence equations with our Maxima function\linebreak \textit{FindRE(f,z,a[n])}}.
	
	\noindent
	\begin{minipage}[t]{8ex}\color{red}\bf
		\begin{verbatim}
			(%i3) 
		\end{verbatim}
	\end{minipage}
	\begin{minipage}[t]{\textwidth}\color{blue}
		\begin{verbatim}
			FindRE(cos(z)+sin(z),z,a[n]);
		\end{verbatim}
	\end{minipage}
	\definecolor{labelcolor}{RGB}{100,0,0}
	\[\displaystyle
	\parbox{10ex}{$\color{labelcolor}\mathrm{\tt (\%o6) }\quad $}
	\left( 1+n\right) \cdot \left( 2+n\right) \cdot {{a}_{n+2}}+{{a}_{n}}=0\mbox{}
	\]
\end{example}

\subsection{Computing larger-order Taylor expansions of holonomic functions}

Hypergeometric type functions are strictly contained in the family of holonomic functions. Indeed, it is proved that linear combinations and products of holonomic functions are also holonomic (\cite{koepf1997algebra, stanley1980differentiably}). Although power series expansions of linear combinations of hypergeometric type functions remain accessible through the algorithm of Section \ref{sec3}, it is not generally the case with their products. This is because algorithm mfoldHyper cannot find explicit formulas for the coefficients of power series expansions of certain holonomic functions. Nevertheless, as we are able to find recurrence equations for the coefficients, the use of enough initial values coupled with their corresponding holonomic REs uniquely characterizes their Taylor coefficients in a certain neighborhood. It is thanks to this observation that Koepf proceeded in computing Taylor polynomials of holonomic expressions by using the outputs of \textit{FindRE} (see \cite[Chapter 10]{koepf2006computeralgebra}). We are going to use \textit{FindRE} to develop an algorithm to compute Taylor polynomials of holonomic functions and compare the result with Maxima's internal command \textit{taylor}. First, we give some particular normal forms for holonomic functions.

\subsubsection{On normal forms of holonomic functions}\label{normalform}

By an application of the well-known Cauchy-Lipschitz (also called Picard-Lindelöf) theorem (see \cite[Theorem 2.2]{teschl2012ordinary}) for uniqueness, the holonomic differential equation of lowest order and enough initial values corresponding to a holonomic function can be used for identification purposes. Therefore, such a representation constitutes a normal form (see \cite[Chapter 3]{geddes1992algorithms}). However, our procedure to compute holonomic DEs reduces this normal form definition of functions to expressions, because it might happen that equivalent expressions have two different representations. Consider the Chebyshev polynomials for example. For $|x|<1$, the following two differential equations define the same function.

\noindent
\begin{minipage}[t]{8ex}\color{red}\bf
	\begin{verbatim}
		(%i1) 
	\end{verbatim}
\end{minipage}
\begin{minipage}[t]{\textwidth}\color{blue}
	\begin{verbatim}
		DE1:HolonomicDE(cos(4*acos(x)),F(x));
	\end{verbatim}
\end{minipage}
\definecolor{labelcolor}{RGB}{100,0,0}
\[\displaystyle
\parbox{10ex}{$\color{labelcolor}\mathrm{\tt (\%o1) }\quad $}
\left( x-1\right) \cdot \left( 1+x\right) \cdot \left( \frac{{{d}^{2}}}{d\,{{x}^{2}}}\cdot \mathrm{F}\left( x\right) \right) +x\cdot \left( \frac{d}{d\,x}\cdot \mathrm{F}\left( x\right) \right) -16\cdot \mathrm{F}\left( x\right) =0\mbox{}
\]

\noindent
\begin{minipage}[t]{8ex}\color{red}\bf
	\begin{verbatim}
		(%i2) 
	\end{verbatim}
\end{minipage}
\begin{minipage}[t]{\textwidth}\color{blue}
	\begin{verbatim}
		DE2:HolonomicDE(8*x^4-8*x^2+1,F(x));
	\end{verbatim}
\end{minipage}
\definecolor{labelcolor}{RGB}{100,0,0}
\[\displaystyle
\parbox{10ex}{$\color{labelcolor}\mathrm{\tt (\%o2) }\quad $}
\left( 1-8\cdot {{x}^{2}}+8\cdot {{x}^{4}}\right) \cdot \left( \frac{d}{d\,x}\cdot \mathrm{F}\left( x\right) \right) -16\cdot x\cdot \left( 2\cdot {{x}^{2}}-1\right) \cdot \mathrm{F}\left( x\right) =0\mbox{}
\]

Note that this happens because \textit{HolonomicDE} does not use simplifications on its input expressions. However, one can easily prove that these two differential equations are compatible by substituting the lower order differential equation into the larger one. We implemented such a procedure in Maxima as \textit{CompatibleDE(DE1,DE2,F(x))}.

\noindent
\begin{minipage}[t]{4.000000em}\color{red}\bfseries
\begin{verbatim}
	(%i3) 
\end{verbatim}
\end{minipage}
\begin{minipage}[t]{\textwidth}\color{blue}
\begin{verbatim}
	CompatibleDE(DE1,DE2,F(x)); 
\end{verbatim}
\end{minipage}
\mbox{}\\The two differential equations are compatible
\[\parbox{10ex}{$\color{labelcolor}\mathrm{\tt (\%o3) }\quad $}
\text{true}\mbox{}
\]

Next we give the representation that we use as a first step towards computing Taylor expansions. Given an analytic expression $f(z)$ at $z_0$ whose Taylor coefficients satisfy a holonomic recurrence equation of the form
\begin{equation}
	P_d(n)\cdot a_{n+d} + P_{d-1}(n)\cdot P_{n+d-1}+\cdots+P_0(n)\cdot a_n =0,~ n\in\mathbb{Z},~d\in\mathbb{N}\label{simpleRE}
\end{equation}
with $P_0(n)\cdot P_d(n)\neq 0, \forall n \geqslant n_0$, $f(z)$ is identified to
\begin{equation}
	\sum_{n=0}^{\infty} a_{n+n_0} (z-z_0)^{n+n_0},~\text{ with }
	\begin{cases}
		a_{n+d} = \dfrac{P_{d-1}(n)\cdot a_{n+d-1}+\ldots+P_0(n)\cdot a_n}{P_d(n)},~n\geqslant n_0\\[3mm]
		a_j = \lim_{z\rightarrow z_0} \dfrac{\left(\frac{{{d}^{j}}}{d\,{{z}^{j}}}\cdot f\right)(z)}{j!},~j=n_0,n_0+1,\ldots,n_0+d-1
	\end{cases}. \label{formalRec}
\end{equation}

The value of $n_0$ is deduced by the property $P_0(n)\cdot P_d(n)\neq 0, \forall n \geqslant n_0$, that is
\begin{equation}
	n_0 = 1 + \max \lbrace n\in\mathbb{Z},~P_0(n)\cdot P_d(n) = 0 \rbrace. \label{n_0}
\end{equation}
Notice that $n_0$ is computed before any cancellation of common factors in $(\ref{simpleRE})$, which guarantees that $n_0$ does exist in general for any output of \textit{FindRE} of order $d\geqslant 1$. Indeed, the rewrite rule $(\ref{Algo3corr})$ allows to remark that the differential equation terms with derivative order greater than $1$ lead to recurrence equation terms with non-constant polynomial coefficients. The determination of $n_0$ is crucial to extract parts in the series expansion that are not involved in the summation formula. Much details about the computation of power series extra parts are given in Section \ref{sec4}.

\subsubsection{Taylor expansions of holonomic functions}\label{section432}

Let $f(z)$ be a holonomic function. The Taylor expansion of $f(z)$ at $z_0$ is computed as the one of $g(z)=f(z+z_0)$ at 0 if $z_0$ is a constant, of $g(z)=f\left(-1/z\right)$ if $z_0=-\infty$, and of $g(z)=f\left(1/z\right)$ if $z_0=\infty$. This algorithm is an immediate use of $(\ref{formalRec})$ with the following steps.

\begin{algorithm}[h!]
	\caption{Computing Taylor polynomials of holonomic functions at $z_0\in\mathbb{C}\cup\{-\infty,\infty\}$}
	\label{Algo3}
	\begin{algorithmic}[3]
		\Require A holonomic expression $f(z)$, a point $z_0$, and an integer $N$.
		\Ensure Taylor polynomial of degree $N$ of $f(z)$.
		\begin{enumerate}
			\item If $z_0\in\mathbb{C}$, set $g(z):=f(z+z_0)$, else if $z_0=-\infty$, set $g(z):=f\left(-\frac{1}{z}\right)$, else set $g(z):=f\left(\frac{1}{z}\right)$.
			\item Use \textit{FindRE} to compute a holonomic recurrence equation satisfied by the Taylor coefficients of $g(z)$ and write it in the form
			\begin{equation}
				P_d(n)\cdot a_{n+d} + P_{d-1}(n)\cdot a_{n+d-1}+\cdots+P_0(n)\cdot a_n =0,~ n\in\mathbb{Z},~d\in\mathbb{N}.\label{simpleREAlgo3}
			\end{equation}
			\item \label{step1Algo3} If $d=0$ then return the Taylor expansion of order $N$ with the internal command of Taylor expansions, say $taylor(f(z),z_0,N)$.
			\item Compute
			\begin{equation}
				n_0 = 1 + \max \lbrace n\in\mathbb{Z},~P_0(n)\cdot P_d(n) = 0 \rbrace. \label{n_0Algo3}
			\end{equation}
			\item \label{step2Algo3} If $N\leqslant n_0 + d - 1$ then stop and return $taylor(f(z),z_0,N)$.
			\item \label{Tlaststep} If $N>n_0 + d - 1$ then 
			\begin{equation}
				\mathcal{T} := taylor(f(z),z_0,n_0+d-1).
			\end{equation}
			\begin{itemize}
				\item[\ref{Tlaststep}-1-1] If $z_0\in\mathbb{C}$ then compute
				\begin{equation}
					a_j := \text{coeff}\left(T,z-z_0,j\right),~j=n_0,n_0+1,\ldots,n_0+d-1,
				\end{equation}
			\end{itemize}
		\end{enumerate}
		\algstore{pause2}
	\end{algorithmic}
\end{algorithm}
\clearpage 

\begin{algorithm}
	\ContinuedFloat
	\caption{Computing Taylor polynomials of holonomic functions at $z_0\in\mathbb{C}\cup\{-\infty,\infty\}$}
	\begin{algorithmic}
		\algrestore{pause2}	
		\State
		\begin{enumerate}
			\setcounter{enumi}{5}
			\item (resume) where $\text{coeff}(\mathcal{T},z-z_0,j)$ collects the coefficient of $(z-z_0)^j$ in $\mathcal{T}$. The Maxima syntax is adopted.
			\begin{itemize}
				\item[\ref{Tlaststep}-1-2] For $j=n_0, n_0+1,\ldots, N-d$, compute
				\begin{eqnarray}
					a_{j+d} &=& \dfrac{P_{d-1}(j)\cdot a_{j+d-1}+\ldots+P_0(j)\cdot a_j}{P_d(j)} \label{linearAlgo3}\\[3mm]
					\mathcal{T} 		&=& \mathcal{T} + a_{j+d} \cdot (z-z_0)^{j+d}
				\end{eqnarray}
			\end{itemize}
			\begin{itemize}
				\item[\ref{Tlaststep}-2-1] Else if $|z_0|=\infty$ then compute
				\begin{equation}
					a_j := \text{coeff}\left(\mathcal{T},1/z,j\right),~j=n_0,n_0+1,\ldots,n_0+d-1,
				\end{equation}
				\item[\ref{Tlaststep}-2-2] For $j=n_0, n_0+1,\ldots, N-d$, compute
				\begin{eqnarray}
					a_{j+d} &=& \dfrac{P_{d-1}(j)\cdot a_{j+d-1}+\ldots+P_0(j)\cdot a_j}{P_d(j)} \label{linearAlgo31}\\[3mm]
					\mathcal{T} 		&=& \mathcal{T} + a_{j+d} \cdot \left(\frac{1}{z}\right)^{j+d}
				\end{eqnarray}
			\end{itemize}
			
			\item Return $\mathcal{T}$.
		\end{enumerate}
		%    	\end{enumerate}
	\end{algorithmic}
\end{algorithm}

\begin{remark}\item
\begin{itemize}
	\item The relation $(\ref{linearAlgo3})$ shows that the coefficients are computed in the same finite number of operations. Therefore the complexity is linear.
	\item As we are interested by the asymptotic complexity of this algorithm, there is no issue of comparison when the internal command is called in step $\ref{step1Algo3}$ and $\ref{step2Algo3}$. And moreover, this helps to extract the part of the expansion which cannot be deduced from the recurrence equation used. An example is $\arcsech(z)$ for which Maxima's command \textit{taylor} gives the following expansion of order $4$ at $0$.
	
	\noindent
	\begin{minipage}[t]{8ex}\color{red}\bf
		\begin{verbatim}
			(%i1) 
		\end{verbatim}
	\end{minipage}
	\begin{minipage}[t]{\textwidth}\color{blue}
		\begin{verbatim}
			taylor(asech(z),z,0,4);
		\end{verbatim}
	\end{minipage}
	\definecolor{labelcolor}{RGB}{100,0,0}
	\[\displaystyle
	\parbox{10ex}{$\color{labelcolor}\mathrm{\tt (\%o1)/T/ }\quad $}
	-\mathrm{log}\left( z\right) +\mathrm{log}\left( 2\right) + \cdots
	-\frac{{{z}^{2}}}{4}-\frac{3\cdot {{z}^{4}}}{32}+ \cdots \mbox{}
	\]
	
	\item In order to treat certain interesting non-analytic cases like $\arcsech(z)$ in Maxima, instead of the limit command which can generate errors due to singularities, the internal Maxima command \textit{taylor} is used. The initial values are then the coefficients of $(z-z_0)^j,~j=n_0,\ldots,n_0+d-1$.
\end{itemize}
\end{remark}

We implemented Algorithm \ref{Algo3} as \textit{Taylor} with the same syntax as the internal command \textit{taylor}. Let us present some examples.
\clearpage
\begin{example}\item
	
	\noindent
	\begin{minipage}[t]{8ex}\color{red}\bf
		\begin{verbatim}
			(%i1) 
		\end{verbatim}
	\end{minipage}
	\begin{minipage}[t]{\textwidth}\color{blue}
		\begin{verbatim}
			Taylor(asech(z),z,0,7);
		\end{verbatim}
	\end{minipage}
	\definecolor{labelcolor}{RGB}{100,0,0}
	\[\displaystyle
	\parbox{10ex}{$\color{labelcolor}\mathrm{\tt (\%o1) }\quad $}
	-\mathrm{log}\left( z\right) -\frac{5\cdot {{z}^{6}}}{96}-\frac{3\cdot {{z}^{4}}}{32}-\frac{{{z}^{2}}}{4}+\mathrm{log}\left( 2\right) \mbox{}
	\]
	
	\noindent
	\begin{minipage}[t]{8ex}\color{red}\bf
		\begin{verbatim}
			(%i2) 
		\end{verbatim}
	\end{minipage}
	\begin{minipage}[t]{\textwidth}\color{blue}
		\begin{verbatim}
			Taylor(atan(z),z,inf,7);
		\end{verbatim}
	\end{minipage}
	\definecolor{labelcolor}{RGB}{100,0,0}
	\[\displaystyle
	\parbox{10ex}{$\color{labelcolor}\mathrm{\tt (\%o2) }\quad $}
	-\frac{1}{z}+\frac{1}{3\cdot {{z}^{3}}}-\frac{1}{5\cdot {{z}^{5}}}+\frac{1}{7\cdot {{z}^{7}}}+\frac{\pi }{2}\mbox{}
	\]
	
	\textnormal{Next we evaluate the timings for larger orders. We mention that when the given expression is a classical one like $\sin(z),\arctan(z), \exp(z), etc$, Maxima seems to use the power series formula and has very good asymptotic timings. Therefore, for tests we rather use expressions for which the internal Maxima command} \textit{powerseries} \textnormal{cannot find the power series formulas.}

\noindent
\begin{minipage}[t]{4.000000em}\color{red}\bfseries
\begin{verbatim}
	(%i3) 
\end{verbatim}
\end{minipage}
\begin{minipage}[t]{\textwidth}\color{blue}
\begin{verbatim}
	Taylor(atan(z)*exp(z),z,0,1000)$
\end{verbatim}
\end{minipage}
\[\displaystyle \text{Evaluation took 2.3120 seconds (2.3140 elapsed) using 834.089 MB.}\mbox{}
\]

\noindent
\begin{minipage}[t]{4.000000em}\color{red}\bfseries
\begin{verbatim}
	(%i4) 
\end{verbatim}
\end{minipage}
\begin{minipage}[t]{\textwidth}\color{blue}
\begin{verbatim}
	taylor(atan(z)*exp(z),z,0,1000)$
\end{verbatim}
\end{minipage}
\[\displaystyle \text{Evaluation took 34.2350 seconds (34.2220 elapsed) using 5800.949 MB.}\mbox{}
\]
	
\textnormal{The gap between the two computations gets larger as we increase the order.}
\end{example}

\section{Algorithm mfoldHyper}\label{sec2}

This section presents a new result in solving holonomic recurrence equations, that is the computation of their $m$-fold hypergeometric term solutions, for positive integers $m$. Several proposals have been given to compute such solutions, among the most recent work in this direction one could cite (\cite{horn2012m}), which is a revisited and improved approach of the one described in (\cite{petkovvsek1993finding}). In the latter, a key step of the proposed algorithm relies on the determination of the linear operator's right factors of the given holonomic RE. Such a factorization is not unique in general because the factors do not commute. In (\cite{horn2012m}), the authors adapted van Hoeij's approach as explained in (\cite{cluzeau2006computing}) and define a concept like the $m$-Newton polygon for $m$-fold hypergeometric term solutions of a given holonomic RE. This approach computes special types of right factors corresponding to $m$-fold hypergeometric term solutions using the shift operator of order $m$ with the hypothesis that no rational solution exists. With similar approaches, $m$-fold hypergeometric terms were referred to as $m$-hypergeometric sequences in (\cite{petkovvsek1993finding}), $m$-interlacings of hypergeometric sequences (see the conclusion of \cite{van1999finite}), or Liouvillian sequences (see \cite{hendricks1999solving}). Having considered all these developments, we propose to attack the problem from another point of view. We will see that contrary to all these previous methods which try to find solutions in higher dimensions (defined by the solution space) encoded by right factors of a given RE's operator, our method only focus on one dimension to deduce a basis of the subspace of $m$-fold hypergeometric term solutions. However, before getting in the description of mfoldHyper, we would like to highlight its importance based in hypergeometric summation.

\subsection{Scope of the algorithm}

Computations of infinite series were connected to holonomic REs by Celine Fasenmyer. To find hypergeometric term representations of hypergeometric series, she proposed the following approach (see \cite[Chapter 4]{WolfBook}).

Given a sum $\sum_{k=0}^{n}F(n,k)$, write\footnote{The sum is taken over $\mathbb{Z}$ because the summation term vanishes outside a finite set, we say that it has a finite support.}
\begin{equation}
	s_n = \sum_{k=-\infty}^{\infty} F(n,k),
\end{equation}
and search for polynomials $p_{i,j}=p_{i,j}(n), i=0,\ldots,b,~j=0,\ldots,d$ with respect to $n$ and do not depend on $k$, such that
\begin{equation}
	\sum_{i=0}^{b}\sum_{j=0}^{d} p_{i,j} F(n+j,k+i) =0.
\end{equation}
If such polynomials are found, then one deduces a holonomic recurrence equation of order at most $d$ for $s_n$ as follows
\begin{multline*}
0  = \sum_{k=-\infty}^{\infty}\sum_{i=0}^{b}\sum_{j=0}^{d}p_{i,j}F(n+j,k+i)=\sum_{i=0}^{b}\sum_{j=0}^{d}p_{i,j}\cdot \left(\sum_{k=-\infty}^{\infty} F(n+j,k+i)\right)\\
=\sum_{i=0}^{b}\sum_{j=0}^{d}p_{i,j} s_{n+j}=\sum_{j=0}^{d}\left(\sum_{i=0}^{b}p_{i,j}\right)s_{n+j},
\end{multline*}
where of course we use the fact that $p_{i,j}$ do not depend on $k$ and the advantage of working with bilateral sums: their value is invariant with respect to shifts of the summation variable. In the 1940s, Fasenmyer's method could be used to compute explicit formulas of hypergeometric series only when the obtained holonomic RE was of first order or a two-term recurrence relation.
\begin{example}Applied to $s_n=\sum_{k=0}^{n}k\binom{n}{k}$, Fasenmyer's method leads to the recurrence equation
	$$ns_{n+1}-2 (n+1) s_n =0,$$
	which after use of the initial value $s_1=\sum_{k=0}^{1}k\binom{1}{k}=1$, yields $s_n=n2^{n-1}$.
\end{example}

For the definite summation case, Gosper (see \cite[Chapter 5]{WolfBook}) proposed an algorithm which deals with the question of how to find a (forward) anti-difference $s_k$ for a given $a_k$, that is a sequence $s_k$ such that 
\begin{equation}
	a_k = \Delta s_k = s_{k+1} - s_k, \label{Gosperdiff}
\end{equation}
in the particular case that $s_k$ is a hypergeometric term. Thus, once a hypergeometric anti-difference $s_k$ of $a_k$ is computed, by telescoping definite summation yields
\begin{equation*}
	\sum_{k=n_0}^{n}a_k = (s_{n+1}-s_n) + (s_n - s_{n-1}) + \cdots + (s_{n_0+1}-s_{n_0}) = s_{n+1} - s_{n_0},
\end{equation*}
by an evaluation at the limits of summation. Gosper's idea is based on the representation 
\begin{equation}
	\dfrac{a_{k+1}}{a_k} = \dfrac{p_{k+1}}{p_k}\cdot \dfrac{q_{k+1}}{r_{k+1}}, \label{Gosper1}
\end{equation}
with the property
\begin{equation}
	\gcd(q_k,r_{k+j})=1,~\forall j\in\mathbb{N}_{\geqslant 0} \label{Gosper2}
\end{equation}
that can be algorithmically generated (see \cite[Lemma 5.1 and Algorithm]{WolfBook}). Using $(\ref{Gosper1})$ and $(\ref{Gosper2})$, one proves that the function
\begin{equation}
	f_k := \dfrac{s_{k+1}}{a_{k+1}}\cdot \dfrac{p_{k+1}}{r_{k+1}} \label{Gosperf}
\end{equation}
must be a polynomial for a hypergeometric term anti-difference $a_k$ to exist. Thus using $(\ref{Gosperdiff})$, $(\ref{Gosperf})$ and $(\ref{Gosper1})$ it follows that $f_k$ satisfies the inhomogeneous recurrence equation
\begin{equation}
	p_k = q_{k+1} f_{k+1} - r_k f_{k-1}.
\end{equation}
Gosper gives an upper bound for the degree of $f_k$ in terms of $p_k,$ $q_k,$ and $r_k$ which yields a method for calculating $f_k$ by introducing the appropriate generic polynomial, equating coefficients, and solving the corresponding linear system so that we finally find
\begin{equation}
	s_k = \dfrac{r_k}{p_k}f_{k-1} a_k.
\end{equation}

Gosper implemented his algorithm in Maxima as \textit{nusum} with the same syntax as the Maxima \textit{sum} command.
\begin{example}\item
	
	\noindent
	\begin{minipage}[t]{8ex}\color{red}\bf
		\begin{verbatim}
			(%i1) 
		\end{verbatim}
	\end{minipage}
	\begin{minipage}[t]{\textwidth}\color{blue}
		\begin{verbatim}
			nusum(k*k!,k,0,n);
		\end{verbatim}
	\end{minipage}
	\definecolor{labelcolor}{RGB}{100,0,0}
	\mbox{}\\\mbox{solve: dependent equations eliminated: (1)}
	\[\displaystyle
	\parbox{10ex}{$\color{labelcolor}\mathrm{\tt (\%o1) }\quad $}
	\left( 1+n\right) !-1\mbox{}
	\]
	
	\noindent
	\begin{minipage}[t]{8ex}\color{red}\bf
		\begin{verbatim}
			(%i2) 
		\end{verbatim}
	\end{minipage}
	\begin{minipage}[t]{\textwidth}\color{blue}
		\begin{verbatim}
			nusum(k^3,k,0,n);
		\end{verbatim}
	\end{minipage}
	\definecolor{labelcolor}{RGB}{100,0,0}
	\[\displaystyle
	\parbox{10ex}{$\color{labelcolor}\mathrm{\tt (\%o2) }\quad $}
	\frac{{{n}^{2}}\cdot {{\left( 1+n\right) }^{2}}}{4}\mbox{}
	\]
\end{example}

Gosper's algorithm is the essential tool of the so called Wilf-Zeilberger (often named WZ) method (see \cite[Chapter 6]{WolfBook}). That is a clever application of Gosper's algorithm to prove identities of the form
\begin{equation}
	s_n:=\sum_{k=-\infty}^{\infty} F(n,k) = 1 \label{wzeq}
\end{equation}
where $F(n,k)$ is a hypergeometric term with respect to both $n$ and $k$ with finite support. For this purpose, one applies Gosper's algorithm to the expression
\begin{equation}
	a_k:= F(n+1,k)-F(n,k)
\end{equation}
with respect to the variable $k$. If successful, this generates $G(n,k)$ with
\begin{equation}
	a_k=F(n+1,k)-F(n,k)=G(n,k+1)-G(n,k),
\end{equation}
and summing over $\mathbb{Z}$ yields
\begin{equation}
	s_{n+1}-s_n = \sum_{k=-\infty}^{\infty} F(n+1,k)-F(n,k) = \sum_{k=-\infty}^{\infty} G(n,k+1) - G(n,k) =0 \label{wz1}
\end{equation}
since the right-hand side is telescoping. Therefore, $s_n$ is constant, $s_n=s_0$, and it only remains to prove that $s_0=1$. In practice, once the function $G(n,k)$ is computed one uses the rational function 
\begin{equation}
	R(n,k):=\dfrac{G(n,k)}{F(n,k)}
\end{equation}
called the WZ certificate of $F(n,k)$, to establish $(\ref{wzeq})$ by proving the rational identity
\begin{equation}
	\dfrac{F(n+1,k)}{F(n,k)} - 1 + R(n,k) - R(n,k+1)\dfrac{F(n,k+1)}{F(n,k)} = 0 \label{idwz}
\end{equation}
which is deduced from $(\ref{wz1})$ after division by $F(n,k)$.
\begin{example} For $s_n:=\sum_{k=0}^{n}F(n,k)=\sum_{k=0}^{n}\dfrac{1}{2^n}\binom{n}{k}=1$, the WZ certificate is 
	\begin{equation*}
		R(n,k) = -\dfrac{k}{2(n+1-k)}.
	\end{equation*} 
	Therefore the corresponding left hand side of identity $(\ref{idwz})$ is
	\begin{equation*}
		\dfrac{n+1}{2(n+1-k)} - 1 - \dfrac{k}{2(n+1-k)} + \dfrac{k+1}{(2(n-k))}\cdot\dfrac{n-k}{k+1}
	\end{equation*}
	which trivially yields zero.
\end{example}

Although Gosper's algorithm applies to finite summation, it constitutes a useful tool in discovering a method for infinite sums. This is observable in Zeilberger's algorithm (see \cite[Chapter 7]{WolfBook}). Zeilberger brings back the computation of a holonomic recurrence equation for $s_n:=\sum_{k=-\infty}^{\infty} F(n,k)$. The idea is to apply Gosper's algorithm in the following way: For suitable $d=1,2,\ldots$ set
\begin{equation}
	a_k:= F(n,k) + \sum_{j=1}^{d}\sigma_j(n) F(n+j,k)
\end{equation}
where $\sigma_j$ is supposed to be a rational function depending on $n$ and not on $k$. Zeilberger's main observation is that the computation of the polynomial $f_k$ defined in $(\ref{Gosperf})$ yields a linear system not only for the unknown coefficients of $f_k$, but also for the rational functions $\sigma_j, j=1,\ldots,d$.

Thus in a successful case, one obtains an anti-difference $G(n,k)$ of $a_k$ and rational functions $\sigma_j(n),j=1,\ldots,d$ such that
\begin{equation}
	a_k = G(n,k+1)-G(n,k) = F(n,k) + \sum_{j=1}^{d} \sigma_j(n) F(n+j,k).
\end{equation}
Hence, by summation
\begin{eqnarray}
	0=\sum_{k=-\infty}^{\infty} G(n,k+1)-G(n,k)&=& \sum_{k=-\infty}^{\infty} \left(F(n,k) + \sum_{j=1}^{d}\sigma_j(n)F(n+j,k) \right)\nonumber\\
	&=& s_n + \sum_{j=1}^{d}\sigma_j(n) s_{n+j}.
\end{eqnarray}
After multiplication by the common denominator one gets the holonomic recurrence equation sought.

Zeilberger's algorithm gives a much better possibility of computing identities since it computes holonomic recurrence equations generally of lowest order (iteration on the order $d$) for a given hypergeometric series. Moreover, this approach is also used to show the coincidence of two sums provided their initial values. This constitutes a normal form again, as we have seen in Section $\ref{normalform}$ (see \cite[Chapter 3]{geddes1992algorithms}).
\begin{example} The sums $\sum_{k=0}^{n}\binom{n}{k}^3$ and $\sum_{k=0}^{n}\binom{n}{k}^2\binom{2k}{n}$ are proved to coincide by Zeilberger's algorithm as they lead to the same holonomic recurrence equation
	\begin{equation}
		-(n+2)^2 s_{n+2} + (7n^2+21n+16)s_{n+1} + 8(n+1)^2s_n =0,
	\end{equation}
	and have the same initial values 
	\[\sum_{k=0}^{0}\binom{0}{k}^3=\sum_{k=0}^{0}\binom{0}{k}^2\binom{2k}{0}=1 \text{ and }  ~\sum_{k=0}^{1}\binom{1}{k}^3=\sum_{k=0}^{1}\binom{1}{k}^2\binom{2k}{1}=2\]
\end{example}

Note, however, that Fasenmyer's, Gosper's, Wilf-Zeilberger's and Zeilberger's methods are only reduced to the hypergeometric case. To use Gosper's algorithm, one has to check whether the given $a_k$ is a hypergeometric term with respect to the variable $k$. When this is not the case, Gosper's, WZ's and Zeilberger's methods cannot be applied. From his algorithm \cite[Algorithm 2.2]{WolfBook}, Koepf observed that Gosper's method could miss some results when rational-linear $\Gamma$ inputs are considered rather than only integer-linear ones. It turns out that this observation constitutes the connection to the general case of $m$-fold hypergeometric terms (see \cite[Chapter 8]{WolfBook}).
\begin{example} To the Watson's function
	\begin{equation*}
		_3F_2\left(\begin{matrix}-n& b& c\\ \frac{-n+b+1}{2}& 2c\end{matrix}~\bigg|~ 1 \right),
	\end{equation*}
	Zeilberger's algorithm does not directly apply. However, using Koepf's extended version yields the recurrence equation
	\begin{equation}
		(b-2c-n-1)(n+1)a_n - (b-n-1)(2c+n+1) a_{n+2} = 0.
	\end{equation}
	Note that the computation of this holonomic RE is made possible after application of \cite[Algorithm 8.4]{WolfBook} for finding the corresponding $m$ to use. In this example one finds $m=2$.
\end{example}

The importance of finding $m$-fold hypergeometric term solutions of holonomic recurrence equations could not easily be seen from the known hypergeometric database (see \cite[Chapter 3]{WolfBook}). This might be linked to the influence of the combinatorial interpretation often present in the use of hypergeometric summations in the last century. From this point of view, the $m$-fold hypergeometric case ($m\in\mathbb{N}_{\geqslant 2}$) is particularly hidden and Zeilberger's, Petkov{\v{s}}ek's and van Hoeij's algorithms or their modifications for any hypergeometric situation (multivariate for example) remain the best approach possible. However, as pointed out in \cite{koepf1995algorithms}, the general $m$-fold hypergeometric case might be the source of a wider family of hypergeometric identities. This is shown in particular in the computation of power series where some computed holonomic recurrence equations do not have hypergeometric term solutions but only $m$-fold hypergeometric ones ($m\in\mathbb{N}_{\geqslant 2}$). Note that for the definite summation case the availability of an algorithm which computes $m$-fold hypergeometric term solutions of holonomic REs can be combined with Koepf's extension of Zeilberger's algorithm to generate new identities. We will not deal with definite summation, instead as the importance of $m$-fold hypergeometric terms for hypergeometric series is already shown, we focus on power series computations where the need of such terms is essential for the goal of this paper. There are many limitations of the currently used approach to compute power series, we presented some in Section \ref{sec0}; more can be found in (\cite{BTphd}). 

\subsection{Description of the algorithm}

Let $\mathbb{K}$ be a field of characteristic zero. We consider the generic holonomic recurrence equation
\begin{equation}
	P_d(n) a_{n+d} + P_{d-1}(n)a_{n+d-1}+\cdots+ P_0(n)a_n=0, \label{REmyalgo}
\end{equation}
$P_d(n),\ldots,P_0(n)\in\mathbb{K}[n], P_d(n)\cdot P_0(n)\neq 0$.

By definition, a sequence $a_n$ is said to be $m$-fold hypergeometric, $m\in\mathbb{N}$, if there exists a fixed rational function $r(n)\in\mathbb{K}(n)$ such that
\begin{equation}
	r(n) = \dfrac{a_{n+m}}{a_n}.\label{mhypdef}
\end{equation}

At first glance, one should remark that $m$-fold hypergeometric sequences have rational functions as the ratio of terms with index difference equal to $m$. Consequently, if we can find a way to transform this property to the simple one of hypergeometric sequence then iteratively up to the order of the given holonomic RE, by van Hoeij's algorithm we are done.

From the characterization $(\ref{mhypdef})$ one can deduce that for $0\leqslant j \leqslant m-1$ the following is valid
\begin{equation}
	r(m\cdot n+j) = \dfrac{a_{m\cdot (n+1)+j}}{a_{m\cdot n+j}}. \label{mjhypdef}
\end{equation}
Therefore instead of considering the representation $(\ref{mhypdef})$ one could rather see an $m$-fold hypergeometric term with $m$ related rational functions as defined in $(\ref{mjhypdef})$. This latter representation is the one used to find a "simple" formula of an $m$-fold hypergeometric term as shown in (\cite{BThyper}). Moreover, for fixed $m\in\mathbb{N}$ and $j\in\mathbb{N}_{\geqslant 0}$, if we compute an $m$-fold hypergeometric term solution of $(\ref{REmyalgo})$ with ratio $r(m\cdot n + j)$ for some rational function $r$, then this gives the information that there are $m-1$ other similar $m$-fold hypergeometric term solutions of $(\ref{REmyalgo})$. There appears the particularity of our approach. Indeed, $r(m\cdot n + j)$ is the term of a sub-sequence of $(r(n))_n$. Unlike previous approaches that try to find $r(n)$ directly, we rather compute the sub-terms $r(m\cdot n + j)$ for fixed $j\in\llbracket 0, m-1\rrbracket$ so that $r(n)$ is constructed as a linear combination of these sub-terms. 

Thus for every positive integer $m$, the computation of an $m$-fold hypergeometric term solution of  $(\ref{REmyalgo})$ with representation $(\ref{mhypdef})$ reduces to the computation of an $m$-fold hypergeometric term solution of  $(\ref{REmyalgo})$ with representation $(\ref{mjhypdef})$ for a fixed $j\in\llbracket 0, m-1\rrbracket$ since the other representations can be similarly computed. When given for a specific value of $j$, we will say that the solution is given in an incomplete form. The complete form of the solution is then given with all the representations of $(\ref{mjhypdef})$ for all $m$-fold hypergeometric term solutions. By default in our algorithm we choose $j=0$ for the incomplete form of the solution to be given as output.

In certain cases, depending on the field $\mathbb{K}$ or the index variable subset of $\mathbb{Z}$, holonomic recurrence equations can have $m$-fold hypergeometric term solutions that cannot be computed over $\mathbb{K}$, but in an extension field instead. This situation occurs with the power series of $\exp(z)\sin(z)$ for $\mathbb{K}=\mathbb{Q}$. Our implementation \textit{HolonomicDE(f,F(z),[destep])} has an optional variable \textit{destep} whose default value is $1$. This number represents the minimum positive difference possible between the derivatives of \textit{F(z)} in the holonomic differential equation sought. Our program \textit{FindRE} is also adapted for such computations. This particular tool turns out to be important in few cases. Let us examine the situation with $\exp(z)\sin(z)$.

\noindent
\begin{minipage}[t]{8ex}\color{red}\bf
	\begin{verbatim}
		(%i1) 
	\end{verbatim}
\end{minipage}
\begin{minipage}[t]{\textwidth}\color{blue}
	\begin{verbatim}
		DE1:HolonomicDE(exp(z)*sin(z),F(z));
	\end{verbatim}
\end{minipage}
\definecolor{labelcolor}{RGB}{100,0,0}
\[\displaystyle
\parbox{10ex}{$\color{labelcolor}\mathrm{\tt (\%o1) }\quad $}
\frac{{{d}^{2}}}{d\,{{z}^{2}}}\cdot \mathrm{F}\left( z\right) -2\cdot \left( \frac{d}{d\,z}\cdot \mathrm{F}\left( z\right) \right) +2\cdot \mathrm{F}\left( z\right) =0\mbox{}
\]

\noindent
\begin{minipage}[t]{8ex}\color{red}\bf
	\begin{verbatim}
		(%i2) 
	\end{verbatim}
\end{minipage}
\begin{minipage}[t]{\textwidth}\color{blue}
	\begin{verbatim}
		DE2:HolonomicDE(exp(z)*sin(z),F(z),2);
	\end{verbatim}
\end{minipage}
\definecolor{labelcolor}{RGB}{100,0,0}
\[\displaystyle
\parbox{10ex}{$\color{labelcolor}\mathrm{\tt (\%o2) }\quad $}
\frac{{{d}^{4}}}{d\,{{z}^{4}}}\cdot \mathrm{F}\left( z\right) +4\cdot \mathrm{F}\left( z\right) =0\mbox{}
\]
The compatibility of these two differential equations is shown below (see Section $\ref{normalform}$).

\noindent
\begin{minipage}[t]{8ex}\color{red}\bf
	\begin{verbatim}
		(%i3) 
	\end{verbatim}
\end{minipage}
\begin{minipage}[t]{\textwidth}\color{blue}
	\begin{verbatim}
		CompatibleDE(DE1,DE2,F(z));
	\end{verbatim}
\end{minipage}
\definecolor{labelcolor}{RGB}{100,0,0}
\mbox{}\\The two differential equations are compatible
\[\displaystyle
\parbox{10ex}{$\color{labelcolor}\mathrm{\tt (\%o3) }\quad $}
\text{true}\mbox{}
\]
Let us now compute the corresponding holonomic REs.

\noindent
\begin{minipage}[t]{8ex}\color{red}\bf
	\begin{verbatim}
		(%i4) 
	\end{verbatim}
\end{minipage}
\begin{minipage}[t]{\textwidth}\color{blue}
	\begin{verbatim}
		RE1:FindRE(exp(z)*sin(z),z,a[n]);
	\end{verbatim}
\end{minipage}
\definecolor{labelcolor}{RGB}{100,0,0}
\[\displaystyle
\parbox{10ex}{$\color{labelcolor}\mathrm{\tt (\%o4) }\quad $}
\left( 1+n\right) \cdot \left( 2+n\right) \cdot {{a}_{n+2}}-2\cdot \left( 1+n\right) \cdot {{a}_{n+1}}+2\cdot {{a}_{n}}=0\mbox{}
\]

\noindent
\begin{minipage}[t]{8ex}\color{red}\bf
	\begin{verbatim}
		(%i5) 
	\end{verbatim}
\end{minipage}
\begin{minipage}[t]{\textwidth}\color{blue}
	\begin{verbatim}
		RE2: FindRE(exp(z)*sin(z),z,a[n],2);
	\end{verbatim}
\end{minipage}
\definecolor{labelcolor}{RGB}{100,0,0}
\[\displaystyle
\parbox{10ex}{$\color{labelcolor}\mathrm{\tt (\%o5) }\quad $}
\left( 1+n\right) \cdot \left( 2+n\right) \cdot \left( 3+n\right) \cdot \left( 4+n\right) \cdot {{a}_{n+4}}+4\cdot {{a}_{n}}=0\mbox{}
\]

Observe that for no solution of \textit{RE2} a relationship can be deduced between two of its terms whose index difference is not a multiple of $4$. Therefore considering indices in $4\mathbb{N}=\{0,4,8,\ldots\}$ might be more appropriate. On the other hand, it is trivial that \textit{RE2} is a characteristic holonomic recurrence equation of a $4$-fold hypergeometric term. However, this $4$-fold hypergeometric term is completely hidden in \textit{RE1} when looking for solutions over $\mathbb{Q}$. Indeed, since the corresponding $4$-fold symmetric terms are linearly independent over $\mathbb{Q}$, a substitution in the left-hand side of \textit{RE1} does not yield zero. Note, however, that \textit{RE1} and \textit{RE2} have hypergeometric term solutions over $\mathbb{C}$. We compute these solutions using the first author's variant of van Hoeij's algorithm (\cite{BThyper}) implemented in Maxima as \textit{HypervanHoeij}.

\noindent
\begin{minipage}[t]{8ex}\color{red}\bf
	\begin{verbatim}
		(%i6) 
	\end{verbatim}
\end{minipage}
\begin{minipage}[t]{\textwidth}\color{blue}
	\begin{verbatim}
		HypervanHoeij(RE1,a[n],C);
	\end{verbatim}
\end{minipage}
\definecolor{labelcolor}{RGB}{100,0,0}
\[\displaystyle
\parbox{10ex}{$\color{labelcolor}\mathrm{\tt (\%o6) }\quad $}
\left\{\frac{{{\left( 1-i\right) }^{n}}}{n!},\frac{{{\left( 1+i\right) }^{n}}}{n!}\right\}\mbox{}
\]

\noindent
\begin{minipage}[t]{8ex}\color{red}\bf
	\begin{verbatim}
		(%i7) 
	\end{verbatim}
\end{minipage}
\begin{minipage}[t]{\textwidth}\color{blue}
	\begin{verbatim}
		HypervanHoeij(RE2,a[n],C);
	\end{verbatim}
\end{minipage}
\definecolor{labelcolor}{RGB}{100,0,0}
\[\displaystyle
\parbox{10ex}{$\color{labelcolor}\mathrm{\tt (\%o7) }\quad $}
\left\{\frac{{{\left( -{{\left( -1\right) }^{\frac{1}{4}}}\cdot \sqrt{2}\right) }^{n}}}{n!},\frac{{{\left( {{\left( -1\right) }^{\frac{1}{4}}}\cdot \sqrt{2}\right) }^{n}}}{n!},\frac{{{\left( -{{\left( -1\right) }^{\frac{1}{4}}}\cdot \sqrt{2}\cdot i\right) }^{n}}}{n!},\frac{{{\left( {{\left( -1\right) }^{\frac{1}{4}}}\cdot \sqrt{2}\cdot i\right) }^{n}}}{n!}\right\}\mbox{}
\]

The basis of hypergeometric term solutions of \textit{RE1} spans a sub-space of the space of solutions of \textit{RE2}. The corresponding $4$-fold hypergeometric term solutions over $\mathbb{Q}$ can be written as a linear combination of the above bases of hypergeometric terms over $\mathbb{C}$. The main thing that we point out from this example is that our approach to compute $m$-fold hypergeometric term solutions of a given holonomic RE depends on the shifts between indices of the indeterminate ($a_n$) in that RE and the field considered.

Next, we introduce some properties and definitions that clarify this situation and help to compute $m$-fold hypergeometric term solutions of $(\ref{REmyalgo})$ in a given field $\mathbb{K}$ which we want to be the smallest algebraic extension field possible of $\mathbb{Q}$ in terms of inclusion.

The following lemma gives a condition on the order of a given holonomic RE for its $m$-fold hypergeometric term solutions to be computable over a given field $\mathbb{K}$. An equivalent statement was given in \cite[Theorem 5.1]{hendricks1999solving} but with a different perspective as shown by the two proofs.

\begin{lemma}\label{lemmamfold1} Let $h_n$ be an $m$-fold hypergeometric term, $m\in\mathbb{N}$. Assume
	\begin{equation}
		\forall u \in\mathbb{N},~ u < m, ~\text{ there is no rational function } r_u(n)\in\mathbb{K}(n) ~:~ h_{u+n} = r_u(n) h_n. \label{asumlem}
	\end{equation}
	Then there is no holonomic recurrence equation over $\mathbb{K}$ of order less than $m$ satisfied by $h_n$.
\end{lemma}

\begin{proof}
	Let $h_n$ be an $m$-fold hypergeometric term such that 
	\begin{equation}
		h_{n+m} = r(n)\cdot h_n \iff Q_m(n)\cdot h_{n+m} + Q_0(n)\cdot h_n = 0, \label{defhn}
	\end{equation}
	where $Q_m(n),Q_0(n)\in\mathbb{K}[n]$ and $r(n)=-Q_0(n)/Q_m(n)\in\mathbb{K}(n)$.
	
	Suppose that $h_n$ satisfies a holonomic recurrence equation of order less than $m$. Then there exists an equation of the form
	\begin{equation}
		P_{m-1}a_{n+m-1} + P_{m-2} a_{n+m-2}+\cdots + P_{1} a_{n+1} + P_0 a_n =0, \label{hypholo} 
	\end{equation}
	with polynomials $P_j=P_j(n) \in\mathbb{K}[n], j\in\llbracket 0, m-1 \rrbracket$, and $P_0(n)\neq 0$, satisfied by $h_n$.
	\begin{itemize}
		\item If $P_0$ is the only non-zero polynomial in the equation then $h_n$ must be zero, which is a contradiction by definition.
		\item We assume that at least one other polynomial factor in the equation is non-zero. Then $h_n$ satisfying $(\ref{hypholo})$ yields the following equation after substitution of $n$ by $m\cdot n$
		\begin{equation}
			P_{m-1}(mn)h_{mn+m-1} + P_{m-2}(mn) h_{mn+m-2}+\cdots + P_{1}(mn) h_{mn+1} + P_0(mn) h_{mn} =0. \label{hypholo2} 
		\end{equation}
		By assumption $(\ref{asumlem})$, we know that $\forall u \in\mathbb{N},~ u < m$, $h_n$ is not a $u$-fold hypergeometric term. So the holonomic recurrence equation of lowest order over $\mathbb{K}$ satisfied by $h_n$ is
		\begin{equation*}
			Q_m(n)\cdot a_{n+m} + Q_0(n)\cdot a_n = 0,
		\end{equation*}
		which is a two-term recurrence relation whose subspace of $m$-fold hypergeometric term ($m$ is fixed) solutions can be represented by the basis
		\begin{equation}
			\left(h_{mn+m-1}, h_{mn+m-2},\ldots,h_{mn+1},h_{mn}\right). \label{linind}
		\end{equation}
		Therefore $(\ref{hypholo2})$ cannot hold since the left-hand side is a linear combination of linearly independent terms with respect to $\mathbb{K}(n)$, which implies that all the polynomial coefficients must be zero. Therefore we get a contradiction.
	\end{itemize}
\end{proof}

%\begin{remark}\label{Rem2}\item
%	Observe that the linear independence with respect to $\mathbb{K}(n)$ of the elements of the basis $(\ref{linind})$ used in the proof of Lemma \ref{lemmamfold1} can be interpreted in the following different way. Since $h_n\neq 0$, $h_n$ satisfying $(\ref{hypholo})$ yields the following identity after dividing $(\ref{hypholo})$ by $h_n$
%	\begin{equation}
%		-P_0 = P_{m-1}\dfrac{h_{n+m-1}}{h_n} + P_{m-2} \dfrac{h_{n+m-2}}{h_n} + \ldots + P_1  \dfrac{h_{n+1}}{h_n}. \label{hypholo3}
%	\end{equation}
%	By assumption $(\ref{asumlem})$, we know that $\forall u \in\mathbb{N},~ u < m$, the ratio $\frac{h_{n+u}}{h_n}$ is not a rational function over $\mathbb{K}(n)$. So, each non-zero term on the right-hand side of $(\ref{hypholo3})$ is not rational over $\mathbb{K}$. This does not necessarily implies the non-rationality of the whole right-hand side. However by the linear independence of the elements of the basis $(\ref{linind})$ one can interpret that this holds. For simplicity we can then consider that the obtained equality between an irrational and a rational $(P_0(n))$ terms with respect to $\mathbb{K}(n)$ to conclude the proof. this is how we will proceed in the proof of Theorem \ref{mfoldtheo} where the situation is more sophisticated.
%\end{remark}

More generally, any shift of a holonomic recurrence equation of order less than $m$ does not have $m$-fold hypergeometric term solutions.

Remark that checking the hypothesis of Lemma $\ref{lemmamfold1}$ is an important task for the algorithm. Fortunately, this can be done iteratively. Once the field $\mathbb{K}$ is fixed, if we have already looked for $u$-fold hypergeometric term solutions for integers $u<m$, then we can safely proceed to the computation of $m$-fold hypergeometric term solutions knowing that $m$ is less than the order of that recurrence equation.

Thus, we now know that all the $m$-fold hypergeometric term solutions of $(\ref{REmyalgo})$ have $m\leqslant d$. Furthermore, we can extend this view of $m$-fold hypergeometric sequences in order to determine which type of terms can appear in a holonomic recurrence equation that they satisfy. For that purpose, let us first introduce the following definition.

\begin{definition}[$m$-fold holonomic recurrence equation] A holonomic recurrence equation is said to be $m$-fold holonomic, $m\in\mathbb{N}$, if it has at least two non-zero polynomial coefficients and the difference between indices of two appearing terms of the indeterminate sequence in that equation is a multiple of $m$. Choosing $0$ as the trailing term order gives the general form
	\begin{equation}
		P_d(n)\cdot a_{n+md} + P_{d-1}(n)\cdot a_{n+m(d-1)} + \cdots + P_1(n)\cdot a_{n+m} + P_0(n)\cdot a_{n}=0, \label{mhypre}.
	\end{equation} 
	so that $P_d\cdot P_0\neq 0$.
\end{definition}

Assume an $m$-fold holonomic RE with representation $(\ref{mhypre})$ is given. We are going to present how to compute a basis of all $m$-fold hypergeometric term solutions of $(\ref{mhypre})$ with representation $(\ref{mjhypdef})$ for $j=0$. And by a similar reasoning we will show how to deduce the other bases of $m$-fold hypergeometric solutions for $j\in\llbracket 0,m-1\rrbracket$.

Observe that if we have an $m$-fold hypergeometric sequence $a_n$ starting with $a_0$ (by shift it is always possible to define the initial term by $a_0$ ), then using the representation $(\ref{mhypdef})$ the next term which can be computed from $a_{0}$ is $a_{m}$, and afterwards $a_{2m}$ \ldots, $a_{km},\ldots$. Thus if we set $s_{n}=a_{mn}$ then all terms computed from $s_0$ have their indices corresponding to multiples of $m$ for $a_n$. Moreover, since $a_{m\cdot(n+1)}/a_{m\cdot n}=s_{n+1}/s_{n}\in\mathbb{K}(n)$, $s_n$ is a hypergeometric term whose formula is the same as that of $a_{mn}$. Therefore we can update $(\ref{mhypre})$ accordingly so the algorithm in (\cite{BThyper}) can be applied to compute a basis of all hypergeometric term solution $s_n$ of an updated version of $(\ref{mhypre})$, which is nothing but the basis of all $m$-fold hypergeometric term solutions of $(\ref{mhypre})$ with representation $(\ref{mjhypdef})$ for $j=0$. 

This view of $m$-fold hypergeometric terms is our main idea. As explained, a basis of all $m$-fold hypergeometric term solutions of the $m$-fold holonomic RE $(\ref{mhypre})$ can be found by van Hoeij's algorithm provided the following crucial change of variable:

\begin{equation}
	\begin{cases}
		m\cdot k = n \\
		s_{k}  = a_{m\cdot k}
	\end{cases}.\label{changeofvar}
\end{equation}
This leads to a $1$-fold holonomic RE for $s_k$ which has hypergeometric term solutions because
\begin{equation}
	\dfrac{s_{k+1}}{s_k} = \dfrac{a_{mk+m}}{a_{mk}}=r(mk).
\end{equation}
The resulting RE is
\begin{equation}
	P_d(mk)\cdot s_{k+d} + P_{d-1}(mk)\cdot s_{k+(d-1)} + \cdots + P_1(mk)\cdot s_{k+1} + P_0(mk)\cdot s_{k}=0 \label{changeRE}.
\end{equation}

In the general case of $j\in\llbracket 0,m-1\rrbracket$, the bases of all $m$-fold hypergeometric term solutions of $(\ref{mhypre})$ with representation $(\ref{mjhypdef})$ are computed by the algorithm in (\cite{BThyper}) after the application of the change of variable 
\begin{equation}
	\begin{cases}
		m\cdot k + j = n, \\
		s_{k}  = a_{m\cdot k + j}
	\end{cases}~~ (0\leqslant j\leqslant m-1).\label{changeofvargen}
\end{equation}
This is because in $(\ref{mjhypdef})$, the $m$-fold hypergeometric term indices can always be seen as $m\cdot n + j$, $j\in\llbracket 0,m-1\rrbracket$. 

Let us apply this to an example. We consider the two-term recurrence relation of the Taylor coefficient of $\exp(z)\sin(z)$ which is a $4$-fold holonomic RE, so we are going to compute $4$-fold hypergeometric term solutions.

\noindent
\begin{minipage}[t]{8ex}\color{red}\bf
	\begin{verbatim}
		(%i1) 
	\end{verbatim}
\end{minipage}
\begin{minipage}[t]{\textwidth}\color{blue}
	\begin{verbatim}
		RE:FindRE(exp(z)*sin(z),z,a[n],2);
	\end{verbatim}
\end{minipage}
\definecolor{labelcolor}{RGB}{100,0,0}
\begin{equation}
	\displaystyle
	\parbox{10ex}{$\color{labelcolor}\mathrm{\tt (\%o1) }\quad $}
	\left( 1+n\right) \cdot \left( 2+n\right) \cdot \left( 3+n\right) \cdot \left( 4+n\right) \cdot {{a}_{n+4}}+4\cdot {{a}_{n}}=0\mbox{} \label{expsinRE}
\end{equation}

\noindent
\begin{minipage}[t]{8ex}\color{red}\bf
	\begin{verbatim}
		(%i2) 
	\end{verbatim}
\end{minipage}
\begin{minipage}[t]{\textwidth}\color{blue}
	\begin{verbatim}
		RE:subst(4*n,n,RE);
	\end{verbatim}
\end{minipage}
\definecolor{labelcolor}{RGB}{100,0,0}
\[\displaystyle
\parbox{10ex}{$\color{labelcolor}\mathrm{\tt (\%o2) }\quad $}
\left( 1+4\cdot n\right) \cdot \left( 2+4\cdot n\right) \cdot \left( 3+4\cdot n\right) \cdot \left( 4+4\cdot n\right) \cdot {{a}_{4\cdot n+4}}+4\cdot {{a}_{4\cdot n}}=0\mbox{}
\]

\noindent
\begin{minipage}[t]{8ex}\color{red}\bf
	\begin{verbatim}
		(%i3) 
	\end{verbatim}
\end{minipage}
\begin{minipage}[t]{\textwidth}\color{blue}
	\begin{verbatim}
		RE:subst([a[4*n]=s[n],a[4*n+4]=s[n+1]],RE);
	\end{verbatim}
\end{minipage}
\definecolor{labelcolor}{RGB}{100,0,0}
\[\displaystyle
\parbox{10ex}{$\color{labelcolor}\mathrm{\tt (\%o3) }\quad $}
\left( 1+4\cdot n\right) \cdot \left( 2+4\cdot n\right) \cdot \left( 3+4\cdot n\right) \cdot \left( 4+4\cdot n\right) \cdot {{s}_{n+1}}+4\cdot {{s}_{n}}=0\mbox{}
\]

\noindent
\begin{minipage}[t]{8ex}\color{red}\bf
	\begin{verbatim}
		(%i4) 
	\end{verbatim}
\end{minipage}
\begin{minipage}[t]{\textwidth}\color{blue}
	\begin{verbatim}
		HypervanHoeij(RE,s[n]);
	\end{verbatim}
\end{minipage}
\definecolor{labelcolor}{RGB}{100,0,0}
\[\displaystyle
\parbox{10ex}{$\color{labelcolor}\mathrm{\tt (\%o4) }\quad $}
\left\{\frac{{{\left( -1\right) }^{n}}\cdot {{4}^{n}}}{{{\left( \frac{1}{4}\right) }_{n}}\cdot {{\left( \frac{3}{4}\right) }_{n}}\cdot {{64}^{n}}\cdot \left( 2\cdot n\right) !}\right\}\mbox{}
\]
This set is a basis of all $4$-fold hypergeometric term solutions $(\ref{expsinRE})$ for $j=0$ in the representation $(\ref{mjhypdef})$. Similarly for the case $j=3$ we get the analogous basis

\noindent
\begin{minipage}[t]{8ex}\color{red}\bf
	\begin{verbatim}
		(%i5) 
	\end{verbatim}
\end{minipage}
\begin{minipage}[t]{\textwidth}\color{blue}
	\begin{verbatim}
		RE:FindRE(exp(z)*sin(z),z,a[n],2)$
	\end{verbatim}
\end{minipage}

\noindent
\begin{minipage}[t]{8ex}\color{red}\bf
	\begin{verbatim}
		(%i6) 
	\end{verbatim}
\end{minipage}
\begin{minipage}[t]{\textwidth}\color{blue}
	\begin{verbatim}
		RE:subst(4*n+3,n,RE);
	\end{verbatim}
\end{minipage}
\definecolor{labelcolor}{RGB}{100,0,0}
\[\displaystyle
\parbox{10ex}{$\color{labelcolor}\mathrm{\tt (\%o6) }\quad $}
\left( 4+4\cdot n\right) \cdot \left( 5+4\cdot n\right) \cdot \left( 6+4\cdot n\right) \cdot \left( 7+4\cdot n\right) \cdot {{a}_{4\cdot n+7}}+4\cdot {{a}_{4\cdot n+3}}=0\mbox{}
\]

\noindent
\begin{minipage}[t]{8ex}\color{red}\bf
	\begin{verbatim}
		(%i7) 
	\end{verbatim}
\end{minipage}
\begin{minipage}[t]{\textwidth}\color{blue}
	\begin{verbatim}
		RE:subst([a[4*n+3]=s[n],a[4*n+7]=s[n+1]],RE);
	\end{verbatim}
\end{minipage}
\definecolor{labelcolor}{RGB}{100,0,0}
\[\displaystyle
\parbox{10ex}{$\color{labelcolor}\mathrm{\tt (\%o7) }\quad $}
\left( 4+4\cdot n\right) \cdot \left( 5+4\cdot n\right) \cdot \left( 6+4\cdot n\right) \cdot \left( 7+4\cdot n\right) \cdot {{s}_{n+1}}+4\cdot {{s}_{n}}=0\mbox{}
\]

\noindent
\begin{minipage}[t]{8ex}\color{red}\bf
	\begin{verbatim}
		(%i8) 
	\end{verbatim}
\end{minipage}
\begin{minipage}[t]{\textwidth}\color{blue}
	\begin{verbatim}
		HypervanHoeij(RE,s[n]);
	\end{verbatim}
\end{minipage}
\definecolor{labelcolor}{RGB}{100,0,0}
\[\displaystyle
\parbox{10ex}{$\color{labelcolor}\mathrm{\tt (\%o8) }\quad $}
\left\{\frac{{{\left( -1\right) }^{n}}\cdot {{4}^{n}}}{{{\left( \frac{1}{4}\right) }_{n}}\cdot {{\left( \frac{3}{4}\right) }_{n}}\cdot \left( 32\cdot {{n}^{3}}+48\cdot {{n}^{2}}+22\cdot n+3\right) \cdot {{64}^{n}}\cdot \left( 2\cdot n\right) !}\right\}\mbox{}
\]

On the other hand, note that the $m$-fold holonomic RE case is the easiest part for the whole algorithm. Indeed, an arbitrary holonomic recurrence equation is not necessarily $m$-fold holonomic, $m\in\mathbb{N}_{\geqslant 2}$. It could have $m_1$-fold and $m_2$-fold hypergeometric term solutions with positive integers $m_1\neq m_2$. Therefore we should define what to do in the general case. 

Observe that without the shift that transforms an $m$-fold holonomic recurrence equation in the form $(\ref{mhypre})$, its general representation is given by
\begin{equation}
	P_{d} a_{n+k+md} + P_{d-1} a_{n+k+m(d-1)} + \cdots + P_{0} a_{n+k}=0, \label{mreg}
\end{equation}
where $k\in\llbracket 0, m-1\rrbracket$.

Let us consider the three following $3$-fold holonomic REs
\begin{align}
	& RE1:~{{P}_{1,3}}\cdot {{a}_{n+7}}+{{P}_{1,2}}\cdot {{a}_{n+4}}+{{P}_{1,1}}\cdot {{a}_{n+1}}=0,\nonumber\\
	& RE2:~{{P}_{2,4}}\cdot {{a}_{n+11}}+{{P}_{2,3}}\cdot {{a}_{n+8}}+{{P}_{2,2}}\cdot {{a}_{n+5}}+{{P}_{2,1}}\cdot {{a}_{n+2}}=0,\nonumber\\
	& RE3:~{{P}_{3,4}}\cdot {{a}_{n+13}}+{{P}_{3,3}}\cdot {{a}_{n+10}}+{{P}_{3,2}}\cdot {{a}_{n+7}}+{{P}_{3,1}}\cdot {{a}_{n+4}}=0.  \label{3foldh}
\end{align}

\begin{itemize}
	\item The difference between an index of the indeterminate ($a_n$) in $RE1$ and another of its index taken in $RE2$ is always not divisible $3$. In this case we say that $RE1$ and $RE2$ are $3$-fold distinct.
	\item The difference between an index of the indeterminate ($a_n$) in $RE1$ and another of its index taken in $RE2$ is always a multiple of $3$. that $RE1$ and $RE3$ are $3$-fold equivalent.
\end{itemize}
More generally we have the following definitions.
\begin{definition}
	Let $m\in\mathbb{N}$,
	\begin{equation}
		RE_1:~	P_{d_1} a_{n+k_1+md_1} + P_{d_1-1} a_{n+k_1+m(d_1-1)} + \cdots + P_{0_1} a_{n+k_1}=0, \label{mregi}
	\end{equation}
	and 
	\begin{equation}
		RE_2:~	P_{d_2} a_{n+k_2+md_2} + P_{d_2-1} a_{n+k_2+m(d_2-1)} + \cdots + P_{0_2} a_{n+k_2}=0 \label{mregj}
	\end{equation}
	be two $m$-fold holonomic recurrence equations.
	\begin{itemize}
		\item We say that $RE_1$ and $RE_2$ are $m$-fold distinct holonomic equations if $k_2-k_1$ is not divisible by $m$.
		\item We say that $RE_1$ and $RE_2$ are $m$-fold equivalent holonomic equations if $k_2-k_1$ is divisible by $m$.
	\end{itemize}
\end{definition}

An immediate consequence of these definitions is that linear combinations of $m$-fold equivalent holonomic REs always give $m$-fold holonomic recurrence equations whereas linear combinations of $m$-fold distinct holonomic REs are never $m$-fold holonomic. For example, let us sum $RE1$ and $RE3$ from $(\ref{3foldh})$. This yields the following $3$-fold holonomic RE
\begin{equation}
	RE1 + RE3 :~ {{P}_{3,4}}\cdot {{a}_{n+13}}+{{P}_{3,3}}\cdot {{a}_{n+10}}+\left({{P}_{1,3}}+{{P}_{3,2}}\right)\cdot {{a}_{n+7}}+\left({{P}_{1,2}} +{{P}_{3,1}}\right)\cdot {{a}_{n+4}}+{{P}_{1,1}}\cdot {{a}_{n+1}}=0.
\end{equation}

The whole algorithm is based on the following fundamental theorem from which the general approach to compute $m$-fold hypergeometric term solutions of any given holonomic recurrence equation is deduced.

\begin{theorem}[Structure of holonomic REs having $m$-fold hypergeometric term solutions]\label{mfoldtheo} Let $m\in\mathbb{N}$, $\mathbb{K}$ a field of characteristic zero, and $h_n$ be an $m$-fold hypergeometric term which is not $u$-fold hypergeometric over $\mathbb{K}$ for all positive integers $u<m$. Then $h_n$ is a solution of a given holonomic recurrence equation, if that equation can be written as a linear combination of $m$-fold holonomic recurrence equations. When this is the case, $h_n$ is moreover solution of each of the $m$-fold distinct holonomic recurrence equations of that linear combination.
\end{theorem}
\begin{proof}
	Let $h_n$ be an $m$-fold hypergeometric term solution of the recurrence equation
	\begin{equation}
		P_d a_{n+d} + P_{d-1}a_{n+d-1} + \cdots + P_0 a_n = 0,~~ d>m, ~P_d\cdot P_0\neq 0. \label{RE2}
	\end{equation}
	It suffices to show that for any non-zero term $P_ja_{n+j}$ in $(\ref{RE2})$, there exists another summand, say $P_ia_{n+i}$, such that $m$ divides $j-i$. Indeed, by summing $m$-fold holonomic REs we are sure that for each summand appearing on the left-hand side of the sum there must exist another summand whose index differs from the one of that summand by a multiple of $m$. 
	
	We proceed by contradiction. Assume there exists a non-zero term $P_ja_{n+j}$ in $(\ref{RE2})$ such that any other summand $P_ia_{n+i}$, $i\neq j$ does not verify that $m$ divides $j-i$. Since $h_n$ is a non-zero solution, we can divide the equation by $h_{n+j}$ and write
	\begin{equation}
		- P_j = \sum_{\underset{i\neq j}{i=0}}^{d} P_i \cdot \dfrac{h_{n+i}}{h_{n+j}}. \label{rhsproof}
	\end{equation}
	
	The situation is now two-fold:
	\begin{itemize}
		\item For $i$ verifying $|i-j|<m$, for each corresponding summand $P_i \cdot h_{n+i}/h_{n+j}$ on the right-hand side of $(\ref{rhsproof})$, the fact that $m$ does not divide $j-i$ implies that $ h_{n+i}/h_{n+j}\notin\mathbb{K}(n)$ since by assumption $h_n$ is an $m$-fold hypergeometric term over $\mathbb{K}$ that is not $u$-fold hypergeometric for all integers $u<m$. Therefore the whole term $P_i \cdot  h_{n+i}/h_{n+j}\notin\mathbb{K}(n)$.
		\item For $i$ verifying $|i-j|>m$, for each corresponding summand $P_i \cdot  h_{n+i}/h_{n+j}$ on the right-hand side of $(\ref{rhsproof})$, we have two possibilities:
		\begin{itemize}
			\item either $ h_{n+i}/h_{n+j}\notin\mathbb{K}(n)$ and we have the same conclusion as in the previous case;
			\item or $ h_{n+i}/h_{n+j}\in\mathbb{K}(n)$, but in this case since $m$ does not divide $j-i$, this implies that $h_n$ is not an $m$-fold hypergeometric term and we get a contradiction.
		\end{itemize}
	\end{itemize}
	Thus the identity $(\ref{rhsproof})$ is valid only if all the summands on its right-hand side do not belong to $\mathbb{K}(n)$. Therefore $(\ref{rhsproof})$ holds if and only if
	$$ \sum_{\underset{i\neq j}{i=0}}^{d} P_i \cdot \dfrac{h_{n+i}}{h_{n+j}} \notin \mathbb{K}(n),$$
	since similarly as in the proof of Lemma \ref{lemmamfold1} we know that the non-rationality is caused by the linear independency. However, the left-hand side $P_j(n)\in\mathbb{K}[n]\subset \mathbb{K}(n)$. Hence we obtain a contradiction.
	
	Let us now prove the second part of the theorem. Since the multiplication of a holonomic recurrence equation by a polynomial does not affect the computation of its $m$-fold hypergeometric term solutions, the linear combination of $m$-fold holonomic REs can always be considered as a sum of $m$-fold holonomic REs. Therefore it is enough to show that an $m$-fold hypergeometric term solution of a sum of $m$-fold holonomic recurrence equations is a solution of each of the involved $m$-fold distinct holonomic recurrences.
	
	The sum of $M$ $m$-fold holonomic recurrence equations, $M\in\mathbb{N}$, can be written as
	\begin{equation}
		\sum_{j=1}^{M}RE_j(a_n)=\sum_{j=1}^{M} \left( P_{d_j} a_{n+k_j+md_j} + P_{d_j-1} a_{n+k_j+m(d_j-1)} + \cdots + P_{0_j} a_{n+k_j} \right)=0, \label{summholo}
	\end{equation}
	where $k_j\in\llbracket 0, m-1\rrbracket$, and $P_{d_j}\cdot P_{0_j}\neq 0, j\in\llbracket 1, M\rrbracket$.
	
	If $M=1$, then $(\ref{summholo})$ is an $m$-fold holonomic recurrence equation and $h_n$ is an $m$-fold hypergeometric term solution of it.
	
	We assume now that $M\geqslant 2$ and that there are at least two $m$-fold distinct holonomic recurrence equations in $(\ref{summholo})$. Note that if the $M$ $m$-fold holonomic REs are $m$-fold equivalent then the situation is similar to the case $M=1$ since every linear combination of $m$-fold equivalent holonomic REs is an $m$-fold holonomic RE. 
	
	Now suppose that $h_n$ is not solution of $RE_{j_1}$ in $(\ref{summholo})$, $j_1\in\llbracket 1, M \rrbracket$, then given that $\sum_{j=1}^{M}RE_j(h_n)=0$, there must be at least one second $m$-fold holonomic recurrence equation $RE_{j_2}$, $j_2\in\llbracket 1, M \rrbracket$, $m$-fold distinct with $RE_{j_1}$ such that $RE_{j_2}(h_n)\neq 0$. Without loss of generality, we consider that $RE_{j_2}$ is the only second $m$-fold holonomic RE with these properties. Of course, if $RE_{j_1}(h_n)\neq 0$ and $RE_{j_1}(h_n) +RE_{j_2}(h_n)=0$ then $RE_{j_2}(h_n)\neq 0$. Thus, we have
	\begin{equation}
		\begin{cases}
			RE_{j_1}(h_n)\neq 0\\
			RE_{j_2}(h_n)\neq 0\\
			RE_{j_1}(h_n) + RE_{j_2}(h_n) = 0
		\end{cases}  .\label{EiEj}
	\end{equation}
	The fact that the $m$-fold holonomic recurrence equations $RE_{j_1}$ and $RE_{j_2}$ are $m$-fold distinct implies that $k_{j_1} - k_{j_2}$ is not a multiple of $m$.  
	
	Using $(\ref{EiEj})$, after substitution of $h_n$ in the sum of the equations and division by $h_{n+k_{j_1}+d_{j_1}m}$, we deduce that
	\begin{equation}
		-P_{d_{j_1}} = \sum_{e_{j_1}=0_{j_1}}^{d_{j_1}-1} P_{e_{j_1}}\dfrac{h_{n+k_{j_1}+e_{j_1}m}}{h_{n+k_{j_1}+d_{j_1}m}}  + \sum_{e_{j_2}=0_{j_2}}^{d_{j_2}} P_{e_{j_2}}\dfrac{h_{n+k_{j_2}+e_{j_2}m}}{h_{n+k_{j_1}+d_{j_1}m}} = S_{j_1} + S_{j_1,j_2},
	\end{equation}
	which is equivalent to
	\begin{equation}
		- P_{d_{j_1}} - S_{j_1} = S_{j_1,j_2}. \label{endproof}
	\end{equation}
	All the summands of $S_{j_1}$ belong to $\mathbb{K}(n)$ since $h_n$ is an $m$-fold hypergeometric sequence and the corresponding index differences 
	$$n+k_{j_1}+e_{j_1}m - (n+k_{j_1}+d_{j_1}m) = m\cdot (e_{j_1}-d_{j_1})$$
	are multiples of $m$. However, for $S_{j_1,j_2}$ the index differences
	$$ n+k_{j_2}+e_{j_2}m - (n+k_{j_1}+d_{j_1}m) = k_{j_2}-k_{j_1} + m\cdot (e_{j_2}-d_{j_1}) $$
	are not multiples of $m$. Therefore by the same argument used in the first part of the proof we deduce that $S_{j_1,j_2}\notin \mathbb{K}(n)$.
	Thus $(\ref{endproof})$ holds if and only if $- P_{d_{j_1}} - S_{j_1} \in\mathbb{K}(n)$ and $S_{j_1,j_2}\notin \mathbb{K}(n)$. Therefore we get a contradiction.
\end{proof}

From this theorem, given $m\in\mathbb{N}$, we are now sure to compute a basis of all $m$-fold hypergeometric term solutions of a given holonomic recurrence equation by splitting it into the sum of $m$-fold distinct holonomic recurrence equations and use van Hoeij's algorithm or the algorithm in (\cite{BThyper}) to solve these holonomic REs provided the change of variable $(\ref{changeofvargen})$.

Note that since we compute $m$-fold hypergeometric terms as elements of a basis of all $m$-fold hypergeometric term solutions of holonomic REs, an $m$-fold hypergeometric term is solution of two given holonomic REs if it is linearly dependent to an element of the basis of all $m$-fold hypergeometric term solutions of each of these holonomic REs. Therefore, the solutions sought are built by all the linearly dependent $m$-fold hypergeometric term solutions of each involved $m$-fold distinct holonomic REs. Note that the computation of $m$-fold hypergeometric term solutions with representation $(\ref{mjhypdef})$ for $j=0$ of each $m$-fold holonomic RE
\begin{equation}
	P_{d_{i}} a_{n+k_{i}+md_{i}} + P_{d_{i}-1} a_{n+k_{i}+m(d_{i}-1)} + \cdots + P_{0_{i}} a_{n+k_{i}}=0, \label{rei}
\end{equation}
is done after writing it in the form $(\ref{mhypre})$. Thus $(\ref{rei})$ is transformed as
\begin{equation}
	P_{d_{i}}(n-k_i) a_{n+md_{i}} + P_{d_{i}-1}(n-k_i) a_{n+m(d_{i}-1)} + \cdots + P_{0_{i}}(n-k_i) a_{n}=0.
\end{equation}

Let us take as an example the holonomic RE satisfied by the Taylor coefficients of \linebreak $\exp(z)+\cos(z)$.

\noindent
\begin{minipage}[t]{8ex}\color{red}\bf
	\begin{verbatim}
		(%i1) 
	\end{verbatim}
\end{minipage}
\begin{minipage}[t]{\textwidth}\color{blue}
	\begin{verbatim}
		FindRE(cos(z)+exp(z),z,a[n]);
	\end{verbatim}
\end{minipage}
\definecolor{labelcolor}{RGB}{100,0,0}
\[\displaystyle
\parbox{10ex}{$\color{labelcolor}\mathrm{\tt (\%o1) }\quad $}\hspace{-0.4cm}
\left( 1+n\right) \cdot \left( 2+n\right) \cdot \left( 3+n\right) \cdot {{a}_{n+3}}-\left( 1+n\right) \cdot \left( 2+n\right) \cdot {{a}_{n+2}}+\left( 1+n\right) \cdot {{a}_{n+1}}-{{a}_{n}}=0\mbox{}
\]
This is a linear combination of two $2$-fold distinct holonomic REs, namely
$$RE1:~ \left( 1+n\right) \cdot \left( 2+n\right) \cdot \left( 3+n\right) \cdot {{a}_{n+3}}+\left( 1+n\right) \cdot {{a}_{n+1}}=0,  $$
and 
$$RE2:~\left( -1-n\right) \cdot \left( 2+n\right) \cdot {{a}_{n+2}}-{{a}_{n}}=0.$$
Only $RE1$ has to be transformed as its trailing term is not of order $0$. This yields
$$RE11:~ n\cdot \left( 1+n\right) \cdot \left( 2+n\right) \cdot {{a}_{n+2}}+n\cdot {{a}_{n}}=0.$$
From this one easily sees that the given holonomic RE has $2$-fold hypergeometric term solutions since we get two two-term recurrence relations that are linearly dependent: 
$$-n\cdot RE2 = RE11.$$

Remember that there is no need to use all the $m$ changes of variable of $(\ref{changeofvargen})$ because as we explained earlier, once one succeeds in computing a basis of $m$-fold hypergeometric term solutions corresponding to the representation $(\ref{mjhypdef})$ for a fixed $j\in\llbracket 0, m-1\rrbracket$, the other ones can be computed in a similar way. This will be used for power series computations in order to consider all possible linear combinations of hypergeometric type series of type $m$.

This result is a consequence of observing $m$-fold hypergeometric terms as sequences whose indices are taken in $m\cdot \mathbb{Z}+j$, $j\in\llbracket0, m-1 \rrbracket$. Commonly this notion is used according to its definition for the set of integers $\mathbb{Z}$ which is generally the chosen set of indices. In this case two terms of a sequence are said to be consecutive if their index difference is $1$ or $-1$. Such a definition is more useful for hypergeometric terms since it allows van Hoeij's algorithm to look for hypergeometric term solutions of holonomic REs in such a way that the ratio of two consecutive terms is a rational function over the considered field. For the $m$-fold case, however, one rather needs to consider $m\mathbb{Z}$ as the set of indices so that the computation of $m$-fold hypergeometric term ($m\geqslant 2$) solutions of holonomic REs is done analogously to the one of hypergeometric terms. In this situation one could say that two terms of an $m$-fold sequence are consecutive if the difference of their indices is $m$ or $-m$.

To compute the basis of all $m$-fold hypergeometric term solutions of a given holonomic RE, the algorithm proceeds by iteration up to the order of the RE. Note, however, that more often the number of cases to be considered is much smaller than the order of the given RE. For example, the recurrence equation

\noindent
\begin{minipage}[t]{8ex}\color{red}\bf
	\begin{verbatim}
		(%i1) 
	\end{verbatim}
\end{minipage}
\begin{minipage}[t]{\textwidth}\color{blue}
	\begin{verbatim}
		RE:FindRE(sin(z^3)^3,z,a[n]);
	\end{verbatim}
\end{minipage}
\definecolor{labelcolor}{RGB}{100,0,0}
\[\displaystyle
\parbox{10ex}{$\color{labelcolor}\mathrm{\tt (\%o1) }\quad $}
\left( n-8\right) \cdot \left( n-5\right) \cdot \left( n-2\right) \cdot \left( 1+n\right) \cdot {{a}_{n+1}}+90\cdot \left( n-8\right) \cdot \left( n-5\right) \cdot {{a}_{n-5}}+729\cdot {{a}_{n-11}}=0\mbox{}
\]
is a $2$-fold, $3$-fold and $6$-fold holonomic RE of order $12$. It is straightforward to see that all the other cases do not lead to a solution since the recurrence equation cannot be written as a sum of $m$-fold distinct holonomic REs for $m\notin\{1,2,3,6\}$. 

Our algorithm to compute $m$-fold hypergeometric term solutions of a given holonomic RE, called mfoldHyper, is the following.

\begin{algorithm}[h!]
	\caption{mfoldHyper: $m$-fold hypergeometric term solutions of holonomic recurrence equation of order $d\in\mathbb{N}$}\label{mymfold}
	\begin{algorithmic}[3]
		\Require A holonomic recurrence equation
		\begin{equation}
			P_d a_{n+d} + P_{d-1}a_{n+d-1} + \cdots + P_0 a_n = 0,~~ d>m, ~P_d\cdot P_0\neq 0 \label{RE3}
		\end{equation}
		\Ensure A basis (incomplete form) of all $m$-fold hypergeometric term solutions of $(\ref{RE3})$.
		
		\begin{enumerate}
			\item Set $H=\{\}$.
			\item Use the algorithm in (\cite{BThyper}) to find the basis, say $H_1$, of all hypergeometric term solutions of $(\ref{RE3})$. If $H_1\neq\emptyset$, then add $[1,H_1]$ to $H$.
			\item For $2\leqslant m \leqslant d$ do:
			\begin{itemize}
				\item[(a)] Extract the following $m$-fold holonomic recurrence equations from $(\ref{RE3})$ and construct the system
			\end{itemize}
	    \end{enumerate}
		
		 	\algstore{pause4}
		 \end{algorithmic}
		\end{algorithm}
		\clearpage 
		
		\begin{algorithm}[H]
		\ContinuedFloat
		\caption{mfoldHyper}
		\begin{algorithmic}[3]
			\algrestore{pause4}	
			\State     	
		\begin{enumerate}
			\setcounter{enumi}{2}
			\item
			\begin{itemize}
				\item[(a)]
				\begin{equation}
					\hspace{-1.5cm}
					\begin{cases}
						P_0(n)\cdot a_n + P_m(n)\cdot a_{n+m}+ \cdots + P_{m\cdot \lfloor \frac{d}{m} \rfloor}(n)\cdot a_{n+m\cdot \lfloor \frac{d}{m} \rfloor}=0\\
						P_1(n)\cdot a_{n+1} + P_{m+1}(n)\cdot a_{n+m+1}+ \cdots + P_{m\cdot \lfloor \frac{d}{m} \rfloor + 1}(n)\cdot a_{n+m\cdot \lfloor \frac{d}{m} \rfloor+1}=0\\
						\ldots\\
						P_{m-1}(n)\cdot a_{n+m-1} + P_{2m-1}(n)\cdot a_{n+2m-1}+ \cdots + P_{m\cdot \lfloor \frac{d}{m} \rfloor+m-1}(n)\cdot a_{n+m\cdot \lfloor \frac{d}{m} \rfloor + m-1}=0\\
						\\
					\end{cases} \label{mhypeq},
				\end{equation}
				assuming $P_j(n)=0$ for $j>d$.
				\item[(b)] If there exists a holonomic RE with only one non-zero polynomial coefficient in $(\ref{mhypeq})$, then stop and go back to step $3$.(a) for $m+1$.
				\item[(c)] Shift all the $m$-fold holonomic recurrence equations in $(\ref{mhypeq})$ so that the order of the trailing term equals $0$.
				\item[(d)] Apply the change of variable $(\ref{changeofvar})$ for each $m$-fold holonomic recurrence equation.
				\item[(e)] Compute a basis of all hypergeometric term solutions $s_k$ as defined in $(\ref{changeofvar})$ for $(\ref{changeRE})$ of each resulting holonomic recurrence equation with the algorithm in (\cite{BThyper}).
				\item[(f)] Construct the set $H_m$ of hypergeometric terms which are each linearly dependent to one term in each of the $m$ computed bases in step $3$.(d).
				\item[(g)] If $H_m\neq\emptyset$ then add  $[m, H_m]$ in $H$.
			\end{itemize}
		   \item Return $H$.
		%\end{itemize}
	\end{enumerate}
	\end{algorithmic}
\end{algorithm}

We implemented mfoldHyper in Maxima as \textit{mfoldHyper(RE,a[n],[m,j])}, by default \textit{[m,j]} is an empty list. In that default case each list of $m$-fold hypergeometric term solutions, say $[m,[h_{1,m},h_{2,m},\ldots]]$, contains "simple formulas" of hypergeometric terms corresponding to $j=0$ in $(\ref{mjhypdef})$. Once we know that there are some $m$-fold hypergeometric term solutions for particular $m\in\mathbb{N}$, the algorithm can be called as \textit{mfoldHyper(RE,a[n],m,j)} for $0\leqslant j < m$ to get the solutions in their other representations.

Let us now apply the algorithm to some examples. We hide the recurrence equations for space saving purposes. All these computations can be done with our package FPS currently available as third-party Maxima package on Github.

\noindent
\begin{minipage}[t]{8ex}\color{red}\bf
	\begin{verbatim}
		(%i1) 
	\end{verbatim}
\end{minipage}
\begin{minipage}[t]{\textwidth}\color{blue}
	\begin{verbatim}
		RE:FindRE(atan(z)+exp(z),z,a[n])$
	\end{verbatim}
\end{minipage}

\noindent
\begin{minipage}[t]{8ex}\color{red}\bf
	\begin{verbatim}
		(%i2) 
	\end{verbatim}
\end{minipage}
\begin{minipage}[t]{\textwidth}\color{blue}
	\begin{verbatim}
		mfoldHyper(RE,a[n]);
	\end{verbatim}
\end{minipage}
\definecolor{labelcolor}{RGB}{100,0,0}
\[\displaystyle
\parbox{10ex}{$\color{labelcolor}\mathrm{\tt (\%o2) }\quad $}
\left[\left[1,\left\{\frac{1}{n!}\right\}\right],\left[2,\left\{\frac{{{\left( -1\right) }^{n}}}{n}\right\}\right]\right]\mbox{}
\]

Sometimes computations may involve algebraic extension fields of $\mathbb{Q}$, the syntax is \textit{mfoldHyper(RE,a[n],[K])} for the two possible values \textit{K=C} or \textit{K=Q} (default value). To ask for specific $m$-fold hypergeometric term solutions the syntax is \textit{mfoldHyper(RE,a[n],[K,m,j])}.

\noindent
\begin{minipage}[t]{8ex}\color{red}\bf
	\begin{verbatim}
		(%i3) 
	\end{verbatim}
\end{minipage}
\begin{minipage}[t]{\textwidth}\color{blue}
	\begin{verbatim}
		RE:FindRE(log(1+z+z^2)+cos(z),z,a[n])$
	\end{verbatim}
\end{minipage}

\noindent
\begin{minipage}[t]{8ex}\color{red}\bf
	\begin{verbatim}
		(%i4) 
	\end{verbatim}
\end{minipage}
\begin{minipage}[t]{\textwidth}\color{blue}
	\begin{verbatim}
		mfoldHyper(RE,a[n],C);
	\end{verbatim}
\end{minipage}
\definecolor{labelcolor}{RGB}{100,0,0}
\[\displaystyle
\parbox{10ex}{$\color{labelcolor}\mathrm{\tt (\%o4) }\quad $}
\left[\left[1,\left\{\frac{{{\left( \frac{-1-\sqrt{3}\cdot i}{2}\right) }^{n}}}{n},\frac{{{\left( \frac{\sqrt{3}\cdot i-1}{2}\right) }^{n}}}{n},\frac{{{\left( -i\right) }^{n}}}{n!},\frac{{{\left( -1\right) }^{\frac{n}{2}}}}{n!}\right\}\right],\left[2,\left\{\frac{{{\left( -1\right) }^{n}}}{\left( 2\cdot n\right) !}\right\}\right]\right]\mbox{}
\]
where the obtained $2$-fold hypergeometric term is the coefficient of the hypergeometric type series of $\cos(z)$.

\noindent
\begin{minipage}[t]{8ex}\color{red}\bf
	\begin{verbatim}
		(%i5) 
	\end{verbatim}
\end{minipage}
\begin{minipage}[t]{\textwidth}\color{blue}
	\begin{verbatim}
		declare(q1,constant)$
	\end{verbatim}
\end{minipage}

\noindent
\begin{minipage}[t]{8ex}\color{red}\bf
	\begin{verbatim}
		(%i6) 
	\end{verbatim}
\end{minipage}
\begin{minipage}[t]{\textwidth}\color{blue}
	\begin{verbatim}
		declare(q2,constant)$
	\end{verbatim}
\end{minipage}

\noindent
\begin{minipage}[t]{8ex}\color{red}\bf
	\begin{verbatim}
		(%i7) 
	\end{verbatim}
\end{minipage}
\begin{minipage}[t]{\textwidth}\color{blue}
	\begin{verbatim}
		RE:FindRE(1/((q1-z^2)*(q2-z^3)),z,a[n])$
	\end{verbatim}
\end{minipage}

\noindent
\begin{minipage}[t]{8ex}\color{red}\bf
	\begin{verbatim}
		(%i8) 
	\end{verbatim}
\end{minipage}
\begin{minipage}[t]{\textwidth}\color{blue}
	\begin{verbatim}
		mfoldHyper(RE,a[n],C);
	\end{verbatim}
\end{minipage}
\definecolor{labelcolor}{RGB}{100,0,0}
\[\displaystyle
\parbox{10ex}{$\color{labelcolor}\mathrm{\tt (\%o8) }\quad $}\hspace{-0.7cm}
\left[\left[1,\left\{{{\left( -\sqrt{\frac{1}{\mathit{q1}}}\right) }^{n}},{{\left( \frac{1}{\mathit{q1}}\right) }^{\frac{n}{2}}},{{\left( \frac{\sqrt{3}\cdot i\cdot {{\left( \frac{1}{\mathit{q2}}\right) }^{\frac{1}{3}}}-{{\left( \frac{1}{\mathit{q2}}\right) }^{\frac{1}{3}}}}{2}\right) }^{n}},{{\left( \frac{1}{\mathit{q2}}\right) }^{\frac{n}{3}}}\right\}\right],\left[2,\left\{{{\left( \frac{1}{\mathit{q1}}\right) }^{n}}\right\}\right],\left[3,\left\{{{\left( \frac{1}{\mathit{q2}}\right) }^{n}}\right\}\right]\right]\mbox{}\]

For these previous examples, the current Maple \textit{convert/FormalPowerSeries} yields complicated power series representations because the above $m$-fold hypergeometric terms, $m\geqslant 2$, are not found. Next we compute the power series coefficients of some expressions for which \textit{convert/FormalPowerSeries} does not find representations.

\noindent
\begin{minipage}[t]{8ex}\color{red}\bf
	\begin{verbatim}
		(%i9) 
	\end{verbatim}
\end{minipage}
\begin{minipage}[t]{\textwidth}\color{blue}
	\begin{verbatim}
		RE:FindRE(exp(z^2)+cos(z^2),z,a[n])$
	\end{verbatim}
\end{minipage}

\noindent
\begin{minipage}[t]{8ex}\color{red}\bf
	\begin{verbatim}
		(%i10) 
	\end{verbatim}
\end{minipage}
\begin{minipage}[t]{\textwidth}\color{blue}
	\begin{verbatim}
		mfoldHyper(RE,a[n]);
	\end{verbatim}
\end{minipage}
\definecolor{labelcolor}{RGB}{100,0,0}
\[\displaystyle
\parbox{10ex}{$\color{labelcolor}\mathrm{\tt (\%o10) }\quad $}
\left[\left[2,\left\{\frac{1}{n!}\right\}\right],\left[4,\left\{\frac{{{\left( -1\right) }^{n}}}{\left( 2\cdot n\right) !}\right\}\right]\right]\mbox{}
\]

\noindent
\begin{minipage}[t]{8ex}\color{red}\bf
	\begin{verbatim}
		(%i11) 
	\end{verbatim}
\end{minipage}
\begin{minipage}[t]{\textwidth}\color{blue}
	\begin{verbatim}
		RE:FindRE(cosh(z^3)+sin(z^2),z,a[n])$
	\end{verbatim}
\end{minipage}

\noindent
\begin{minipage}[t]{8ex}\color{red}\bf
	\begin{verbatim}
		(%i12) 
	\end{verbatim}
\end{minipage}
\begin{minipage}[t]{\textwidth}\color{blue}
	\begin{verbatim}
		mfoldHyper(RE,a[n]);
	\end{verbatim}
\end{minipage}
\definecolor{labelcolor}{RGB}{100,0,0}
\[\displaystyle
\parbox{10ex}{$\color{labelcolor}\mathrm{\tt (\%o12) }\quad $}
\left[\left[3,\left\{\frac{1}{n!},\frac{{{\left( -1\right) }^{n}}}{n!}\right\}\right],\left[4,\left\{\frac{{{\left( -1\right) }^{n}}}{\left( 2\cdot n\right) !}\right\}\right],\left[6,\left\{\frac{1}{\left( 2\cdot n\right) !}\right\}\right]\right]\mbox{}
\]

\noindent
\begin{minipage}[t]{8ex}\color{red}\bf
	\begin{verbatim}
		(%i13) 
	\end{verbatim}
\end{minipage}
\begin{minipage}[t]{\textwidth}\color{blue}
	\begin{verbatim}
		RE:FindRE(asin(z^2)^2+acos(z),z,a[n])$
	\end{verbatim}
\end{minipage}

\noindent
\begin{minipage}[t]{8ex}\color{red}\bf
	\begin{verbatim}
		(%i13) 
	\end{verbatim}
\end{minipage}
\begin{minipage}[t]{\textwidth}\color{blue}
	\begin{verbatim}
		mfoldHyper(RE,a[n]);
	\end{verbatim}
\end{minipage}
\definecolor{labelcolor}{RGB}{100,0,0}
\[\displaystyle
\parbox{10ex}{$\color{labelcolor}\mathrm{\tt (\%o13) }\quad $}
\left[\left[2,\left\{\frac{{{4}^{n}}\cdot {{n!}^{2}}}{{{n}^{2}}\cdot \left( 2\cdot n\right) !}\right\}\right],\left[4,\left\{\frac{{{4}^{n}}\cdot {{n!}^{2}}}{{{n}^{2}}\cdot \left( 2\cdot n\right) !}\right\}\right]\right]\mbox{}
\]

\noindent
\begin{minipage}[t]{8ex}\color{red}\bf
	\begin{verbatim}
		(%i14) 
	\end{verbatim}
\end{minipage}
\begin{minipage}[t]{\textwidth}\color{blue}
	\begin{verbatim}
		RE:FindRE(sqrt(sqrt(8*z^3+1)-1)+sqrt(7+13*z^4),z,a[n])$
	\end{verbatim}
\end{minipage}

\noindent
\begin{minipage}[t]{8ex}\color{red}\bf
	\begin{verbatim}
		(%i15) 
	\end{verbatim}
\end{minipage}
\begin{minipage}[t]{\textwidth}\color{blue}
	\begin{verbatim}
		mfoldHyper(RE,a[n]);
	\end{verbatim}
\end{minipage}
\definecolor{labelcolor}{RGB}{100,0,0}
\[\displaystyle
\parbox{10ex}{$\color{labelcolor}\mathrm{\tt (\%o15) }\quad $}
\left[\left[3,\left\{\frac{{{\left( \frac{1}{4}\right) }_{n}}\cdot {{\left( \frac{3}{4}\right) }_{n}}\cdot {{\left( -8\right) }^{n}}\cdot {{4}^{n}}}{\left( 4\cdot n-1\right) \cdot \left( 2\cdot n\right) !}\right\}\right],\left[4,\left\{ \frac{{{4}^{-4-n}}\cdot {{\left( -13\right) }^{n}}\cdot \left( 2\cdot n\right) !}{\left( 2\cdot n-1\right) \cdot {{7}^{n}}\cdot {{n!}^{2}}}\right\}\right]\right]\mbox{}
\]

\noindent
\begin{minipage}[t]{8ex}\color{red}\bf
	\begin{verbatim}
		(%i16) 
	\end{verbatim}
\end{minipage}
\begin{minipage}[t]{\textwidth}\color{blue}
	\begin{verbatim}
		RE:FindRE(sin(z^3)^3,z,a[n])$
	\end{verbatim}
\end{minipage}

\noindent
\begin{minipage}[t]{8ex}\color{red}\bf
	\begin{verbatim}
		(%i17) 
	\end{verbatim}
\end{minipage}
\begin{minipage}[t]{\textwidth}\color{blue}
	\begin{verbatim}
		mfoldHyper(RE,a[n]);
	\end{verbatim}
\end{minipage}
\definecolor{labelcolor}{RGB}{100,0,0}
\[\displaystyle
\parbox{10ex}{$\color{labelcolor}\mathrm{\tt (\%o17) }\quad $}
\left[\left[6,\left\{\frac{{{\left( -9\right) }^{n}}}{\left( 2\cdot n\right) !},\frac{{{\left( -1\right) }^{n}}}{\left( 2\cdot n\right) !}\right\}\right]\right]\mbox{}
\]

Let us now use our implementation for the computation of a specific representation of $m$-fold hypergeometric term solutions. In this case the user has to specify a value for $m$ and $j$ with $j\in\llbracket 0,m-1\rrbracket$.

\noindent
\begin{minipage}[t]{8ex}\color{red}\bf
	\begin{verbatim}
		(%i18) 
	\end{verbatim}
\end{minipage}
\begin{minipage}[t]{\textwidth}\color{blue}
	\begin{verbatim}
		RE:FindRE(asin(z)^2+log(1+z^5),z,a[n])$
	\end{verbatim}
\end{minipage}

\noindent
\begin{minipage}[t]{8ex}\color{red}\bf
	\begin{verbatim}
		(%i19) 
	\end{verbatim}
\end{minipage}
\begin{minipage}[t]{\textwidth}\color{blue}
	\begin{verbatim}
		mfoldHyper(RE,a[n],5,0);
	\end{verbatim}
\end{minipage}
\definecolor{labelcolor}{RGB}{100,0,0}
\[\displaystyle
\parbox{10ex}{$\color{labelcolor}\mathrm{\tt (\%o19) }\quad $}
\left\{\frac{{{\left( -1\right) }^{n}}}{2\cdot n}\right\}\mbox{}
\]

\noindent
\begin{minipage}[t]{8ex}\color{red}\bf
	\begin{verbatim}
		(%i20) 
	\end{verbatim}
\end{minipage}
\begin{minipage}[t]{\textwidth}\color{blue}
	\begin{verbatim}
		mfoldHyper(RE,a[n],5,3);
	\end{verbatim}
\end{minipage}
\definecolor{labelcolor}{RGB}{100,0,0}
\[\displaystyle
\parbox{10ex}{$\color{labelcolor}\mathrm{\tt (\%o20) }\quad $}
\left\{\frac{{{\left( -1\right) }^{n}}}{2\cdot \left( 5\cdot n+3\right) }\right\}\mbox{}
\]

\noindent
\begin{minipage}[t]{8ex}\color{red}\bf
	\begin{verbatim}
		(%i21) 
	\end{verbatim}
\end{minipage}
\begin{minipage}[t]{\textwidth}\color{blue}
	\begin{verbatim}
		mfoldHyper(RE,a[n],2,1);
	\end{verbatim}
\end{minipage}
\definecolor{labelcolor}{RGB}{100,0,0}
\[\displaystyle
\parbox{10ex}{$\color{labelcolor}\mathrm{\tt (\%o21) }\quad $}
\left\{\frac{\left( 2\cdot n\right) !}{\left( 2\cdot n+1\right) \cdot {{4}^{n}}\cdot {{n!}^{2}}}\right\}\mbox{}
\]

Eventually, note that the existence of $m$-fold hypergeometric term solutions of a holonomic recurrence equation satisfied by the Taylor coefficients of a given expression does not necessarily guarantee that this expression represents a hypergeometric type function. For example, $\arctan(z)\cdot \cos(z)$ yields a recurrence equations satisfied by the coefficients of $\cos(z)$.

\noindent
\begin{minipage}[t]{8ex}\color{red}\bf
	\begin{verbatim}
	(%i22)
	\end{verbatim}
\end{minipage}
\begin{minipage}[t]{\textwidth}\color{blue}
	\begin{verbatim}
		RE:FindRE(atan(z)*cos(z),z,a[n])$
	\end{verbatim}
\end{minipage}

\noindent
\begin{minipage}[t]{8ex}\color{red}\bf
	\begin{verbatim}
	(%i23) 
\end{verbatim}
\end{minipage}
\begin{minipage}[t]{\textwidth}\color{blue}
	\begin{verbatim}
	mfoldHyper(RE,a[n]);
   \end{verbatim}
\end{minipage}
\[\displaystyle \parbox{10ex}{$\color{labelcolor}\mathrm{\tt (\%o23) }\quad $}
\left[\left[2\operatorname{,}\left\{\frac{{{\left( -1\right) }^{n}}}{\left( 2 n\right) \operatorname{!}}\right\}\right]\right]\mbox{}
\]
However, we know that the coefficient must be different. In the next section, by finding the linear combination of hypergeometric type power series we will be able to decide using some initial values whether a potential coefficient is the correct one.

\section{Hypergeometric type power series}\label{sec3}

The novelty of the results in this section could not be well understood without a precise definition of what we consider as hypergeometric type power series. 

\begin{definition}[Hypergeometric type power series]\label{def1}
Let $\mathbb{K}$ be a field of characteristic zero. For an expansion around $z_0\in\mathbb{K}$, a series $s(z)$ is said to be of hypergeometric type if it can be written as
   \begin{equation}
		s(z) := T(z) + \sum_{j=1}^{J} s_j(z),~ s_j=\sum_{n=n_{j,0}}^{\infty} a_{j,n}(z-z_0)^{n/k_j}\label{eq1}
	\end{equation}
	where $n$ is the summation variable, $T(z)\in\mathbb{K}[z,1/z,\ln(z)]$, $n_0\in\mathbb{Z}$, $J, k_j\in\mathbb{N}$, and $a_{j,n}$ is an $m_j$-fold hypergeometric term, $m_j\in\mathbb{N}$.
	
	Thus a hypergeometric type power series is a linear combination of Laurent-Puiseux series whose coefficients are $m$-fold hypergeometric terms. A hypergeometric function is a function that can be expanded as a hypergeometric type power series. $T$ is called the Laurent polynomial part of the expansion, and the $k_j$'s are its Puiseux numbers.
\end{definition}
The presence of $\ln(z)$ in a hypergeometric type expansion is justified by the solution of the underlying holonomic differential equation (see \cite{kauers2011}). The definition in (\cite{Koepf1992}) reduces to the case $T=0$ and $J\leqslant m$, where $m$ is the unique type\footnote{Originally the type was used to denote the value of $m$ for an $m$-fold hypergeometric term coefficient.} encountered in Definition \ref{def1}. With this new definition, we can define the type of the series $(\ref{eq1})$ as the tuple $(m_1,m_2,\ldots,m_J)$. Note, however, that we do not compute the coefficients as they appear in $(\ref{eq1})$, but instead for powers of the form $z^{m_j\cdot n + i}$, $0\leqslant i<m_j$ which is more suitable for the coefficients computed using mfoldHyper.

According to the general algorithm described in (\cite{Koepf1992}), we are at the final step of the power series computation procedure. Let us first recall what these steps are. For a given expression $f$, we compute its power series in the following way:
\begin{enumerate}
	\item Find a holonomic differential equation for $f$ using the algorithm in Subsection \ref{mymethodDE};
	\item Convert that holonomic DE into a holonomic recurrence equation satisfied by the power series coefficients of $f$ (see $(\ref{Algo3corr})$);
	\item Solve the obtained holonomic RE which in our case reduces to compute a basis of all the $m$-fold hypergeometric term solutions of that RE using algorithm mfoldHyper;
	\item If there are solutions, use initial values to find the linear combination of the resulting hypergeometric type power series that corresponds to the power series expansion of $f$, if such a linear combination is valid.
\end{enumerate}

Regarding Puiseux series, we will give a generalization of an idea in \cite[Section 5]{SC93}. We will see that computing Puiseux numbers $k_j$'s appearing in $(\ref{eq1})$ reduces to finding a number $k$ which can be defined as the Puiseux number of the corresponding hypergeometric type series. Once $k$ is found, we use the substitution $h(z)=f(z^k)$ to bring the situation to the Laurent series one, and finally divide the general power of the indeterminate $z$ in the obtained power series representation of $h(z)$ by $k$ to get the expansion sought. This is an intermediate step between the second and the third step above.

Our goal is to compute a representation of the form
\begin{equation}
	f(z) = T(z) + F(z),
\end{equation}
where $T(z)\in\mathbb{K}[z,\frac{1}{z},\log(z)]$ is a Laurent polynomial in the variable $z$ with coefficients in $\mathbb{K}[\log(z)]$, and $F(z)$ is a linear combination of hypergeometric type series. We mention that $T(z)$ is not uniquely determined but its determination will be made more precise by Lemma $\ref{polypart}$ and Algorithm $\ref{PolyPart}$.

\section{Finding the Puiseux number}\label{Puiseuxnbr}

Assume we are looking for a representation of the form

\begin{equation}
	f(z) = T(z) + F(z) := T(z) + \sum_{i=1}^{I} \sum_{n=0}^{\infty} s_{i_{n}} z^{(m_i\cdot n + j_i)/k_i} \label{Grep}
\end{equation}
where $m_i,k_i\in\mathbb{N}$, $j_i\in\llbracket 0,m_i-1\rrbracket$, $s_{i_n}$ is an $m_i$-fold hypergeometric term corresponding to $j=j_i$ in the representation $(\ref{mjhypdef})$, and $T(z)$ is an extra term whose computation will be explained in the next subsection.

It is enough to suppose that $F(z)$ in $(\ref{Grep})$ is the sum of two hypergeometric type series of type $m_1$ and $m_2$ since our development works similarly in the general situation. We have
\begin{equation}
	F(z) := \sum_{n=0}^{\infty} s_{1_n}z^{(m_1\cdot n + j_1)/k_1} + \sum_{n=0}^{\infty} s_{2_n}z^{(m_2\cdot n + j_2)/k_2}, \label{F2}
\end{equation}
with the same definitions in $(\ref{Grep})$ for $I=2$. For simplicity, we also assume that $k_1$ and $k_2$ are co-prime. This is to avoid the use of more variables since in particular this assumption implies that the least common multiple of $k_1$ and $k_2$ is $\lcm(k_1,k_2)=k_1\cdot k_2$. Substituting $z$ by $z^{\lcm(k_1,k_2)}$ in $(\ref{F2})$ gives
\begin{eqnarray}
	F(z^{\lcm(k_1,k_2)}) &=& \sum_{n=0}^{\infty} s_{1_n}z^{(m_1\cdot n + j_1)\cdot k_2} + \sum_{n=0}^{\infty} s_{2_n}z^{(m_2\cdot n + j_2)\cdot k_1}\label{F21} \\
	&=& \sum_{n\in k_2\cdot \left( m_1\cdot \mathbb{N}_{\geqslant 0} + j_1\right)}^{} a_{1_{\frac{n}{k_2}}}z^{n} + \sum_{n\in k_1\cdot \left( m_2\cdot \mathbb{N}_{\geqslant 0} + j_2\right)}^{} a_{2_{\frac{n}{k_1}}}z^{n},  \label{F22}
\end{eqnarray}
where $a_{i_n}$ is obtained from $s_{i_n}$ by the change of variable $(\ref{changeofvargen})$, $i\in\{1,2\}$. 

Observe that in $(\ref{F21})$ the powers of the indeterminate $z$ are integers. In general, the right-hand side of $(\ref{Grep})$ always gives a representation with integer powers when we substitute $z$ by $z^{\mu}$, for any positive multiple $\mu$ of $\lcm(k_1,k_2)$. Power series with integer powers are dealt with in other sections of this chapter. Thus our aim of determining the positive integers $k_i$, $i\in\llbracket 1, I\rrbracket$ in $(\ref{Grep})$ can be reduced in finding a positive multiple $\mu$ of $\lcm(k_1,\ldots, k_I)$ so that we can compute the power series of $f(z^{\mu})$ and substitute $z$ by $z^{1/\mu}$ in the obtained representation to get the one of $f(z)$.

By the general representation $(\ref{mhypdef})$ of an $m$-fold hypergeometric term, we know that there exist rational functions $r_1(n)$ and $r_2(n)$ such that
\begin{equation*}
	a_{1_{n+m_1}} = r_1(n) \cdot a_{1_n}   ~\text{ and }~ a_{2_{n+m_2}} = r_2(n) \cdot a_{2_n}, 
\end{equation*}
for the coefficients in $(\ref{F22})$. Therefore we can write
\begin{equation}
	a_{1_{\frac{n}{k_2}+m_1}} = r_1\left(\frac{n}{k_2}\right) \cdot a_{1_{\frac{n}{k_2}}}   ~\text{ and }~ a_{2_{\frac{n}{k_1}+m_2}} = r_2\left(\frac{n}{k_1}\right) \cdot a_{2_{\frac{n}{k_1}}}. \label{F23}
\end{equation}
where $\frac{n}{k_1}$ and $\frac{n}{k_2}$ are not necessarily integers.

To compute the holonomic recurrence equation of smallest order for the $m_1$-fold and the $m_2$-fold hypergeometric terms $a_{1_{\frac{n}{k_2}}}$ and $a_{2_{\frac{n}{k_1}}}$, we need to use the smallest integer $k$ such that $k\cdot \frac{n}{k_2}\in\mathbb{N}_{\geqslant 0}$ and $k\cdot \frac{n}{k_1}\in\mathbb{N}_{\geqslant 0}$. Thus $k=\lcm(k_1,k_2)$ and the obtained holonomic RE is of course compatible with the one computed using \textit{FindRE} for the input expression $F(z)$. From $(\ref{F23})$, substituting $n$ by $\lcm(k_1,k_2)\cdot n=k_1\cdot k_2\cdot n$ yields
\begin{equation}
	a_{1_{k_1\cdot n+m_1}} = r_1\left(k_1\cdot n\right) \cdot a_{1_{k_1\cdot n}}   ~\text{ and }~ a_{2_{k_2\cdot n+m_2}} = r_2\left(k_2\cdot n\right) \cdot a_{2_{k_2\cdot n}}. \label{F24}
\end{equation}

Since $a_{1_{k_1\cdot n+m_1}}$ and $a_{2_{k_2\cdot n+m_2}}$ are, respectively, $m_1$-fold and $m_2$-fold hypergeometric term solutions of a holonomic recurrence equation satisfied by the power series coefficients of $f(z)$, by algorithm mfoldHyper we know how such terms are computed using an algorithm to compute the equivalent hypergeometric terms $s_{i_n}$ such that 
$$\dfrac{s_{i_{n+1}}}{s_{i_n}} = \dfrac{a_{i_{n+m_i}}}{a_{i_n}}=r_i(k_i\cdot n),~i\in\{1,2\}.$$
By Petkov{\v{s}}ek's algorithm we know that ratios of hypergeometric term solutions of holonomic REs are built from monic factors of the corresponding trailing and leading polynomial coefficients. This implies in particular that some zeros and poles of $r_i(k_i\cdot n)$ are the roots of the shifted\footnote{Integer shift used in Petkov\v{s}ek's algorithm, see also Lemma \ref{polypart}} trailing and leading polynomial coefficient of the holonomic recurrence equation computed by \textit{FindRE} for the power series coefficients of $f(z)$, $i\in\{1,2\}$. Therefore by computing the least common multiple of all the trailing and leading polynomial coefficient rational root denominators of that RE we must obtain a multiple of $\lcm(k_1,k_2)$.

\begin{example}\item
	\textnormal{We consider the expression $f(z)=\exp(z^{3/4})+\sin(\sqrt{z})$}.
	
	\noindent
	\begin{minipage}[t]{8ex}\color{red}\bf
		\begin{verbatim}
			(%i1) 
		\end{verbatim}
	\end{minipage}
	\begin{minipage}[t]{\textwidth}\color{blue}
		\begin{verbatim}
			f:exp(z^(3/4)) + sin(sqrt(z))$
		\end{verbatim}
	\end{minipage}

	\noindent
	\begin{minipage}[t]{8ex}\color{red}\bf
		\begin{verbatim}
			(%i2) 
		\end{verbatim}
	\end{minipage}
	\begin{minipage}[t]{\textwidth}\color{blue}
		\begin{verbatim}
			RE:FindRE(f,z,a[n])$
		\end{verbatim}
	\end{minipage}
	\mbox{}\\
	\textnormal{We collect the coefficients with our Maxima function \textit{REcoeff}}
	
	\noindent
	\begin{minipage}[t]{8ex}\color{red}\bf
		\begin{verbatim}
			(%i3) 
		\end{verbatim}
	\end{minipage}
	\begin{minipage}[t]{\textwidth}\color{blue}
		\begin{verbatim}
			CoeffsRE: REcoeff(RE,a[n])$
		\end{verbatim}
	\end{minipage}
	\mbox{}\\
	\textnormal{The corresponding leading polynomial coefficient is}
	
	\noindent
	\begin{minipage}[t]{8ex}\color{red}\bf
		\begin{verbatim}
			(%i4) 
		\end{verbatim}
	\end{minipage}
	\begin{minipage}[t]{\textwidth}\color{blue}
		\begin{verbatim}
			last(CoeffsRE);
		\end{verbatim}
	\end{minipage}
	\definecolor{labelcolor}{RGB}{100,0,0}
	\[\displaystyle
	\parbox{10ex}{$\color{labelcolor}\mathrm{\tt (\%o4) }\quad $}
	-17920\cdot \left( 8+n\right) \cdot \left( 9+n\right) \cdot \left( 15+2\cdot n\right) \cdot \left( 17+2\cdot n\right) \cdot \left( 27+4\cdot n\right) \cdot \left( 33+4\cdot n\right) \mbox{}
	\]
	\textnormal{and the trailing one is}
	
	\noindent
	\begin{minipage}[t]{8ex}\color{red}\bf
		\begin{verbatim}
			(%i5) 
		\end{verbatim}
	\end{minipage}
	\begin{minipage}[t]{\textwidth}\color{blue}
		\begin{verbatim}
			first(CoeffsRE);
		\end{verbatim}
	\end{minipage}
	\definecolor{labelcolor}{RGB}{100,0,0}
	\[\displaystyle
	\parbox{10ex}{$\color{labelcolor}\mathrm{\tt (\%o5) }\quad $}
	-531441\mbox{}
	\]

	\textnormal{Therefore we deduce the Puiseux number $\lcm(1,2,4,4)=4$. Indeed the factors $(8+n)$ and $(9+n)$ have both denominator roots equal to $1$, $(15+2\cdot n)$ and $(17+2\cdot n)$ have both denominator roots equal to $2$, and $(27+4\cdot n)$ and $(33+4\cdot n)$ have both denominator roots equal to $4$. After substitution the new holonomic RE is free of Puiseux numbers.}
	
	\noindent
	\begin{minipage}[t]{8ex}\color{red}\bf
		\begin{verbatim}
			(%i6) 
		\end{verbatim}
	\end{minipage}
	\begin{minipage}[t]{\textwidth}\color{blue}
		\begin{verbatim}
			RE:FindRE(subst(z^4,z,f),z,a[n])$
		\end{verbatim}
	\end{minipage}
	\vspace{0.25cm}
	
	\noindent\textnormal{with leading term}
	
	\noindent
	\begin{minipage}[t]{8ex}\color{red}\bf
		\begin{verbatim}
			(%i7) 
		\end{verbatim}
	\end{minipage}
	\begin{minipage}[t]{\textwidth}\color{blue}
		\begin{verbatim}
			CoeffsRE:REcoeff(RE,a[n])$
		\end{verbatim}
	\end{minipage}
	
	\noindent
	\begin{minipage}[t]{8ex}\color{red}\bf
		\begin{verbatim}
			(%i8) 
		\end{verbatim}
	\end{minipage}
	\begin{minipage}[t]{\textwidth}\color{blue}
		\begin{verbatim}
			last(CoeffsRE);
		\end{verbatim}
	\end{minipage}
	\definecolor{labelcolor}{RGB}{100,0,0}
	\[\displaystyle
	\parbox{10ex}{$\color{labelcolor}\mathrm{\tt (\%o8) }\quad $}
	-4 \left( n+14\right) \, \left( n+15\right) \, \left( n+16\right) \, \left( n+18\right)\mbox{}
	\]
\end{example}

Having given an approach to reduce the computation of Puiseux series to Laurent series, for the next developments we assume that the Puiseux number is $1$.

\subsection{Computing starting points}

We have observed that Maple's current \textit{convert/FormalPowerSeries} command wrongly represents the power series of $\arctan(z) + \exp(z)$ due to the constant term missing. In this section we show how to avoid such a situation by explaining how to deduce exact starting points of hypergeometric type series from the holonomic recurrence equations of their coefficients. By trying to compute the representations of many examples of sums of polynomials and hypergeometric series in Maple, one realizes that such a computation is not well-managed in the implemented algorithm. A simple example is the following.

Maple's FPS gives

\begin{maplegroup}
	\begin{mapleinput}
		\mapleinline{active}{1d}{FPS(z + z\symbol{94}2 * exp(z),z,n);
		}{}
	\end{mapleinput}
	\mapleresult
	\begin{maplelatex}
		\[\displaystyle FPS(z + z^2 e^z, z, n)\]
	\end{maplelatex}
\end{maplegroup}
\noindent whereas our Maxima FPS implementation yields correctly

\noindent
\begin{minipage}[t]{8ex}\color{red}\bf
	\begin{verbatim}
		(%i1) 
	\end{verbatim}
\end{minipage}
\begin{minipage}[t]{\textwidth}\color{blue}
	\begin{verbatim}
		FPS(z+z^2*exp(z),z,n);
	\end{verbatim}
\end{minipage}
\definecolor{labelcolor}{RGB}{100,0,0}
\[\displaystyle
\parbox{10ex}{$\color{labelcolor}\mathrm{\tt (\%o1) }\quad $}
\left( \sum_{n=0}^{\infty }\frac{{{z}^{2+n}}}{n!}\right) +z\mbox{}
\]
which is the sum of the hypergeometric series of $z^2\cdot \exp(z)$ whose starting point is $n=2$ plus the polynomial $z$. The complete information is detected from the corresponding holonomic recurrence equation. Let us now explain how this can be done.

Again, we consider the general representation (assuming Puiseux numbers all equal to $1$)
\begin{equation}
	f(z) := T(z) + F(z) \label{lsgf}
\end{equation}
where $F(z)$ is a sum of hypergeometric type series and $T(z)\in\mathbb{K}[z,\frac{1}{z},\log(z)]$ is an extra term to be determined while computing the starting point for $F(z)$. Note that $T(z)$ can be given explicitly in the input expression, but also implicitly like for the expressions $\arcsech(z)$, $\arccosh(z)$ and $\exp(z)+\log(1+z)$.

First, we focus on the case where $T(z)$ is a Laurent polynomial in $\mathbb{K}[z,\frac{1}{z}]$. For this purpose we need to understand what it means for a Laurent polynomial that its coefficients are solution of a holonomic recurrence equation. Let us compute the holonomic RE for an unknown Laurent polynomial and figure out some properties of its coefficients from that RE.

\noindent
\begin{minipage}[t]{8ex}\color{red}\bf
	\begin{verbatim}
		(%i1) 
	\end{verbatim}
\end{minipage}
\begin{minipage}[t]{\textwidth}\color{blue}
	\begin{verbatim}
		FindRE(randompoly(z),z,a[n]);
	\end{verbatim}
\end{minipage}
\vspace{-0.6cm}
\definecolor{labelcolor}{RGB}{100,0,0}
\begin{multline}
	\displaystyle
	\parbox{10ex}{$\color{labelcolor}\mathrm{\tt (\%o1) }\quad $}
	-3\left( 5+n\right){{a}_{n}}-10\left( 3+n\right){{a}_{n-1}}+3\left( n-1\right){{a}_{n-3}}-7\left( n-5\right){{a}_{n-5}}\\
	-\left( n-7\right) {{a}_{n-6}}+26\left( n-11\right){{a}_{n-8}}=0. \label{explPol}
\end{multline} 

Our code \textit{randompoly} is used to generate an arbitrary Laurent polynomial. Since all polynomials are rational functions, \textit{HolonomicDE} always computes a holonomic differential equation of first order for a given polynomial. Thus we can find a general representation of their holonomic recurrence equations.
\begin{eqnarray}
	T(z) := \sum_{i=M}^{N}c_iz^i =\sum_{i\in\mathbb{Z}}^{}c_i z^i\in\mathbb{K}(z), \label{pol} 
\end{eqnarray}
for $M,N\in\mathbb{Z}, M\leqslant N$ where $c_i=0$ for $i\in\mathbb{Z}\setminus \llbracket M, N \rrbracket$, and $c_{-M}\cdot c_N\neq 0$, then the differential equation found is
\begin{equation}
	\sum_{i=-M}^{N} c_iz^i\cdot F'(z) - \sum_{i=-M}^{N}c_iiz^{i-1}\cdot F(z) = 0. \label{DEpol}
\end{equation}
Therefore using the rewrite rule $(\ref{Algo3corr})$ we obtain the recurrence equation
\begin{equation*}
	\sum_{i=M}^{N} c_i(n+1-i)\cdot a_{n+1-i} - \sum_{i=M}^{N}c_i i\cdot a_{n-(i-1)}=\sum_{i=M}^{N}c_i(n+1-2i)\cdot a_{n+1-i} = 0.
\end{equation*}
Hence the holonomic RE found by \textit{FindRE} of a Laurent polynomial with representation $(\ref{pol})$ is given by
\begin{equation}
	\sum_{i=M}^{N}c_i(n+1-2i)\cdot a_{n+1-i} = 0, \label{REpol1}
\end{equation}
or equivalently
\begin{equation}
	\sum_{i=M}^{N}c_i(n+N-2i)\cdot a_{n+N-i} = 0, \label{REpol2}
\end{equation}
after substitution of $n$ by $n+N$ for normalization.

Thus, without even using initial values a polynomial whose coefficients satisfy the holonomic RE $(\ref{explPol})$ can easily be found by equating the terms of $(\ref{REpol1})$ and $(\ref{explPol})$ to find the unknown coefficients $c_i$, using \textit{FindRE} to compute a holonomic RE for the resulting polynomial and check whether the REs are identical. We obtain the Laurent polynomial
\begin{equation}
	-26\cdot {{z}^{3}}+z-\frac{3}{{{z}^{2}}}+\frac{10}{{{z}^{4}}}+\frac{3}{{{z}^{5}}}+7. \label{polyexp}
\end{equation}

Of course this is not enough because there might be other solutions. And moreover, when the input expression is of the form $(\ref{lsgf})$, the situation is more complicated since $F(z)$ needs initial values in order to be computed. Therefore, we have to find the maximum degree $N\in\mathbb{Z}$ of $T(z)$ so that $F(z)$ starts at $N+1$ and $T(z)$ is computed by a generalized Taylor expansion of order $N$ of $f(z)$.

Observe for each non-zero coefficient $c_i$, $i\in \llbracket M, N \rrbracket$ of $T(z)$, that $2i-N$ is the root of one polynomial coefficient in $(\ref{REpol2})$. In particular, $N$ is the trailing polynomial coefficient root and $M$ is the root of the leading polynomial coefficient shifted by $N-M$. These two properties of the degrees of a potential Laurent polynomial whose coefficients satisfy a holonomic recurrence equation is preserved in the general case. This is stated by the following lemma.
\begin{lemma}\label{polypart} Let $\mathbb{K}$ be a field of characteristic zero, $N,M\in\mathbb{Z}, N\geqslant M$, $T(z)\in\mathbb{K}[z,\frac{1}{z}]$ be a Laurent polynomial of degree $N$ and lowest non-zero monomial degree $M$. The coefficients of $T(z)$ satisfy the holonomic recurrence equation 
	\begin{equation}
		P_d a_{n+d} + P_{d-1} a_{n+d-1} + \ldots + P_0 a_n =0, \label{RE8}
	\end{equation} 
	$d\in\mathbb{N}, P_j\in\mathbb{K}[n], j\in\llbracket 0, d\rrbracket, P_d\cdot P_0\neq 0,$ if $N$ is a root of $P_0$ and $M$ is a root of $P_d(n-d)$.	
\end{lemma}
\begin{proof}
	Suppose that the coefficients of $T(z)$ satisfy $(\ref{RE8})$. Since $T(z)$ has finitely many non-zero coefficients we can write
	$$T(z)=\sum_{i\in\mathbb{Z}}^{}c_n z^n,$$
	where $c_n=0$ for $n\in\mathbb{Z}\setminus \llbracket M, N \rrbracket$. Saying that the coefficients of $T(z)$ satisfy $(\ref{RE8})$ is equivalent to say that the sequence $(c_n)_{n\in\mathbb{Z}}$ is a sequence solution of $(\ref{RE8})$. Given that $(\ref{RE8})$ is valid for all integers, observe that substituting $a_n$ by $c_n$ in $(\ref{RE8})$ for sufficiently large positive or negative integers all the summands on the left-hand side of $(\ref{RE8})$ vanish. 
	
	Furthermore, we can make a substitution such that either the trailing or the leading term does not necessarily give zero. Indeed, since $c_n=0$ for $n\in\mathbb{Z}\setminus \llbracket M, N \rrbracket$, substituting $a_n$ by $c_n$ in $(\ref{RE8})$ for $n=N$ yields 
	$$ P_0(N) c_N = 0,$$
	and therefore using the assumption $c_N\neq 0$ we deduce that $P_0(N)=0$. Similarly, substituting $a_n$ by $c_n$ in $(\ref{RE8})$ for $n=M-d$ gives
	$$ P_d(M-d) c_M = 0,$$
	and therefore as $c_M\neq 0$ by assumption, it follows that $P_d(M-d)=0$.
\end{proof}

\begin{remark}
Note that generally when $T(z)=0$ and $F(z)$ starts at $0$, $N=M=0$ and $0$ is not necessarily a root of the trailing polynomial coefficient. This might be interpreted from the fact that the zero function is always solution of any holonomic RE and moreover takes $0$ at $0$. Therefore as $T(z)$ does not play a disturbing role, we rather say in this case that it does not exist. This is the case with

\noindent
\begin{minipage}[t]{8ex}\color{red}\bf
	\begin{verbatim}
		(%i2) 
	\end{verbatim}
\end{minipage}
\begin{minipage}[t]{\textwidth}\color{blue}
	\begin{verbatim}
		FindRE(exp(z),z,a[n]);
	\end{verbatim}
\end{minipage}
\definecolor{labelcolor}{RGB}{100,0,0}
\[\displaystyle
\parbox{10ex}{$\color{labelcolor}\mathrm{\tt (\%o2) }\quad $}
\left( 1+n\right) \cdot {{a}_{n+1}}-{{a}_{n}}=0\mbox{}
\]
whose trailing polynomial coefficient does not have any root. $T(z)$ in this case is $0$ or is said to not exist. However, note that in this example $M=0$ is a root of the leading polynomial coefficient which represents the starting point of the series expansion of $\exp(z)$. In general, the computation of $M$ is always possible from all the REs computed by \textit{FindRE}, and represents moreover the starting point for the series expansion of the given $f(z)$. Indeed, the fact that \textit{FindRE} does not cancel the common factors after application of the rewrite rule $(\ref{Algo3corr})$ is essential for our computations of starting points. These factors contain necessary information to determine the first non-zero coefficient of the series expansion sought. Let for example 

\noindent
\begin{minipage}[t]{8ex}\color{red}\bf
	\begin{verbatim}
		(%i3) 
	\end{verbatim}
\end{minipage}
\begin{minipage}[t]{\textwidth}\color{blue}
	\begin{verbatim}
		FindRE(z/(1-z),z,a[n]);
	\end{verbatim}
\end{minipage}
\definecolor{labelcolor}{RGB}{100,0,0}
\[\displaystyle
\parbox{10ex}{$\color{labelcolor}\mathrm{\tt (\%o3) }\quad $}
\left( n-1\right) \cdot {{a}_{n-1}}-\left( n-1\right) \cdot {{a}_{n}}=0\mbox{}
\]
for which the cancellation of the common factor $(n-1)$ (or $n$ after normalization) would hide the starting point $N+1=1$ ($N=0$ is the root of the trailing polynomial coefficient).
\end{remark} 

Note that using this lemma we can now confirm that any Laurent polynomial whose sequence of coefficients satisfies the holonomic RE $(\ref{explPol})$ is a constant multiple of the polynomial $(\ref{polyexp})$. Indeed, the leading and the trailing polynomial coefficients of $(\ref{explPol})$ have only one integer root each which are the degree bounds of $(\ref{polyexp})$. Algorithmically, we proceed as follows. 

\begin{algorithm}[h!]
	\caption{Computing $T(z)$ and the starting point of $F(z)$ for $f=T(z)+F(z)$ as in $(\ref{lsgf})$}\label{PolyPart}
	\begin{algorithmic}[3]
		\Require An expression $f$ whose series coefficients satisfy the holonomic recurrence equation
		\begin{equation}
			P_d a_{n+d} + P_{d-1} a_{n+d-1} + \ldots + P_0 a_n =0, \label{RE9}
		\end{equation} 
		$d\in\mathbb{N}, P_j\in\mathbb{K}[n], j\in\llbracket 0, d\rrbracket, P_d\cdot P_0\neq 0,$
		\Ensure $T(z)$ and a starting point $N_0$ for $F(z)$ for the representation $(\ref{lsgf})$ of $f$.
		\begin{enumerate}
			\item\label{PPstep1} Compute the minimum integer roots $M$ of $P_d(n-d)$ and the maximum integer root $N$ of $P_0(n)$.
			\item If $N$ does not exist then set $T(z):=0$ and set $N_0:=M$.
			\item If $N$ does exist then set $T(z):= Taylor(f(z),z,0,N)$ and set $N_0:=N+1$.
			\item Return $[T(z),N_0]$.
		\end{enumerate}
	\end{algorithmic}
\end{algorithm}

Note that Lemma $\ref{polypart}$ extends to Laurent polynomials in $\mathbb{K}[\log(z)][z,\frac{1}{z}]$ as connected to generalized Taylor expansions or power series in \cite[Section 7.3]{kauers2011}. Our Maxima package has the code \textit{LPolyPart(f,z)} that implements Algorithm \ref{PolyPart}. This can be used as follows. 
\begin{example}\item
\noindent
\begin{minipage}[t]{8ex}\color{red}\bf
	\begin{verbatim}
		(%i1) 
	\end{verbatim}
\end{minipage}
\begin{minipage}[t]{\textwidth}\color{blue}
	\begin{verbatim}
		LPolyPart(asech(z),z);
	\end{verbatim}
\end{minipage}
\definecolor{labelcolor}{RGB}{100,0,0}
\[\displaystyle
\parbox{10ex}{$\color{labelcolor}\mathrm{\tt (\%o1) }\quad $}
[\mathrm{log}\left( 2\right) -\mathrm{log}\left( z\right) ,1]\mbox{}
\]

\noindent
\begin{minipage}[t]{8ex}\color{red}\bf
	\begin{verbatim}
		(%i2) 
	\end{verbatim}
\end{minipage}
\begin{minipage}[t]{\textwidth}\color{blue}
	\begin{verbatim}
		LPolyPart(exp(z)+log(1+z),z);
	\end{verbatim}
\end{minipage}
\definecolor{labelcolor}{RGB}{100,0,0}
\[\displaystyle
\parbox{10ex}{$\color{labelcolor}\mathrm{\tt (\%o2) }\quad $}
[1,1]\mbox{}
\]

\noindent
\begin{minipage}[t]{8ex}\color{red}\bf
	\begin{verbatim}
		(%i3) 
	\end{verbatim}
\end{minipage}
\begin{minipage}[t]{\textwidth}\color{blue}
	\begin{verbatim}
		LPolyPart(sin(z)/z^5,z);
	\end{verbatim}
\end{minipage}
\definecolor{labelcolor}{RGB}{100,0,0}
\[\displaystyle
\parbox{10ex}{$\color{labelcolor}\mathrm{\tt (\%o3) }\quad $}
[0,-5]\mbox{}
\]
\end{example}

%Observe in this latter example that $T(z)$ is found to be $1$ but the representation given by our Maxima FPS implementation does not contain this extra term. The reason is that $N+1$ is not necessarily the exact starting point but rather its maximum value possible. This shows that the exact $T(z)$ is not uniquely determined but only the polynomial from which it can always be subtracted. Nevertheless, it is safe to have such a value since it does not affect the correctness. Therefore our algorithm tries to subtract certain terms from $T(z)$ after having found the linear combination needed for $F(z)$. For the latter example, it turns out that $T(z)$ is obtained as $\frac{z^0}{0!}$ from one of the obtained hypergeometric type series.

As last example let us take the case of the Chebyshev polynomial $\cos(4\arccos(z))$.

\noindent
\begin{minipage}[t]{8ex}\color{red}\bf
	\begin{verbatim}
		(%i4) 
	\end{verbatim}
\end{minipage}
\begin{minipage}[t]{\textwidth}\color{blue}
	\begin{verbatim}
		LPolyPart(cos(4*acos(z)),z);
	\end{verbatim}
\end{minipage}
\definecolor{labelcolor}{RGB}{100,0,0}
\[\displaystyle
\parbox{10ex}{$\color{labelcolor}\mathrm{\tt (\%o4) }\quad $}
[8\cdot {{z}^{4}}-8\cdot {{z}^{2}}+1,5]\mbox{}
\]

Thus, the starting point to compute the linear combination for $F(z)$ is $5$. This example leads to a two-term holonomic recurrence equation. We will see in the next subsection that the linear combination of the corresponding hypergeometric type series yields $0$ so that one finally gets the known result $\cos(4\arccos(z))=T(z)=8\cdot {{z}^{4}}-8\cdot {{z}^{2}}+1$.

\subsection{Revision of the two-term holonomic RE case}

One may ask why we need such an algorithm since it is already generalized by using mfoldHyper. This could be answered by the following example whose corresponding two-term holonomic recurrence relation has hypergeometric and $2$-fold hypergeometric term solutions over $\mathbb{Q}$.

\noindent
\begin{minipage}[t]{8ex}\color{red}\bf
	\begin{verbatim}
		(%i1) 
	\end{verbatim}
\end{minipage}
\begin{minipage}[t]{\textwidth}\color{blue}
	\begin{verbatim}
		RE:FindRE(cosh(z),z,a[n]);
	\end{verbatim}
\end{minipage}
\definecolor{labelcolor}{RGB}{100,0,0}
\[\displaystyle
\parbox{10ex}{$\color{labelcolor}\mathrm{\tt (\%o1) }\quad $}
\left( 1+n\right) \cdot \left( 2+n\right) \cdot {{a}_{n+2}}-{{a}_{n}}=0\mbox{}
\]

\noindent
\begin{minipage}[t]{8ex}\color{red}\bf
	\begin{verbatim}
		(%i2) 
	\end{verbatim}
\end{minipage}
\begin{minipage}[t]{\textwidth}\color{blue}
	\begin{verbatim}
		mfoldHyper(RE,a[n]);
	\end{verbatim}
\end{minipage}
\definecolor{labelcolor}{RGB}{100,0,0}
\[\displaystyle
\parbox{10ex}{$\color{labelcolor}\mathrm{\tt (\%o2) }\quad $}
\left[\left[1,\left\{\frac{1}{n!},\frac{{{\left( -1\right) }^{n}}}{n!}\right\}\right],\left[2,\left\{\frac{1}{\left( 2\cdot n\right) !}\right\}\right]\right]\mbox{}
\]

\noindent
\begin{minipage}[t]{8ex}\color{red}\bf
	\begin{verbatim}
		(%i3) 
	\end{verbatim}
\end{minipage}
\begin{minipage}[t]{\textwidth}\color{blue}
	\begin{verbatim}
		mfoldHyper(RE,a[n],2,1);
	\end{verbatim}
\end{minipage}
\definecolor{labelcolor}{RGB}{100,0,0}
\[\displaystyle
\parbox{10ex}{$\color{labelcolor}\mathrm{\tt (\%o3) }\quad $}
{\frac{1}{\left( 2\cdot n+1\right) \cdot \left( 2\cdot n\right) !}}\mbox{}
\]

Thus using mfoldHyper imposes to decide between the representations
\begin{equation}
	\cosh(z) := \sum_{n=0}^{\infty} \dfrac{1+(-1)^n}{2 n!}z^n,
\end{equation}
and 
\begin{equation}
	\cosh(z) = \sum_{n=0}^{\infty }\frac{{{z}^{2\cdot n}}}{\left( 2\cdot n\right) !},
\end{equation}
which are both correct. We will see in the next subsection that this is decided by the linear system to be solved; arbitrary constants appearing in the solution are set to zero to ease the decision and this may lead to a little more complicated representation. Therefore we propose to revisit the algorithm as described in (\cite{Koepf1992}) with our formalism of starting points and Puiseux numbers of series expansions.

Let us sketch the two-term holonomic RE algorithm for our introductory example. The recurrence equation found for $\cosh(z)$ is
\[\left( 1+n\right) \cdot \left( 2+n\right) \cdot {{a}_{n+2}}-{{a}_{n}}=0.\]
We immediately get the symmetry number $m=2$. Therefore the corresponding $2$-fold symmetric ratios are
\begin{equation}
	\dfrac{a_{2(n+1)}}{a_{2n}} = \frac{1}{\left( 2\cdot n+1\right) \cdot \left( 2\cdot n+2\right) }~~\text{ and }~~ \dfrac{a_{2(n+1)+1}}{a_{2n+1}} = \frac{1}{\left( 2\cdot n+2\right) \cdot \left( 2\cdot n+3\right) }.
\end{equation}
We get the coefficients

\[\displaystyle
\frac{1}{\left( 2\cdot n\right) !}\mbox{}
~~\text{ and }~~
\frac{1}{\left( 1+2\cdot n\right) !}\mbox{},
\]
respectively. We now write 
\begin{equation*}
	I(z) := \alpha_0 \cdot \sum_{n=0}^{0} \frac{z^{2\cdot n}}{\left( 2\cdot n\right) !} + \alpha_1 \cdot \sum_{n=0}^{0} \frac{z^{2\cdot n+1}}{\left( 2\cdot n + 1\right) !},
\end{equation*}
and use $2$ initial values $n=0=2\cdot 0$ and $n=1=2\cdot 0 +1$ to search for the unknown constants $\alpha_0$ and $\alpha_1$. We have

\noindent
\begin{minipage}[t]{8ex}\color{red}\bf
	\begin{verbatim}
		(%i4) 
	\end{verbatim}
\end{minipage}
\begin{minipage}[t]{\textwidth}\color{blue}
	\begin{verbatim}
		taylor(cosh(z),z,0,1);
	\end{verbatim}
\end{minipage}
\definecolor{labelcolor}{RGB}{100,0,0}
\[\displaystyle
\parbox{10ex}{$\color{labelcolor}\mathrm{\tt (\%o4)/T/ }\quad $}
1+ \ldots\mbox{}
\]
therefore $\alpha_0 + \alpha_1\cdot z = 1$, hence $\alpha_0=1$ and $\alpha_1=0$. And finally we obtain the power series representation
\begin{equation}
	\cosh(z) = \sum_{n=0}^{\infty }\frac{{{z}^{2\cdot n}}}{\left( 2\cdot n\right) !},
\end{equation}
as expected.

For $\cos(4\cdot \arccos(z))$, \textit{FindRE} gives the recurrence equation

\noindent
\begin{minipage}[t]{8ex}\color{red}\bf
	\begin{verbatim}
		(%i5) 
	\end{verbatim}
\end{minipage}
\begin{minipage}[t]{\textwidth}\color{blue}
	\begin{verbatim}
		RE:FindRE(cos(4*acos(z)),z,a[n]);
	\end{verbatim}
\end{minipage}
\definecolor{labelcolor}{RGB}{100,0,0}
\[\displaystyle
\parbox{10ex}{$\color{labelcolor}\mathrm{\tt (\%o5) }\quad $}
\left( n-4\right) \cdot \left( 4+n\right) \cdot {{a}_{n}}-\left( 1+n\right) \cdot \left( 2+n\right) \cdot {{a}_{n+2}}=0\mbox{}
\]

As computed in the previous section, the starting point for the corresponding hypergeometric type series part is $5$. Shifting the ratio of the trailing and the leading polynomial coefficients by $5$ yields
\[-\frac{\left( 1+n\right) \cdot \left( 9+n\right) }{\left( n+6\right) \cdot \left( n+7\right) }.\]
Therefore the corresponding $2$-fold symmetric ratios are
\begin{equation}
	r_0 =-\frac{\left( 1+2\cdot n\right) \cdot \left( 9+2\cdot n\right) }{\left( 2\cdot n+6\right) \cdot \left( 2\cdot n+7\right) }~\text{and}~ r_1 = -\frac{\left( 2+2\cdot n\right) \cdot \left( 10+2\cdot n\right) }{\left( 2\cdot n+7\right) \cdot \left( 2\cdot n+8\right) }
\end{equation}
which lead to the coefficients
\[\displaystyle
\frac{2\cdot \left( 7+2\cdot n\right) \cdot {{\left( -1\right) }^{n}}\cdot \left( 2\cdot n\right) !}{7\cdot \left( n+1\right) \cdot \left( n+2\right) \cdot {{4}^{n}}\cdot {{n!}^{2}}}\mbox{}
~~\text{ and }~~
\frac{15\cdot \left( 1+n\right) \cdot \left( 2+n\right) \cdot \left( 4+n\right) \cdot {{\left( -1\right) }^{n}}\cdot {{4}^{n}}\cdot {{n!}^{2}}}{\left( 5+2\cdot n\right) !},\mbox{}
\]
respectively. We have to use $2$ initial values corresponding to $n=2\cdot 0 + 5$ and $n=2\cdot 0 + 1 + 5$. First, we define

\noindent
\begin{minipage}[t]{8ex}\color{red}\bf
	\begin{verbatim}
		(%i6) 
	\end{verbatim}
\end{minipage}
\begin{minipage}[t]{\textwidth}\color{blue}
	\begin{verbatim}
		I:alpha[0]*subst(0,n,h0)*z^5 + alpha[1]*subst(0,n,h1)*z^6
		+ ratdisrep(taylor(cos(4*acos(z)),z,0,4));
	\end{verbatim}
\end{minipage}
\definecolor{labelcolor}{RGB}{100,0,0}
\vspace{0.25cm}
\[\displaystyle
\parbox{10ex}{$\color{labelcolor}\mathrm{\tt (\%o6) }\quad $}
{{\alpha}_{1}}\cdot {{z}^{6}}+{{\alpha}_{0}}\cdot {{z}^{5}}+8\cdot {{z}^{4}}-8\cdot {{z}^{2}}+1\mbox{}
\]
for the unknown constants $\alpha_0$ and $\alpha_1$. Remember that $taylor(\cos(4\arccos(z)),z,0,4)$ is $T(z)$ for the representation $(\ref{lsgf})$ of $\cos(4\arccos(z))$. To find the values of $\alpha_0$ and $\alpha_1$, we just have to solve the trivial identity

\noindent
\begin{minipage}[t]{8ex}\color{red}\bf
	\begin{verbatim}
		(%i7) 
	\end{verbatim}
\end{minipage}
\begin{minipage}[t]{\textwidth}\color{blue}
	\begin{verbatim}
		I-Taylor(cos(4*acos(z)),z,0,6)=0;
	\end{verbatim}
\end{minipage}
\[\displaystyle
\parbox{10ex}{$\color{labelcolor}\mathrm{\tt (\%o7) }\quad $}
{{\alpha}_{1}}\cdot {{z}^{6}}+{{\alpha}_{0}}\cdot {{z}^{5}}=0\mbox{}
\]
We thus obtain $\alpha_0=\alpha_1=0$ and therefore 
\[\cos(4\cdot\arccos(z)) = 1 -8\cdot {{z}^{2}} + 8\cdot {{z}^{4}}\]
In a nutshell the algorithm is presented as follows.

\begin{algorithm}[h!]
	\caption{Power series representations of hypergeometric type functions whose \textit{FindRE} computes two-term holonomic REs for their coefficients.}\label{twoterm}
	\begin{algorithmic}[3]
		\Require A holonomic expression $f$ and a holonomic RE
		\begin{equation}
			Q(n) a_{n+m} - P(n) a_n =0, \label{RE2term}
		\end{equation}
		$P,Q\in\mathbb{K}[n], P\cdot Q\neq 0$, computed by \textit{FindRE} for the series coefficient of $f$.
		\Ensure The power series representation of $f$ at $0$.
		\begin{enumerate}
			\item Use Algorithm $\ref{PolyPart}$ to compute the corresponding $T(z)$ and a starting point $N_0$ for the representation $(\ref{lsgf})$ of $f$.
			\item Set $r(n) = P(n+N_0)/Q(n+N_0).$
			\item Compute the $m$ symmetric ratios
			\begin{equation}
				r_j(n) := r(m\cdot n + j), ~ j= 0,\ldots, m-1.
			\end{equation}
			\item Compute formulas (see "simple formulas" in (\cite{BThyper}) or hypergeometric formula in \cite[Chapter 2]{WolfBook}) for $h_j(n):= \prod_{i=0}^{n-1} r_j(i)$, for  $j=0,\ldots,m-1.$
			\item\label{ttstep5} Set $I(z):= T(z) + \sum_{j=0}^{m-1}\alpha_j \cdot h_j(0)\cdot z^{j + N_0}$, where $\alpha_j, j=0,\ldots,m-1$ are unknown constants.
			\item Find the values of $\alpha_j$ by equating the coefficient of the Laurent polynomial
			\begin{equation}
				Taylor(f(z),z,0,N_0+m-1)-I(z)
			\end{equation}
			to zero.
			\item Return $ T(z) + \sum_{j=0}^{m-1} \alpha_j' \cdot \sum_{n=0}^{\infty} h_j(n)\cdot z^{m\cdot n + j + N_0}$, where $\alpha_j'$ is the value found for $\alpha_j$.
		\end{enumerate}
	\end{algorithmic}
\end{algorithm}

Most results of the category of Algorithm \ref{twoterm} are well handled by the current Maple \textit{convert/FormalPowerSeries}. However, since Puiseux numbers and starting points are not symbolically computed in this previous implementation, we give two examples missed by \textit{convert/FormalPowerSeries} but accessible to our implementation. Many more can be constructed.

\begin{maplegroup}
	\begin{mapleinput}
		\mapleinline{active}{1d}{FPS[FPS](1/(sqrt(1-4*z))*((1-sqrt(1-4*z))/2*z)\symbol{94}2,z,n)
		}{}
	\end{mapleinput}
	\mapleresult
	\begin{maplelatex}
		\[\displaystyle \sum _{n=0}^{\infty }{\frac { \left( 2\,n+2 \right) !\, \left( n+2 \right)  \left( n+1 \right) {z}^{n+4}}{ \left(  \left( n+2 \right) ! \right) ^{2}}}\]
	\end{maplelatex}
\end{maplegroup}
\begin{maplegroup}
	\begin{mapleinput}
		\mapleinline{active}{1d}{FPS(1/(z+z\symbol{94}2),z,n)
		}{}
	\end{mapleinput}
	\mapleresult
	\begin{maplelatex}
		\[\displaystyle \sum _{n=0}^{\infty } \left( -1 \right) ^{n}{z}^{n-1}\]
	\end{maplelatex}
\end{maplegroup}

Although Algorithm \ref{twoterm} can be applied to summands of linear combinations of hypergeometric type functions to get good results in certain cases, we mention that this is not equivalent with the direct approach to be described next. Emphasis on that fact can be found in (\cite{BTmaple}). 

\subsection{Hypergeometric type series with arbitrary holonomic REs}

As previously we consider an expression $f(z)$ related to a hypergeometric type function and we want to compute the representation
\begin{equation}
	f(z) := T(z) + F(z), \label{gcrep}
\end{equation} 
where $T(z)$ is generally a Laurent polynomial in $\mathbb{K}[z,\frac{1}{z},\log(z)]$, and $F(z)$ is a hypergeometric type series. We have already shown how to compute $T(z)$ and a starting point for $F(z)$. 

We have to avoid negative arguments for the evaluation of $m$-fold hypergeometric terms which are supposed to start at least at $0$ according to the algorithm in (\cite{BThyper}). When there are many different types involved in $F(z)$ as pointed out with $\sin(z)/z^3+\exp(z)$, the exact starting point is in fact the minimal one among those of the hypergeometric type series in $F(z)$. We only have to make sure that this value is positive in order to avoid inappropriate operations like a division by $0$ or factorials of negative integers.

In practice, we use 
\begin{equation}
	N_1= \max\{0,N_0\}\geqslant 0,
\end{equation}
from which all necessary $m$-fold hypergeometric terms can be evaluated for initial conditions. Once the linear combination sought is found, if possible we subtract terms from the corresponding $F(z)$ for the indices $n\in\llbracket N_0, N_1\rrbracket$. Note that $T(z)$ should also be modified accordingly. We set
\begin{equation}
	T_1(z):= Taylor(f(z),z,0,N_1-1).
\end{equation}
By these measures, the final representation gives a normal form after shifting the power of the indeterminate by the first non-zero term index in each hypergeometric type series.

Let us now find a representation $(\ref{gcrep})$ of $f(z)$ for a computed $T_1(z)$ and a starting point $N_1$ knowing that we can make further computations to subtract terms in $T_1(z)$ that can be deduced from $F(z)$ and vice versa.

To ease the understanding of the general case, we first describe the details for the situation where
\begin{equation}
	\mathcal{H} = \left\lbrack \left[2, \{h_{2n}\}\right], \left[3, \{h_{3n}\}\right]  \right\rbrack \label{hmfold},
\end{equation}
represent the obtained basis of $m$-fold hypergeometric term solutions of the holonomic RE given by \textit{FindRE} for the series coefficients of $f(z)$. Let us also assume for simplicity that $N_1=T_1(z)=0$, then the general form for the corresponding $F(z)$ can be written as
\begin{multline}
	F(z) = \alpha_{2,0} \cdot \sum_{n=0}^{\infty} h_{2n}z^{2n}  + \alpha_{2,1} \cdot \sum_{n=0}^{\infty} h_{2n+1}z^{2n+1} + \alpha_{3,0} \cdot \sum_{n=0}^{\infty} h_{3n}z^{3n}\\
	+ \alpha_{3,1} \cdot \sum_{n=0}^{\infty} h_{3n+1}z^{3n+1} + \alpha_{3,2} \cdot \sum_{n=0}^{\infty} h_{3n+2}z^{3n+2}, \label{Feg}
\end{multline}
$\alpha_{2,i},\alpha_{3,j}\in\mathbb{K},i=0,1,j=0,1,2$. Hence we have five unknowns to determine. Observe that computing a series expansion of order $4$ of $f(z)$ might not be enough. Indeed, the Taylor expansion of order $4$ would give five linear equations for the unknown constants but it turns out that the obtained linear system is not sufficient to determine these. Assume
\begin{equation}
	Taylor(f(z),z,0,4) = t_0 + t_1 z + t_2 z^2 + t_3 z^3 + t_4 z^4, \label{T4}
\end{equation} 
then equating the coefficients with their corresponding terms in $(\ref{Feg})$ yields the linear system
\begin{equation}
	\begin{cases}
		\alpha_{2,0} \cdot h_{2n}(0) + \alpha_{3,0}\cdot h_{3n}(0) = t_0\\
		\alpha_{2,1} \cdot h_{2n+1}(0) + \alpha_{3,1}\cdot h_{3n+1}(0) = t_1\\
		\alpha_{2,0} \cdot h_{2n}(1) + \alpha_{3,2}\cdot h_{3n+2}(0) = t_2\\
		\alpha_{2,1} \cdot h_{2n+1}(1) + \alpha_{3,0}\cdot h_{3n}(1) = t_3\\
		\alpha_{2,1} \cdot h_{2n+1}(2) + \alpha_{3,1}\cdot h_{3n+1}(1) = t_4
	\end{cases}
\end{equation}
from which a value for $\alpha_{2,0}$ cannot be deduced because it appears in three equations with three other different unknown constants. 

What we need is to use the series expansion of order $p\in\mathbb{N}$ of $f(z)$ in such a way that there exists $q\in\mathbb{N}, q\leqslant p$, so that there are at least $q$ linear equations with $q$ unknowns each in the resulting linear system. The minimal value of such a $p$ in this particular example is $2\cdot x_2 =  3\cdot x_3$ where $x_2$ and $x_3$ are the minimal positive integers verifying $2\cdot x_2 = 3\cdot x_3,$ hence $x_2=3, x_3=2$ and  $p=6=\lcm(2,3)$. Indeed, a series expansion of order $6$ gives two linear equations for $\alpha_{2,0}$ and $\alpha_{3,0}$ and this allows to find their values and deduce those of the other constants. If moreover there were a hypergeometric term in $(\ref{hmfold})$, then $6$ linear equations could not be enough. In this case the minimal value for $p$ would be $2\cdot 6 = 12$ in order to have at least three equations for $\alpha_{2,0},$ $\alpha_{3,0}$ and the unknown constant related to the hypergeometric term.

We now move to the general case.

Let
\begin{eqnarray}
	\mathcal{H} &:=& \bigg\lbrack \left[1,\left\lbrace h_{n,1},\ldots,h_{n,l_1}\right\rbrace\right], \left[m_1,\left\lbrace h_{m_1n,1},\ldots,h_{m_1n,l_{m_1}}\right\rbrace\right],\ldots,\left[m_{\mu},\left\lbrace h_{m_{\mu} n,1},\ldots,h_{m_{\mu} n,l_{m_{\mu}}}\right\rbrace\right]\bigg\rbrack\label{Hmfold}\nonumber\\ 
	&=& \left\lbrack \left[1,S_{1,0}\right],\left[m_1,S_{m_1,0}\right],\ldots,\left[m_{\mu}, S_{m_{\mu},0}\right]\right\rbrack
\end{eqnarray}
for integers $1<m_1<\cdots<m_{\mu}$ be the non-empty generator of all $m$-fold hypergeometric term solutions of a holonomic recurrence equation satisfied by the series coefficients of $f(z)$. $m_{\mu}$ is the maximum symmetry number, $l_{m}$ is the number of $m$-fold hypergeometric terms in $\mathcal{H}$ $m\in\{1, m_1,\ldots, m_{\mu}\}$. The representation $(\ref{gcrep})$ for $f(z)$ is computed as follows.

\begin{algorithm}[h!]
	\caption{Computing hypergeometric type series}\label{FPScombin}
	\begin{algorithmic}[3]
		\Require $f(z)$, the recurrence equation, say \textit{RE} computed by \textit{FindRE},  the span of all $m$-fold hypergeometric term solutions of \textit{RE}, say $\mathcal{H}$, computed by mfoldHyper, $T(z)$ and $N_0$ computed by Algorithm $\ref{PolyPart}$.
		\Ensure The representation $(\ref{gcrep})$ of $f$.
		\begin{enumerate}
			\item Find the other $m$-fold symmetric terms associated to each $m$-fold hypergeometric term in $\mathcal{H}$ for $m\in\{ m_1,\ldots, m_{\mu}\}$. For that purpose one calls Algorithm $\ref{mymfold}$ as \textit{mfoldHyper(RE,a[n],m,j)} for $j=1,\ldots, m-1,$ $m\in\{ m_1,\ldots, m_{\mu}\}$. This allows to build the sets
			\begin{equation}
				S_m:= \left\lbrace S_{m,0},S_{m,1}\ldots,S_{m,m-1}\right\rbrace,
			\end{equation}
			for $m\in\{1, m_1,\ldots, m_{\mu}\}$, where 
			\begin{equation}
				S_{m,j} := \left\lbrace h_{m n + j,1},h_{m n + j, 2}, \ldots, h_{m n + j, l_m} \right\rbrace,~0\leqslant j \leqslant m-1.
			\end{equation}
		\end{enumerate}
		
		\algstore{pause444}
	\end{algorithmic}
\end{algorithm}
\clearpage 

\begin{algorithm}[h!]
	\ContinuedFloat
	\caption{Computing hypergeometric type series}
	\begin{algorithmic}[3]
		\algrestore{pause444}	
		\State		
		
		\begin{enumerate}
			\setcounter{enumi}{1}
			\item Set $N_1=\max\{0,N_0\}$ and $T_1(z):=Taylor(f(z),z,0,N_1-1)$.
			\item Compute $i_{m,j} = \left\lceil \frac{N_1-j}{m}\right\rceil$ for $j=0,\ldots,m-1$, $m\in\{ m_1,\ldots, m_{\mu}\}$.
			\item Set
			\begin{equation}
				\mathcal{N} = N_1 + \left(\sum_{m\in\{1,m_1,\ldots,m_{\mu}\}}^{}l_m  - 1\right) \cdot \lcm(1,m_1,\ldots, m_{\mu}) + m_{\mu} - 1\label{Nfps}
			\end{equation}
			\item Compute $p_{m,j}= \left\lfloor \frac{\mathcal{N}-j}{m}\right\rfloor$, $j=0,\ldots,m-1$, $m\in\{ m_1,\ldots, m_{\mu}\}$.
			\item Let $\alpha_{m,j,k}\in\mathbb{K}$, $m\in\{1, m_1,\ldots, m_{\mu}\}$, $j=0,\ldots,m-1$, $k=1,\ldots,l_m$ be some unknown constants and define 
			\begin{equation}
				I(z) := \sum_{m\in\{1, m_1,\ldots, m_{\mu}\}}\sum_{j=0}^{m-1}\sum_{k=1}^{l_m} \alpha_{m,j,k} \sum_{n=i_{m,j}}^{p_{m,j}} h_{mn+j,k} z^{mn+j}.
			\end{equation}
			\item\label{fpsstep} Solve the linear system resulting from the equation
			\begin{equation}
				I(z) + T_1(z) - Taylor(f(z),z,0,\mathcal{N})=0, \label{systemc}
			\end{equation}
			for the unknown $\left(\alpha_{m,j,k}\right)_{m\in\{1, m_1,\ldots, m_{\mu}\},~0\leqslant j \leqslant m-1,~1\leqslant k \leqslant l_m}^T\in\mathbb{K}^{\sum_{m\in\{1,m_1,\ldots,m_{\mu}\}}^{}l_m\cdot m}$. 
			\item If there is no solution then stop and return FALSE. No linear combination exists in this case.
			\item If there is a solution then set all parameters of dependency to $0$ (if there are some). This gives the choice of the linear combination. We denote by $\alpha_{m,j,k}'$ the resulting value found for $\alpha_{m,j,k}$, $m\in\{1, m_1,\ldots, m_{\mu}\}$, $j=0,\ldots,m-1$, $k=1,\ldots,l_m$. 
			\item\label{psfpscombin} For each $S_m, m\in\{1,m_1,\ldots, m_{\mu}\}$ construct the term
			\begin{eqnarray}
				S_m' &:=& \sum_{S_{m,j}\in S_m} \left(\sum_{h_{mn+j,k}\in S_{m,j}} \alpha_{m,j,k}'h_{mn+j,k}\right) z^{mn+j-i_{m,j}}\\
				&:=& \sum_{j=0}^{m-1} \left(\sum_{k=1}^{l_m} \alpha_{m,j,k}'h_{mn+j,k}\right) z^{mn+j-i_{m,j}}	
			\end{eqnarray}	
			\item For each $S_m'$, $m\in\{1, m_1,\ldots, m_{\mu}\}$, make evaluations for $n\in\llbracket N_0, N_1 \rrbracket$ to subtract terms in $T_1(z)$ that can be computed from $S_m'$ and shift  the initial index $i_{m,j}$ accordingly. (This step could also be done before step $\ref{psfpscombin}$ to get more suitable values for starting points).
			\item Return $T_1(z) + \sum_{m\in\{1, m_1,\ldots, m_{\mu}\}} \sum_{n=0}^{\infty} S_m'$.
		\end{enumerate}			
	\end{algorithmic}
\end{algorithm}

The correctness of this algorithm depends on whether the solution of the linear system in step $\ref{fpsstep}$ has enough equations to determine the possible coefficients of the linear combination sought. Indeed, we saw for $(\ref{hmfold})$ that we need a linear system for which each unknown has enough equations to be determined. This is established by the following lemma.

\begin{lemma}\label{fpscombin} In Algorithm $\ref{FPScombin}$, $\mathcal{N}$ given in $(\ref{Nfps})$ is a valid integer for which the series expansion of order $\mathcal{N}$ of $f(z)$ allows to determine the linear combination sought.
\end{lemma}
\begin{proof}
	The computation is similar for any integer $N_1$, therefore we assume that $N_1=0$. The number of unknowns in each equation is $q=\sum_{m\in\{1,m_1,\ldots,m_{\mu}\}}^{}l_m$. The aim is to find $\mathcal{N}$ such that $Taylor(f(z),z,0,\mathcal{N})$ in Algorithm $\ref{fpsstep}$ step $\ref{fpsstep}$ yields a linear system with at least $q$ equations with $q$ unknowns each. Of course, the minimal value of $\mathcal{N}$ is an integer that verifies
	$$ \mathcal{N} = m_1\cdot x_1 = m_2\cdot x_2=\cdots=m_{\mu}\cdot x_{\mu},$$
	for some positive integers $x_1,x_2,\ldots,x_{\mu}$, since we have to find $q$ equations that correspond to the $q$ first coincidences of
	$$z^{m_1\cdot n},~z^{m_2\cdot n},~\ldots, z^{m_{\mu}\cdot n}.$$
	The second coincidence is reached at the expansion of order $\lcm(m_1,\ldots, m_{\mu})$, therefore by induction we deduce that for any positive integer $p$, the $p^{\text{th}}$ coincidence is reached at the expansion of order $(p-1)\cdot \lcm(m_1,\ldots, m_{\mu})$. Hence we finally get
	$$\mathcal{N} = (q-1)\cdot \lcm(m_1,\ldots, m_{\mu}) =  \left(\sum_{m\in\{1,m_1,\ldots,m_{\mu}\}}^{}l_m - 1\right) \cdot \lcm(m_1,\ldots,\cdot m_{\mu}) + m_{\mu} - 1$$
	where $m_{\mu} - 1$ is added to get similar coincidences with some
	$$z^{m_1\cdot n + j},~z^{m_2\cdot n + j},~\ldots, z^{m_{\mu}+j\cdot n},$$
	with $j\in\llbracket 1, m_1\rrbracket$.
\end{proof}

\begin{remark} In $(\ref{Nfps})$ we use $\lcm(1,m_1,\ldots, m_{\mu})$ because it allows to recover the order $(l_1-1)\cdot 1$ when there are only hypergeometric terms ($\mu=0$).
\end{remark}

As already used many times, the command \textit{FPS(f(z),z,n,[z\_0])} of our Maxima FPS package computes the power series representation of \textit{f(z)} at the point of expansion $z_0\in\mathbb{C}$ (if given or $0$ otherwise) with the index variable \textit{n} by combining \textit{FindRE}, \textit{mfoldHyper}, \textit{Puiseuxnbrfun} and our implementations of Algorithms $\ref{PolyPart}$ and $\ref{twoterm}$ if the computed holonomic RE is a two-term holonomic RE or Algorithm $\ref{FPScombin}$, if \textit{f(z)} leads to a different type of holonomic RE.

\begin{example}\item
	
	\noindent
	\begin{minipage}[t]{8ex}\color{red}\bf
		\begin{verbatim}
			(%i1) 
		\end{verbatim}
	\end{minipage}
	\begin{minipage}[t]{\textwidth}\color{blue}
		\begin{verbatim}
			FPS(sin(z)^2+cos(z)^3,z,n);
		\end{verbatim}
	\end{minipage}
	\definecolor{labelcolor}{RGB}{100,0,0}
	\[\displaystyle
	\parbox{10ex}{$\color{labelcolor}\mathrm{\tt (\%o1) }\quad $}
	\left( \sum_{n=0}^{\infty }-\frac{\left( -{{\left( -9\right) }^{n}}-3\cdot {{\left( -1\right) }^{n}}+2\cdot {{\left( -1\right) }^{n}}\cdot {{4}^{n}}\right) \cdot {{z}^{2\cdot n}}}{4\cdot \left( 2\cdot n\right) !}\right) +\frac{1}{2}\mbox{}
	\]
	
	\noindent
	\begin{minipage}[t]{8ex}\color{red}\bf
		\begin{verbatim}
			(%i2) 
		\end{verbatim}
	\end{minipage}
	\begin{minipage}[t]{\textwidth}\color{blue}
		\begin{verbatim}
			FPS(1/(q1-z^2)/(q1-z^3),z,n);
		\end{verbatim}
	\end{minipage}
	\definecolor{labelcolor}{RGB}{100,0,0}
	\begin{multline}
		\displaystyle \parbox{10ex}{$\color{labelcolor}\mathrm{\tt (\%o2) }\quad $}
		\left( \sum_{n=0}^{\infty }{\left. \frac{{{\mathit{q1}}^{-n-2}}\, {{z}^{3 n+2}}}{\mathit{q1}-1}\right.}\right) +\left( \sum_{n=0}^{\infty }{\left. \frac{{{\mathit{q1}}^{-n-2}}\, {{z}^{3 n+1}}}{\mathit{q1}-1}\right.}\right) +\left( \sum_{n=0}^{\infty }{\left. -\frac{{{\mathit{q1}}^{-n-2}}\, {{z}^{2 n+1}}}{\mathit{q1}-1}\right.}\right) \\
		+\left( \sum_{n=0}^{\infty }{\left. \frac{{{\mathit{q1}}^{-n-1}}\, {{z}^{3 n}}}{\mathit{q1}-1}\right.}\right) +\sum_{n=0}^{\infty }{\left. -\frac{{{\mathit{q1}}^{-n-2}}\, {{z}^{2 n}}}{\mathit{q1}-1}\right.}\nonumber
	\end{multline}
	
	\noindent
	\begin{minipage}[t]{8ex}\color{red}\bf
		\begin{verbatim}
			(%i3) 
		\end{verbatim}
	\end{minipage}
	\begin{minipage}[t]{\textwidth}\color{blue}
		\begin{verbatim}
			FPS(sin(z^(1/3))+ cos(z^(1/2)),z,n);
		\end{verbatim}
	\end{minipage}
	\definecolor{labelcolor}{RGB}{100,0,0}
	\[\displaystyle
	\parbox{10ex}{$\color{labelcolor}\mathrm{\tt (\%o3) }\quad $}
	\left( \sum_{n=0}^{\infty }\frac{{{\left( -1\right) }^{n}}\cdot {{z}^{\frac{1+2\cdot n}{3}}}}{\left( 2\cdot n+1\right) \cdot \left( 2\cdot n\right) !}\right) +\sum_{n=0}^{\infty }\frac{{{\left( -1\right) }^{n}}\cdot {{z}^{n}}}{\left( 2\cdot n\right) !}\mbox{}
	\]
	
	\noindent
	\begin{minipage}[t]{8ex}\color{red}\bf
		\begin{verbatim}
			(%i4) 
		\end{verbatim}
	\end{minipage}
	\begin{minipage}[t]{\textwidth}\color{blue}
		\begin{verbatim}
			FPS(acos(z^(1/2))+exp(z^2),z,n);
		\end{verbatim}
	\end{minipage}
	\definecolor{labelcolor}{RGB}{100,0,0}
	\[\displaystyle
	\parbox{10ex}{$\color{labelcolor}\mathrm{\tt (\%o4) }\quad $}
	\left( \sum_{n=0}^{\infty }-\frac{\left( 2\cdot n\right) !\cdot {{z}^{\frac{1+2\cdot n}{2}}}}{\left( 2\cdot n+1\right) \cdot {{4}^{n}}\cdot {{n!}^{2}}}\right) +\left( \sum_{n=0}^{\infty }\frac{{{z}^{2\cdot n}}}{n!}\right) +\frac{2+\pi }{2}-1\mbox{}
	\]
	
	\noindent
	\begin{minipage}[t]{8ex}\color{red}\bf
		\begin{verbatim}
			(%i5) 
		\end{verbatim}
	\end{minipage}
	\begin{minipage}[t]{\textwidth}\color{blue}
		\begin{verbatim}
			FPS(log(1+sqrt(z)+z+z^(3/2)),z,n);
		\end{verbatim}
	\end{minipage}
	\definecolor{labelcolor}{RGB}{100,0,0}
	\[\displaystyle
	\parbox{10ex}{$\color{labelcolor}\mathrm{\tt (\%o5) }\quad $}
	\left( \sum_{n=0}^{\infty }{\left. \frac{{{\left( -1\right) }^{n}}\, {{z}^{n+1}}}{n+1}\right.}\right) +\sum_{n=0}^{\infty }{\left. \frac{{{\left( -1\right) }^{n}}\, {{z}^{\frac{n+1}{2}}}}{n+1}\right.}\]
	
	\textnormal{Let us use other points of expansions.}
	
	\noindent
	\begin{minipage}[t]{8ex}\color{red}\bf
		\begin{verbatim}
			(%i6) 
		\end{verbatim}
	\end{minipage}
	\begin{minipage}[t]{\textwidth}\color{blue}
		\begin{verbatim}
			FPS(sin(2*z)+cos(z),z,n,%pi/2);
		\end{verbatim}
	\end{minipage}
	\definecolor{labelcolor}{RGB}{100,0,0}
	\[\displaystyle
	\parbox{10ex}{$\color{labelcolor}\mathrm{\tt (\%o6) }\quad $}
	-\sum_{n=0}^{\infty }\frac{{{\left( -1\right) }^{n}}\cdot \left( 1+2\cdot {{4}^{n}}\right) \cdot {{\left( z-\frac{\pi }{2}\right) }^{1+2\cdot n}}}{\left( 2\cdot n+1\right) \cdot \left( 2\cdot n\right) !}\mbox{}
	\]
	
	\noindent
	\begin{minipage}[t]{8ex}\color{red}\bf
		\begin{verbatim}
			(%i7) 
		\end{verbatim}
	\end{minipage}
	\begin{minipage}[t]{\textwidth}\color{blue}
		\begin{verbatim}
			FPS(exp(z)+log(1+z),z,n,%e);
		\end{verbatim}
	\end{minipage}
	\definecolor{labelcolor}{RGB}{100,0,0}
	\[\displaystyle
	\parbox{10ex}{$\color{labelcolor}\mathrm{\tt (\%o7) }\quad $}
	\left( \sum_{n=0}^{\infty }\frac{\left( {{e}^{e}}+{{e}^{1+e}}+{{e}^{e}}\cdot n+{{e}^{1+e}}\cdot n+\frac{{{\left( -1\right) }^{n}}\cdot \left( 1+n\right) !}{{{\left( 1+e\right) }^{n}}}\right) \cdot {{\left( z-e\right) }^{1+n}}}{\left( e+1\right) \cdot \left( n+1\right) \cdot \left( 1+n\right) !}\right) +\mathrm{log}\left( e+1\right) +{{e}^{e}}\mbox{}
	\]
	\textnormal{The latter series is wrongly represented by Maple's current \textit{convert/FormalPowerSeries} Maple's command due to the missing term} $\mathrm{log}\left( e+1\right)$ $+{{e}^{e}}$.
\end{example}

Algorithm \ref{fpscombin} ends our direct method for hypergeometric type series given in Definition \ref{def1}. In the next section we present a general approach for expressions that are not holonomic.

\section{Non-holonomic power series}\label{sec4}

We described in sub-subsection \ref{normalform} that holonomic recurrence equations together with sufficient initial values can be used to identify holonomic functions. In this section we give an approach to represent the power series of expressions that satisfy homogeneous quadratic differential equations. For a given expression $f$, the procedure follows the following steps:
\begin{enumerate}
	\item compute a quadratic differential equation satisfied by $f$;
	\item use the Cauchy product rule to convert that quadratic differential equation to a non-holonomic recurrence equation satisfied by the series coefficients of $f$;
	\item use the obtained recurrence equation to define a recursive formula for the power series coefficients of $f$.
\end{enumerate}
Gathering all these steps together we are able to define a normal form representation of non-holonomic power series.

\subsection{Computing quadratic differential equations}

Let $f$ be an expression, our algorithm in subsection \ref{mymethodDE} searches for a holonomic differential equation for $f$ by iteration on the order of derivatives of $f$. Depending on the algebraic simplification of ratios of these derivatives, the algorithm finds a holonomic DE of lowest order satisfied by $f$. Our idea in computing homogeneous quadratic differential equation is to define a 'natural' ordering between products of derivatives of $f$ so that by iteration on this ordering a quadratic differential equation satisfied by $f$ is sought.

What we refer to as quadratic differential equation is an algebraic ordinary differential equation (AODE) of degree $2$. This is a DE with at most one product of two derivatives in at least one of its summands (differential monomials). We recall that any product of derivatives increases the degree of a differential monomial by $1$ and therefore the degree of an AODE is the greatest number of products of derivatives among its differential monomials plus $1$ (see \cite[Section 6]{eremenko1982meromorphic}). In particular, linear differential equations are of degree $1$. However, for a good definition of the ordering we are looking for, we first study the general AODE case.

Let $f(z)$ be a given differentiable function. By observing the product rule when applying the derivative operator $\frac{d}{dz}$
\begin{equation}
	\text{(a)}~\dfrac{d}{dz}\left(f(z)^2\right)= 2 \cdot f(z)\cdot \dfrac{d}{dz}f(z),~~~~~ \text{(b)}~\dfrac{d}{dz}\left(\dfrac{1}{f(z)}\right)=\dfrac{-\frac{d}{dz}f(z)}{f(z)^2}, \label{qdiff}
\end{equation}
one can assume that the maximum degree for $f(z)$ in an AODE involving $f(z)$ and $\dfrac{d}{dz}f(z)$ is $2$.

Now, differentiating the right-hand sides of $(a)$ and $(b)$ in $(\ref{qdiff})$ yields $(c)$ and $(d)$, respectively, as below. 
\begin{equation}
	\text{(c)}~2\cdot f(z) \cdot \dfrac{d}{dz^2}f(z) + 2\cdot \left(\dfrac{d}{dz}f(z)\right)^2,~~~~~ \text{(d)}~ \dfrac{2\cdot \left(\dfrac{d}{dz}f(z)\right)^2}{f(z)^3} - \dfrac{\dfrac{d^2}{dz^2}f(z)}{f(z)^2}.
\end{equation}
Thus one can assume that the maximum degrees for $f(z)$ and $\dfrac{d}{dz}f(z)$ in an AODE involving $f(z)$, $\dfrac{d}{dz}f(z)$, and $\dfrac{d^2}{dz^2}f(z)$ are, respectively, $3$ and $2$.

Using this process recursively, we can state that 
\begin{itemize}
	\item $f(z)^4,$ $\left(\dfrac{d}{dz}f(z)\right)^3$ and $\left(\dfrac{d^2}{dz^2}f(z)\right)^2$ gives the maximum degrees for an AODE of order\footnote{As in the linear case, the order is taken as the greatest derivative order.} $3$;
	\item $f(z)^5,$ $\left(\dfrac{d}{dz}f(z)\right)^4$, $\left(\dfrac{d^2}{dz^2}f(z)\right)^3$ and $\left(\dfrac{d^2}{dz^2}f(z)\right)^2$ gives the maximum degrees for an AODE of order $4$;
	\item \ldots
\end{itemize}

However, since we are only interested in computing quadratic differential equations, we modify the above process by avoiding degrees that are greater than $2$. Thus, we can order quadratic derivatives of $f$ as follows

\begin{equation}
	\begin{matrix}
		(1)~ 1,&&&&\\
		(2)~ f,& (3)~ f^2,& & & \\
		(4)~ f',& (5)~ f'f,& (6)~ f'^2,& & \\
		(7)~ f'',& (8)~ f''f,& (9)~ f''f',& (10) f''^2, & \\
		(11)~ f''',& (12)~ f'''f,& (13)~ f'''f',& (14)~ f'''f'',& (15) f'''^2,\\
		\ldots&&&&
	\end{matrix}
	\label{nu2}
\end{equation}
We define
\begin{equation}
	\frac{d^{-1}}{dz}=f^{(-1)}=1, ~\text{ and }~\frac{d^{0}}{dz}=f^{(0)}=f. \label{Qdef1}
\end{equation}

Observe in $(\ref{nu2})$ that for every derivative of order $n$, $n\in\mathbb{N}_{\geqslant 0}$, we compute the product of $f^{(n)}$ and all the derivatives of order less than or equal to $n$ before computing the next derivative. We are going to define a derivative operator, say $\delta_{2,z}^{k} f$, $k\in\mathbb{N}$ so that in $(\ref{nu2})$ the numbers in parenthesis represent the derivative orders. This operator computes the product of two derivatives of $f$ according to the ordering given in $(\ref{nu2})$.

Looking at $(\ref{nu2})$ as an infinite lower triangular matrix reduces the definition of $\delta_{2,z}$ to the one of a bijective map $\nu$ between positive integers and the corresponding subspace of $\mathbb{N}\times\mathbb{N}$: $(i,j)_{i,j\in\mathbb{N}}, i\leqslant j$. This can be done by counting the couple $(i,j)$ in $(\ref{nu2})$ from up to down, from the left to the right. We obtain
\begin{equation}
	\nu(k) = (i,j) = \begin{cases} (l,l)~~ \text{ if }~~ N=k\\ (l+1,k-N)~\text{ oherwise}\end{cases},
	\text{ where } l=\left\lfloor \sqrt{2k+\frac{1}{4}}-\frac{1}{2}\right\rfloor,~\text{ and }~N=\dfrac{l(l+1)}{2}.
\end{equation}
It remains to define a correspondence between the couple $(i,j)=\nu(k),~k\in\mathbb{N}$ and the quadratic products in $(\ref{nu2})$. This is straightforward since we have defined $(\ref{Qdef1})$. We get
\begin{equation}
	\delta_{2,z}^k(f) = \dfrac{d^{i-2}}{dz^{j-2}}f \cdot \dfrac{d^{j-2}}{dz^{i-2}},~\text{ where }~ (i,j) = \nu(k).
\end{equation}

We implemented this operator in our packages as \textit{delta2diff(f,z,k)}. One can use it to recover some products of derivatives in $(\ref{nu2})$.

\begin{example}\item
	
	\noindent
	\begin{minipage}[t]{8ex}\color{red}\bf
		\begin{verbatim}
			(%i1) 
		\end{verbatim}
	\end{minipage}
	\begin{minipage}[t]{\textwidth}\color{blue}
		\begin{verbatim}
			delta2diff(F(z),z,3);
		\end{verbatim}
	\end{minipage}
	\vspace{-0.25cm}
	\definecolor{labelcolor}{RGB}{100,0,0}
	\[\displaystyle
	\hspace{1cm}\parbox{10ex}{$\color{labelcolor}\mathrm{\tt (\%o1) }\quad $}\hspace{1cm}
	{{\mathrm{F}\left( z\right) }^{2}}\mbox{}
	\]
	\vspace{-0.25cm}
	
	\noindent
	\begin{minipage}[t]{8ex}\color{red}\bf
		\begin{verbatim}
			(%i2) 
		\end{verbatim}
	\end{minipage}
	\begin{minipage}[t]{\textwidth}\color{blue}
		\begin{verbatim}
			delta2diff(F(z),z,4);
		\end{verbatim}
	\end{minipage}
	\vspace{-0.25cm}
	\definecolor{labelcolor}{RGB}{100,0,0}
	\[\displaystyle
	\hspace{1cm}\parbox{10ex}{$\color{labelcolor}\mathrm{\tt (\%o2) }\quad $}\hspace{1cm}
	\frac{d}{d\,z}\cdot \mathrm{F}\left( z\right) \mbox{}
	\]
	\vspace{-0.25cm}
	
	\noindent
	\begin{minipage}[t]{8ex}\color{red}\bf
		\begin{verbatim}
			(%i3) 
		\end{verbatim}
	\end{minipage}
	\begin{minipage}[t]{\textwidth}\color{blue}
		\begin{verbatim}
			delta2diff(F(z),z,5);
		\end{verbatim}
	\end{minipage}
	\vspace{-0.25cm}
	\definecolor{labelcolor}{RGB}{100,0,0}
	\[\displaystyle
	\hspace{1cm}\parbox{10ex}{$\color{labelcolor}\mathrm{\tt (\%o3) }\quad $}\hspace{1cm}
	\mathrm{F}\left( z\right) \cdot \left( \frac{d}{d\,z}\cdot \mathrm{F}\left( z\right) \right) \mbox{}
	\]
	\vspace{-0.25cm}
	
	\noindent
	\begin{minipage}[t]{8ex}\color{red}\bf
		\begin{verbatim}
			(%i4) 
		\end{verbatim}
	\end{minipage}
	\begin{minipage}[t]{\textwidth}\color{blue}
		\begin{verbatim}
			delta2diff(F(z),z,6);
		\end{verbatim}
	\end{minipage}
	\vspace{-0.25cm}
	\definecolor{labelcolor}{RGB}{100,0,0}
	\[\displaystyle
	\hspace{1cm}\parbox{10ex}{$\color{labelcolor}\mathrm{\tt (\%o4) }\quad $}\hspace{1cm}
	{{\left( \frac{d}{d\,z}\cdot \mathrm{F}\left( z\right) \right) }^{2}}\mbox{}
	\]
	\vspace{-0.25cm}
	
	\noindent
	\begin{minipage}[t]{8ex}\color{red}\bf
		\begin{verbatim}
			(%i5) 
		\end{verbatim}
	\end{minipage}
	\begin{minipage}[t]{\textwidth}\color{blue}
		\begin{verbatim}
			delta2diff(F(z),z,14);
		\end{verbatim}
	\end{minipage}
	\vspace{-0.25cm}
	\definecolor{labelcolor}{RGB}{100,0,0}
	\[\displaystyle
	\hspace{1cm}\parbox{10ex}{$\color{labelcolor}\mathrm{\tt (\%o5) }\quad $} \hspace{1cm}
	\left( \frac{{{d}^{2}}}{d\,{{z}^{2}}}\cdot \mathrm{F}\left( z\right) \right) \cdot \left( \frac{{{d}^{3}}}{d\,{{z}^{3}}}\cdot \mathrm{F}\left( z\right) \right) \mbox{}
	\]
\end{example}

Using $\delta_{2,z}$ instead of $\dfrac{d}{dz}$ in Koepf's original approach for holonomic functions yields a procedure to compute homogeneous quadratic differential equations generally of lowest order satisfied by a given expression $f(z)$. We therefore obtain the following algorithm.

\begin{algorithm}[h!]
	\caption{Computing a quadratic DE satisfied by an expression $f$}
	\label{AlgoQDE}
	\begin{algorithmic} 
		\Require An expression $f(z)$.
		\Ensure A quadratic differential equation over $\mathbb{K}$ of least order satisfied by $f(z)$.
		\begin{enumerate}
			\item If $f=0$ then the DE is found\footnotemark and we stop.
			\item $f\neq0$, compute $A_0(z)=\delta_{2,z}^3 f(z) / f(z),$
			\begin{itemize}
				\item[(1-a)] if $A_0(z)\in\mathbb{K}(z)$ i.e $A_0(z)=P(z)/Q(z)$ where $P$ and $Q$ are polynomials, then we have found a quadratic DE satisfied by $f$:
				$$Q(z)F(z)^2 - P(z)F(z) = 0.$$
				\item[(1-b)] If $A_0(z)\notin \mathbb{K}(z)$, then go to 3.
			\end{itemize}
			\item Fix a number $QN_{max}\in \mathbb{N},$ the maximal order of the DE sought; a suitable value is $QN_{max}:=19$ which corresponds to the maximum $\delta_{2,z}$-order for having a quadratic differential equation of forth order.
			\begin{itemize}
				\item[(3-a)] set $N:=2$;
				\item[(3-b)] compute $\delta_{2,z}^{N+2}f;$ 
				\item[(3-c)] expand the ansatz
				\begin{equation}
					\delta_{2,z}^{N+2}f(z) + A_{N-1} \delta_{2,z}^{N+1}f(z) + \cdots + A_0 f(z) = \sum_{i=0}^{E}S_i,
				\end{equation}
				in elementary summands with $A_{N}, A_{N-1},\ldots, A_{0}$ as unknowns. $E\geqslant N$ is the total number of summands $S_i$ obtained after expansion.
			\end{itemize}
		\end{enumerate}
		\algstore{pause3b}
	\end{algorithmic}
\end{algorithm}
\footnotetext[11]{A differential equation of order zero: $F=0$, since we assumed that $\frac{d^0}{dz}f=f$.}
\clearpage 

\begin{algorithm}
	\ContinuedFloat
	\caption{Computing a quadratic DE satisfied by an expression $f$}
	\begin{algorithmic}
		\algrestore{pause3b}	
		\State
		\begin{itemize}			
			\item[(3-d)]\label{Algo13d} For each pair of summands $S_i$ and $S_j$ $(0\leqslant i\neq j\leqslant E)$, group them additively together if there exists $r(z)=S_i(z)/S_j(z)\in\mathbb{K}(z)$. If the number of groups is $N$ then we have $N$ linearly independent expressions. In that case, there exists a solution which can be found by equating each group to zero. The resulting system is linear for the unknowns $A_0, A_1,\ldots, A_{N-1}.$ Solving this system gives rational functions in $z$, and the solution is unique since we normalized $A_N=1$. After multiplication by the common denominator of the values found for $A_0(z), A_1(z),\ldots A_{N-1}(z)$ we get the holonomic DE sought. If otherwise the number of groups is larger than $N$, then there is no solution and the step is not successful.
			\item[(3-e)] If (3-d) is not successful, then increment $N$, and go back to (3-b), until $N=QN_{max}.$
		\end{itemize}
     % \end{enumerate}	
	\end{algorithmic}
\end{algorithm}

Note that an algorithm for the general AODE case can be defined similarly provided an appropriate replacement of $\delta_{2,z}$ for derivation. Our Maxima package has an implementation of Algorithm $\ref{AlgoQDE}$ with the syntax \textit{QDE(f(z),F(z),[Type])}. The argument \textit{Type} is either \textit{Inhomogeneous} to allow the search for inhomogeneous quadratic DEs, or \textit{Homogeneous} by default to look for homogeneous ones. We have also implemented the general AODE case as \textit{NLDE(f(z),F(z))} (NL for non-linear).

\begin{example}\label{qdeegs}\item
	
	\noindent
	\begin{minipage}[t]{8ex}\color{red}\bf
		\begin{verbatim}
			(%i1) 
		\end{verbatim}
	\end{minipage}
	\begin{minipage}[t]{\textwidth}\color{blue}
		\begin{verbatim}
			QDE(tan(z),F(z));
		\end{verbatim}
	\end{minipage}
	\definecolor{labelcolor}{RGB}{100,0,0}
	\[\displaystyle
	\parbox{10ex}{$\color{labelcolor}\mathrm{\tt (\%o1) }\quad $}
	\frac{{{d}^{2}}}{d\,{{z}^{2}}}\cdot \mathrm{F}\left( z\right) -2\cdot \mathrm{F}\left( z\right) \cdot \left( \frac{d}{d\,z}\cdot \mathrm{F}\left( z\right) \right) =0\mbox{}
	\]
	
	\noindent
	\begin{minipage}[t]{8ex}\color{red}\bf
		\begin{verbatim}
			(%i2) 
		\end{verbatim}
	\end{minipage}
	\begin{minipage}[t]{\textwidth}\color{blue}
		\begin{verbatim}
			QDE(tan(z),F(z),Inhomogeneous);
		\end{verbatim}
	\end{minipage}
	\definecolor{labelcolor}{RGB}{100,0,0}
	\[\displaystyle
	\parbox{10ex}{$\color{labelcolor}\mathrm{\tt (\%o2) }\quad $}
	\frac{d}{d\,z}\cdot \mathrm{F}\left( z\right) -{{\mathrm{F}\left( z\right) }^{2}}-1=0\mbox{}
	\]

	\noindent
	\begin{minipage}[t]{8ex}\color{red}\bf
		\begin{verbatim}
			(%i3) 
		\end{verbatim}
	\end{minipage}
	\begin{minipage}[t]{\textwidth}\color{blue}
		\begin{verbatim}
			QDE(sec(z)^k,F(z));
		\end{verbatim}
	\end{minipage}
	\definecolor{labelcolor}{RGB}{100,0,0}
	\[\displaystyle
	\parbox{10ex}{$\color{labelcolor}\mathrm{\tt (\%o3) }\quad $}
	k\cdot \mathrm{F}\left( z\right) \cdot \left( \frac{{{d}^{2}}}{d\,{{z}^{2}}}\cdot \mathrm{F}\left( z\right) \right) +\left( -1-k\right) \cdot {{\left( \frac{d}{d\,z}\cdot \mathrm{F}\left( z\right) \right) }^{2}}-{{k}^{2}}\cdot {{\mathrm{F}\left( z\right) }^{2}}=0\mbox{}
	\]
	
	\noindent
	\begin{minipage}[t]{8ex}\color{red}\bf
		\begin{verbatim}
			(%i4) 
		\end{verbatim}
	\end{minipage}
	\begin{minipage}[t]{\textwidth}\color{blue}
		\begin{verbatim}
			QDE(z/(exp(z)-1),F(z));
		\end{verbatim}
	\end{minipage}
	\definecolor{labelcolor}{RGB}{100,0,0}
	\[\displaystyle
	\parbox{10ex}{$\color{labelcolor}\mathrm{\tt (\%o4) }\quad $}
	z\cdot \left( \frac{d}{d\,z}\cdot \mathrm{F}\left( z\right) \right) +{{\mathrm{F}\left( z\right) }^{2}}+\left( z-1\right) \cdot \mathrm{F}\left( z\right) =0\mbox{}
	\]
	
	\noindent
	\begin{minipage}[t]{8ex}\color{red}\bf
		\begin{verbatim}
			(%i5) 
		\end{verbatim}
	\end{minipage}
	\begin{minipage}[t]{\textwidth}\color{blue}
		\begin{verbatim}
			QDE(log(1+sin(z)),F(z));
		\end{verbatim}
	\end{minipage}
	\definecolor{labelcolor}{RGB}{100,0,0}
	\[\displaystyle
	\parbox{10ex}{$\color{labelcolor}\mathrm{\tt (\%o5) }\quad $}
	\frac{{{d}^{3}}}{d\,{{z}^{3}}}\cdot \mathrm{F}\left( z\right) +\left( \frac{d}{d\,z}\cdot \mathrm{F}\left( z\right) \right) \cdot \left( \frac{{{d}^{2}}}{d\,{{z}^{2}}}\cdot \mathrm{F}\left( z\right) \right) =0\mbox{}
	\]
	
	\noindent
	\begin{minipage}[t]{8ex}\color{red}\bf
		\begin{verbatim}
			(%i6) 
		\end{verbatim}
	\end{minipage}
	\begin{minipage}[t]{\textwidth}\color{blue}
		\begin{verbatim}
			QDE(tan(z)^k,F(z));
		\end{verbatim}
	\end{minipage}
	\definecolor{labelcolor}{RGB}{100,0,0}
	\begin{multline*}
		\displaystyle
		\parbox{10ex}{$\color{labelcolor}\mathrm{\tt (\%o6) }\quad $}\hspace{-0.4cm}
		{{k}^{2}}\cdot \mathrm{F}\left( z\right) \cdot \left( \frac{{{d}^{4}}}{d\,{{z}^{4}}}\cdot \mathrm{F}\left( z\right) \right) -2\cdot \left( 2\cdot {{k}^{2}}-3\right) \cdot \left( \frac{d}{d\,z}\cdot \mathrm{F}\left( z\right) \right) \cdot \left( \frac{{{d}^{3}}}{d\,{{z}^{3}}}\cdot \mathrm{F}\left( z\right) \right) +3\cdot \left( k-2\right) \\
		\hspace{-0.15cm} \cdot \left( 2+k\right) \cdot {{\left( \frac{{{d}^{2}}}{d\,{{z}^{2}}}\cdot \mathrm{F}\left( z\right) \right) }^{2}}\hspace{-0.2cm}+4\cdot {{k}^{2}}\cdot \mathrm{F}\left( z\right) \cdot \left( \frac{{{d}^{2}}}{d\,{{z}^{2}}}\cdot \mathrm{F}\left( z\right) \right) +4\cdot \left( 5\cdot {{k}^{2}}-6\right) \cdot {{\left( \frac{d}{d\,z}\cdot \mathrm{F}\left( z\right) \right) }^{2}}\hspace{-0.3cm}=0\mbox{}
	\end{multline*}
	
	\textnormal{Using \textit{NLDE} in the latter example yields an AODE that does not depend on the exponent $k$, but is not quadratic}.
	
	\noindent
	\begin{minipage}[t]{8ex}\color{red}\bf
		\begin{verbatim}
			(%i7) 
		\end{verbatim}
	\end{minipage}
	\begin{minipage}[t]{\textwidth}\color{blue}
		\begin{verbatim}
			NLDE(tan(z)^k,F(z));
		\end{verbatim}
	\end{minipage}
	\definecolor{labelcolor}{RGB}{100,0,0}
	\vspace{-0.5cm}
	\begin{multline*}
		\displaystyle
		\parbox{10ex}{$\color{labelcolor}\mathrm{\tt (\%o7) }\quad $}
		\mathrm{F}\left( z\right) \cdot \left( \frac{d}{d\,z}\cdot \mathrm{F}\left( z\right) \right) \cdot \left( \frac{{{d}^{3}}}{d\,{{z}^{3}}}\cdot \mathrm{F}\left( z\right) \right) -2\cdot \mathrm{F}\left( z\right) \cdot {{\left( \frac{{{d}^{2}}}{d\,{{z}^{2}}}\cdot \mathrm{F}\left( z\right) \right) }^{2}}\\
		+{{\left( \frac{d}{d\,z}\cdot \mathrm{F}\left( z\right) \right) }^{2}}\cdot \left( \frac{{{d}^{2}}}{d\,{{z}^{2}}}\cdot \mathrm{F}\left( z\right) \right) -4\cdot \mathrm{F}\left( z\right) \cdot {{\left( \frac{d}{d\,z}\cdot \mathrm{F}\left( z\right) \right) }^{2}}=0\mbox{}
	\end{multline*}
	
	\textnormal{Compared to \textit{QDE}, generally \textit{NLDE} generates differential equations of lower order but of higher degree. However, in terms of timings \textit{QDE} is faster and the computed differential equations give much simpler recurrence equations than the outputs of \textit{NLDE}.}
\end{example}

\begin{remark}
Observe that unlike the holonomic case where the existence and uniqueness of a solution to the Cauchy problem is quite immediate, in the algebraic case one needs to take into account other important facts since the computed differential equations are not always explicit. However, using the implicit function theorem (see \cite{krantz2012implicit}) on underlying algebraic polynomials, the classical existence and uniqueness theorem \cite[Theorem 2.2]{teschl2012ordinary} can be applied locally to uniquely determine the solution of outputs of \textit{QDE} as a well chosen initial value problem.	
\end{remark}

\subsection{Converting quadratic differential equations to recurrence equations}

We need a rewrite rule similar to $(\ref{Algo3corr})$ for every differential monomial in the expansion of a quadratic differential equation. Let $f(z)$ be a power series with representation 
$$f(z) = \sum_{n=0}^{\infty} a_nz^n.$$
It suffices to find a recurrence equation term that corresponds to the  quadratic differential equation term
\begin{equation}
	z^p \cdot f(z)^{(i)}\cdot f(z)^{(j)}, ~i,j,p\in\mathbb{N}_{\geqslant 0}.
\end{equation}
This is done in a similar manner as for $(\ref{Algo3corr})$. We have
$$   f(z)^{(k)} = \sum_{n=0}^{\infty} (n+1)_k\cdot a_{n+k} \cdot z^n, \forall k\in\mathbb{N}_{\geqslant},$$
therefore
\begin{eqnarray}
	f(z)^{(i)}\cdot f(z)^{(j)} &=& \left(\sum_{n=0}^{\infty} (n+1)_i\cdot a_{n+i} \cdot z^n \right) \cdot \left(\sum_{n=0}^{\infty} (n+1)_j\cdot a_{n+j} \cdot z^n\right)\nonumber\\
	&=& \sum_{n=0}^{\infty} \left( \sum_{k=0}^{n} (k+1)_i\cdot a_{k+i} \cdot (n-k+1)_j\cdot a_{n-k+j}\right) \cdot z^n, \label{qrule1}
\end{eqnarray}
by application of the Cauchy product rule. Finally multiplying $(\ref{qrule1})$ by $z^p$ yields the formula

\begin{equation}
	z^p \cdot f(z)^{(i)}\cdot f(z)^{(j)} = \sum_{n=0}^{\infty} \left(\sum_{k=0}^{n-p}(k+1)_i\cdot (n-p-k+1)_j\cdot a_{k+i}\cdot a_{n-p-k+j}\right) \cdot z^n,
\end{equation}
and the corresponding rewrite rule
\begin{equation}
	z^p \cdot f(z)^{(i)}\cdot f(z)^{(j)} \longrightarrow \left(\sum_{k=0}^{n-p}(k+1)_i\cdot (n-p-k+1)_j\cdot a_{k+i}\cdot a_{n-p-k+j}\right). \label{qrerule}
\end{equation}

Observe that $(\ref{qrerule})$ is a rewrite rule for the power series coefficients of the given expression $f$. When dealing with inhomogeneous DEs, the constant term must be considered differently. This is the main reason why we prefer to work with homogeneous DEs.

Thus a procedure to convert quadratic differential equations to recurrence equations follows immediately. Our packages contains the function \textit{FindQRE(f,z,a[n])} as analogue of \textit{FindRE} for the quadratic case. 

\begin{example}\label{exqre}\item
	
	\noindent
	\begin{minipage}[t]{8ex}\color{red}\bf
		\begin{verbatim}
			(%i1) 
		\end{verbatim}
	\end{minipage}
	\begin{minipage}[t]{\textwidth}\color{blue}
		\begin{verbatim}
			FindQRE(tan(z),z,a[n]);
		\end{verbatim}
	\end{minipage}
	\definecolor{labelcolor}{RGB}{100,0,0}
	\[\displaystyle
	\parbox{10ex}{$\color{labelcolor}\mathrm{\tt (\%o1) }\quad $}
	\left( 1+n\right) \cdot \left( 2+n\right) \cdot {{a}_{n+2}}-2\cdot \sum_{k=0}^{n}\left( k+1\right) \cdot {{a}_{k+1}}\cdot {{a}_{n-k}}=0
	\]
	
	\noindent
	\begin{minipage}[t]{8ex}\color{red}\bf
		\begin{verbatim}
			(%i2) 
		\end{verbatim}
	\end{minipage}
	\begin{minipage}[t]{\textwidth}\color{blue}
		\begin{verbatim}
			FindQRE(z/(exp(z)-1),z,a[n]);
		\end{verbatim}
	\end{minipage}
	\definecolor{labelcolor}{RGB}{100,0,0}
	\[\displaystyle
	\parbox{10ex}{$\color{labelcolor}\mathrm{\tt (\%o2) }\quad $}
	\left( \sum_{k=0}^{n}{{a}_{k}}\cdot {{a}_{n-k}}\right) +\left( n-1\right) \cdot {{a}_{n}}+{{a}_{n-1}}=0\mbox{}
	\]
	
	\noindent
	\begin{minipage}[t]{8ex}\color{red}\bf
		\begin{verbatim}
			(%i3) 
		\end{verbatim}
	\end{minipage}
	\begin{minipage}[t]{\textwidth}\color{blue}
		\begin{verbatim}
			FindQRE(log(1+sin(z)),z,a[n]);
		\end{verbatim}
	\end{minipage}
	\definecolor{labelcolor}{RGB}{100,0,0}
	\vspace{-0.5cm}
	\begin{multline*}
		\displaystyle
		\parbox{10ex}{$\color{labelcolor}\mathrm{\tt (\%o3) }\quad $}\hspace{-0.5cm}
		\left( \sum_{k=0}^{n}\left( k+1\right) \cdot \left( k+2\right) \cdot {{a}_{k+2}}\cdot \left( n-k+1\right) \cdot {{a}_{n-k+1}}\right)+\left( 1+n\right) \cdot \left( 2+n\right) \cdot \left( 3+n\right) \cdot {{a}_{n+3}}=0
	\end{multline*}
\end{example}

\subsection{Normal forms for non-holonomic power series}

One can use the recurrence equation computed by \textit{FindQRE} to define a power series representation of a given holonomic or non-holonomic expression. For a given recurrence equation computed using \textit{FindQRE}, we write the highest order term in terms of the others. And evaluating the recurrence equation at some integers allows to determine the necessary initial values of the representation. 

Note that to get the highest order term when there are terms with symbolic sums in the recurrence equation, we remove parts corresponding to the minimum and the maximum value of the summation variable and substitute the initial conditions until we get a non-zero expression from which the highest order term can be obtained. For example, to get the highest order term of the recurrence equation of $z/(\exp(z)-1)$
\begin{equation}
	\left( \sum_{k=0}^{n}{{a}_{k}}\cdot {{a}_{n-k}}\right) +\left( n-1\right) \cdot {{a}_{n}}+{{a}_{n-1}}=0\mbox{},
\end{equation}
we remove those parts of the symbolic sum corresponding to $k=0$ and $k=n$. This gives
\begin{equation}
	\left( \sum_{k=1}^{n-1}{{a}_{k}}\cdot {{a}_{n-k}}\right) +   2\cdot a_0\cdot a_n + \left( n-1\right) \cdot {{a}_{n}}+{{a}_{n-1}}=0\mbox{}.
\end{equation}
Then we substitute the value of $a_0$ and write the resulting highest order term in terms of the other summands of the equation. In this example $a_n$ is necessarily the highest order term to be used since the part of the equation that has no sum always depends on $a_n$ after substitution of the value of $a_0$.

Observe that this process of determining the highest order term can quite easily be managed in the quadratic case because every term in the computed recurrence equation has at most one sum. In the general AODE case we generally have a more complicated situation where some summands of the equation have many summation symbols. 

The FPS command of our Maxima package combines all the procedures of this section as the last method to determine the power series representation of a given expression.

\begin{example}\item
	
	\noindent
	\begin{minipage}[t]{8ex}\color{red}\bf
		\begin{verbatim}
			(%i1) 
		\end{verbatim}
	\end{minipage}
	\begin{minipage}[t]{\textwidth}\color{blue}
		\begin{verbatim}
			FPS(log(1+sin(z)),z,n);
		\end{verbatim}
	\end{minipage}
	\definecolor{labelcolor}{RGB}{100,0,0}
	\begin{small}
		\begin{multline*}
		\displaystyle
		\parbox{10ex}{$\color{labelcolor}\mathrm{\tt (\%o1) }\quad $}\hspace{-0.5cm}
		\Bigg[\sum_{n=0}^{\infty }{{A}_{n}}\cdot {{z}^{n}},{{A}_{n+4}}=\frac{1}{\left( n+2\right) \cdot \left( n+3\right) \cdot \left( n+4\right)}\Bigg(\left( n+2\right) \cdot {{A}_{n+2}}-\left( 2+n\right) \cdot \left(  3+n\right) \cdot {{A}_{n+3}}\\
		\hspace{-0.7cm}-\sum_{k=1}^{n}\left( k+1\right) \cdot \left( k+2\right) \cdot {{A}_{k+2}}\cdot \left(   n-k+2\right) \cdot {{A}_{n-k+2}}\Bigg),n>=0,
		\left[{{A}_{0}}=0,{{A}_{1}}=1,{{A}_{2}}=-\frac{1}{2},{{A}_{3}}=\frac{1}{6}\right]\Bigg]\mbox{}
	\end{multline*}
	\end{small}
	
	\noindent
	\begin{minipage}[t]{8ex}\color{red}\bf
		\begin{verbatim}
			(%i2) 
		\end{verbatim}
	\end{minipage}
	\begin{minipage}[t]{\textwidth}\color{blue}
		\begin{verbatim}
			FPS(1/(1+sin(z)),z,n);
		\end{verbatim}
	\end{minipage}
	\definecolor{labelcolor}{RGB}{100,0,0}
	\begin{equation}
		\displaystyle
		\parbox{10ex}{$\color{labelcolor}\mathrm{\tt (\%o2) }\quad $}
		\hspace{-0.5cm} \Bigg[\sum_{n=0}^{\infty }{{A}_{n}}\cdot {{z}^{n}},{{A}_{n+2}}=\frac{5\cdot {{A}_{n}}+3\cdot \sum_{k=1}^{n-1}{{A}_{k}}\cdot {{A}_{n-k}}}{\left( n+1\right) \cdot \left( n+2\right)},n>=0,\left[{{A}_{0}}=1,{{A}_{1}}=-1\right]\Bigg]
	\end{equation}
\end{example} 

Note that these outputs can also be used to compute Taylor polynomials. We implemented it as \textit{QTaylor(f(z),z,z0,N)}. However, due to the presence of summation terms the quadratic time complexity cannot be avoided and hence the code is generally slower than the built-in Maxima command \textit{taylor}. Nevertheless one can use a remembering program so that many calls of close orders of the same function require a timing only for the first call.

\begin{example}\item

	\noindent
	\begin{minipage}[t]{8ex}\color{red}\bf
		\begin{verbatim}
			(%i1) 
		\end{verbatim}
	\end{minipage}
	\begin{minipage}[t]{\textwidth}\color{blue}
		\begin{verbatim}
			taylor(log(1+sin(z)),z,0,10);
		\end{verbatim}
	\end{minipage}
	\definecolor{labelcolor}{RGB}{100,0,0}
	\[\displaystyle
	\parbox{10ex}{$\color{labelcolor}\mathrm{\tt (\%o1)/T/ }\quad $}
	z-\frac{{{z}^{2}}}{2}+\frac{{{z}^{3}}}{6}-\frac{{{z}^{4}}}{12}+\frac{{{z}^{5}}}{24}-\frac{{{z}^{6}}}{45}+\frac{61\cdot {{z}^{7}}}{5040}-\frac{17\cdot {{z}^{8}}}{2520}+\frac{277\cdot {{z}^{9}}}{72576}-\frac{31\cdot {{z}^{10}}}{14175}+\mbox{...}
	\]
	\noindent
	\begin{minipage}[t]{8ex}\color{red}\bf
		\begin{verbatim}
			(%i2) 
		\end{verbatim}
	\end{minipage}
	\begin{minipage}[t]{\textwidth}\color{blue}
		\begin{verbatim}
			QTaylor(log(1+sin(z)),z,0,10);
		\end{verbatim}
	\end{minipage}
	\vspace{-0.6cm}
	
	\definecolor{labelcolor}{RGB}{100,0,0}
	\[\displaystyle
	\parbox{10ex}{$\color{labelcolor}\mathrm{\tt (\%o2) }\quad $}\hspace{-0.5cm}
	-\frac{31\cdot {{z}^{10}}}{14175}+\frac{277\cdot {{z}^{9}}}{72576}-\frac{17\cdot {{z}^{8}}}{2520}+\frac{61\cdot {{z}^{7}}}{5040}-\frac{{{z}^{6}}}{45}+\frac{{{z}^{5}}}{24}-\frac{{{z}^{4}}}{12}+\frac{{{z}^{3}}}{6}-\frac{{{z}^{2}}}{2}+z\mbox{}
	\]
	
	\noindent
	\begin{minipage}[t]{8ex}\color{red}\bf
		\begin{verbatim}
			(%i3) 
		\end{verbatim}
	\end{minipage}
	\begin{minipage}[t]{\textwidth}\color{blue}
		\begin{verbatim}
			taylor(log(1+sin(z)),z,0,300)$
		\end{verbatim}
	\end{minipage}
    \vspace{-0.5cm}
    
	\definecolor{labelcolor}{RGB}{100,0,0}
	\mbox{}\\\mbox{Evaluation took 1.2500 seconds (1.2510 elapsed) using 287.200 MB.}
	
	\noindent
	\begin{minipage}[t]{8ex}\color{red}\bf
		\begin{verbatim}
			(%i4) 
		\end{verbatim}
	\end{minipage}
	\begin{minipage}[t]{\textwidth}\color{blue}
		\begin{verbatim}
			QTaylor(log(1+sin(z)),z,0,300)$
		\end{verbatim}
	\end{minipage}
	\vspace{-0.5cm}
	
	\definecolor{labelcolor}{RGB}{100,0,0}
	\mbox{}\\\mbox{Evaluation took 33.5160 seconds (33.5400 elapsed) using 2345.257 MB.}
	
	\noindent
	\begin{minipage}[t]{8ex}\color{red}\bf
		\begin{verbatim}
			(%i5) 
		\end{verbatim}
	\end{minipage}
	\begin{minipage}[t]{\textwidth}\color{blue}
		\begin{verbatim}
			taylor(log(1+sin(z)),z,0,200)$
		\end{verbatim}
	\end{minipage}
	\vspace{-0.5cm}
	
	\definecolor{labelcolor}{RGB}{100,0,0}
	\mbox{}\\\mbox{Evaluation took 0.3280 seconds (0.3410 elapsed) using 88.334 MB}
	
	\noindent
	\begin{minipage}[t]{8ex}\color{red}\bf
		\begin{verbatim}
			(%i6) 
		\end{verbatim}
	\end{minipage}
	\begin{minipage}[t]{\textwidth}\color{blue}
		\begin{verbatim}
			QTaylor(log(1+sin(z)),z,0,200)$
		\end{verbatim}
	\end{minipage}
	\vspace{-0.5cm}
	
	\definecolor{labelcolor}{RGB}{100,0,0}
	\mbox{}\\\mbox{Evaluation took 0.0000 seconds (0.0130 elapsed) using 1.062 MB.}
\end{example}

Notice that our computations of recurrence equations and Taylor polynomials from quadratic differential equations sketch another proof on the existence and uniqueness of the solutions of these differential equations without a use of the implicit function theorem. Therefore, our normal forms are well defined. Note, however, that we do not consider Puiseux series in these computations because our approach to determine Puiseux numbers from recurrence equations generated by \textit{FindQRE} cannot apply. This could be a topic for further studies on non-holonomic Puiseux series. The algebraic geometry approach described in (\cite{Frankenphd}) could be of great help for this purpose.

\subsection{The algorithm as a simplifier}

In this section, we present some new results occurring as consequences using our algorithm of the previous subsection. This is an improvement towards identity proving for non-holonomic expressions. This advantage of symbolic computing was well discussed in the case of hypergeometric identities in (\cite{petkovvsek1996b}). Up to now our algorithms were meant to represent the power series of a given expression A; now we would ask our algorithm to answer the question whether two given expressions A and B have the same power series representations, and therefore whether they are identical in a certain neighborhood. This is the conclusion when two expressions have the same output using our algorithm or when the representation of their difference is zero. We present two examples where this can be observed. In this subsection we use our Maple implementation in order to compare our computations with its built-in command \textit{simplify} which apply many internal simplification rules on its given inputs.

As first example, the expression
\begin{equation}
	\log\left(\tan\left(\dfrac{z}{2}\right) + \sec\left(\dfrac{z}{2}\right)\right) - \arcsinh\left(\dfrac{\sin(z)}{1+\cos(z)}\right) \label{ABeg1}
\end{equation}
from \cite[Section 3.3]{geddes1992algorithms} (see also \cite[Exercise 9.8]{koepf2006computeralgebra}) is known to be difficult to prove equal to zero. One needs non-trivial transformations to simplify this to zero. However, using our algorithm based on the computation of quadratic differential equations yields the same power series representation for

\begin{maplegroup}
	\begin{mapleinput}
		\mapleinline{active}{1d}{f:=ln(tan(z/2)+sec(z/2))
		}{}
	\end{mapleinput}
	\mapleresult
	\begin{maplelatex}
		\[\displaystyle f\, := \,\ln  \left( \tan \left( z/2 \right) +\sec \left( z/2 \right)  \right) \]
	\end{maplelatex}
\end{maplegroup}
\noindent and
\begin{maplegroup}
	\begin{mapleinput}
		\mapleinline{active}{1d}{g:=arcsinh(sin(z)/(cos(z)+1))}{}
	\end{mapleinput}
	\mapleresult
	\begin{maplelatex}
		\[\displaystyle g\, := \,{\rm arcsinh} \left({\frac {\sin \left( z \right) }{\cos \left( z \right) +1}}\right)\]
	\end{maplelatex}
\end{maplegroup}
\noindent as shown below.
\clearpage
\begin{maplegroup}
	\begin{mapleinput}
		\mapleinline{active}{1d}{FPS[FPS](f,z,n,fpstype=Quadratic)
		}{}
	\end{mapleinput}
	\mapleresult
	\begin{maplelatex}
		\begin{small}
		\begin{multline}
			\displaystyle \Mapleoverset{\infty }{\Mapleunderset{n =0}{\sum }}A \left(n \right)~z ^{n },A \left(n +4\right)=-\frac{1}{2~\left(n +2\right)~\left(n +3\right)~\left(n +4\right)}\Bigg(-\frac{\left(n +2\right)~A \left(n +2\right)}{2}\\
			+\left(\Mapleoverset{n }{\Mapleunderset{k =1}{\sum }}4~\left(k +1\right)~\left(k +2\right)~\left(k +3\right)~A \left(k +3\right)~\left(n -k +2\right)~A \left(n -k +2\right)\right)\\
			+\Mapleoverset{n }{\Mapleunderset{k =1}{\sum }}\left(-\left(k +1\right)~A \left(k +1\right)~\left(n -k +2\right)~A \left(n -k +2\right)\right)\\
			+\left(\Mapleoverset{n }{\Mapleunderset{k =1}{\sum }}-8~\left(k +1\right)~\left(k +2\right)~A \left(k +2\right)~\left(n -k +2\right)~\left(n +3-k \right)~A \left(n +3-k \right)\right)\Bigg),\\
			0\le n ,\left[A \left(0\right)=0,A \left(1\right)=\frac{1}{2},A \left(2\right)=0,A \left(3\right)=\frac{1}{48}\right]\nonumber
		\end{multline}
	    \end{small}
	\end{maplelatex}
\end{maplegroup}
\begin{maplegroup}
	\begin{mapleinput}
		\mapleinline{active}{1d}{FPS[FPS](g,z,n,fpstype=Quadratic)}{}
	\end{mapleinput}
	\mapleresult
	\begin{maplelatex}
		\begin{small}
		\begin{multline}
			\displaystyle \Mapleoverset{\infty }{\Mapleunderset{n =0}{\sum }}A \left(n \right)~z ^{n },A \left(n +4\right)=-\frac{1}{2~\left(n +2\right)~\left(n +3\right)~\left(n +4\right)}\Bigg(-\frac{\left(n +2\right)~A \left(n +2\right)}{2}\\
			+\left(\Mapleoverset{n }{\Mapleunderset{k =1}{\sum }}4~\left(k +1\right)~\left(k +2\right)~\left(k +3\right)~A \left(k +3\right)~\left(n -k +2\right)~A \left(n -k +2\right)\right)\\
			+\Mapleoverset{n }{\Mapleunderset{k =1}{\sum }}\left(-\left(k +1\right)~A \left(k +1\right)~\left(n -k +2\right)~A \left(n -k +2\right)\right)\\
			+\left(\Mapleoverset{n }{\Mapleunderset{k =1}{\sum }}-8~\left(k +1\right)~\left(k +2\right)~A \left(k +2\right)~\left(n -k +2\right)~\left(n +3-k \right)~A \left(n +3-k \right)\right)\Bigg),\\
			0\le n ,\left[A \left(0\right)=0,A \left(1\right)=\frac{1}{2},A \left(2\right)=0,A \left(3\right)=\frac{1}{48}\right]\nonumber
    \end{multline}
    \end{small}
	\end{maplelatex}
\end{maplegroup}
\noindent The argument \textit{fpstype=Quadratic} is used to apply our method to non-holonomic functions directly. In fact, our algorithm \textit{QDE} computes the same differential equation for $f$, $g$, and $f-g$. That is
\begin{maplegroup}
	\begin{mapleinput}
		\mapleinline{active}{1d}{FPS[QDE](f-g,F(z))}{}
	\end{mapleinput}
	\mapleresult
	\begin{maplelatex}
\begin{equation}
\displaystyle - \left( {\frac {\rm d}{{\rm d}z}}F \left( z \right)  \right) ^{2}-8\, \left( {\frac {{\rm d}^{2}}{{\rm d}{z}^{2}}}F \left( z \right)  \right) ^{2}+4\, \left( {\frac {{\rm d}^{3}}{{\rm d}{z}^{3}}}F \left( z \right)  \right) {\frac {\rm d}{{\rm d}z}}F \left( z \right) =0.
\end{equation}
\end{maplelatex}
\end{maplegroup}
\noindent Moreover our implementation identifies their difference to zero, which is inaccessible using the built-in command \textit{simplify}.
\begin{maplegroup}
	\begin{mapleinput}
		\mapleinline{active}{1d}{CPUTime(FPS[FPS](f-g,z,n,fpstype=Quadratic))
		}{}
	\end{mapleinput}
	\mapleresult
	\begin{maplelatex}
		\[\displaystyle 85.297, 0\]
	\end{maplelatex}
\end{maplegroup}
$85.297$ indicates the CPU time used.
\clearpage
\begin{maplegroup}
	\begin{mapleinput}
		\mapleinline{active}{1d}{simplify(f-g,trig)
		}{}
	\end{mapleinput}
	\mapleresult
	\begin{maplelatex}
		\[\displaystyle \ln  \left( {\frac {1+\sin \left( z/2 \right) }{\cos \left( z/2 \right) }} \right) -{\rm arcsinh} \left({\frac {\sin \left( z \right) }{\cos \left( z \right) +1}}\right)\]
	\end{maplelatex}
\end{maplegroup}

Next, we consider
\begin{equation}
	\log\left(\dfrac{1+\tan(z)}{1-\tan(z)}\right) - 2 \arctanh\left(\dfrac{\sin(2z)}{1+\cos(2z)}\right). \label{id2}
\end{equation}
Similarly we get the following computations.              
\begin{maplegroup}
	\begin{mapleinput}
		\mapleinline{active}{1d}{f:=ln((1+tan(z))/(1-tan(z)))
		}{}
	\end{mapleinput}
	\mapleresult
	\begin{maplelatex}
		\[\displaystyle f\, := \,\ln  \left( {\frac {1+\tan \left( z \right) }{1-\tan \left( z \right) }} \right) \]
	\end{maplelatex}
\end{maplegroup}
\begin{maplegroup}
	\begin{mapleinput}
		\mapleinline{active}{1d}{g:=2*arctanh(sin(2*z)/(1+cos(2*z)))
		}{}
	\end{mapleinput}
	\mapleresult
	\begin{maplelatex}
		\[\displaystyle g\, := \,2\,{\rm arctanh} \left({\frac {\sin \left( 2\,z \right) }{\cos \left( 2\,z \right) +1}}\right)\]
	\end{maplelatex}
\end{maplegroup}
\noindent $f$ and $g$ satisfy the same differential equation
\begin{maplegroup}
	\begin{mapleinput}
		\mapleinline{active}{1d}{FPS[QDE](f,F(z))
		}{}
	\end{mapleinput}
	\mapleresult
	\begin{maplelatex}
		\[\displaystyle -4\, \left( {\frac {\rm d}{{\rm d}z}}F \left( z \right)  \right) ^{2}-2\, \left( {\frac {{\rm d}^{2}}{{\rm d}{z}^{2}}}F \left( z \right)  \right) ^{2}+ \left( {\frac {{\rm d}^{3}}{{\rm d}{z}^{3}}}F \left( z \right)  \right) {\frac {\rm d}{{\rm d}z}}F \left( z \right) =0.\]
	\end{maplelatex}
\end{maplegroup}
\noindent Therefore we get the same power series representation
\begin{maplegroup}
	\begin{mapleinput}
		\mapleinline{active}{1d}{FPS[FPS](g,z,n,fpstype=Quadratic)
		}{}
	\end{mapleinput}
	\mapleresult
	\begin{maplelatex}
		\begin{small}
		\begin{multline}
			\displaystyle \Mapleoverset{\infty }{\Mapleunderset{n =0}{\sum }}A \left(n \right)~z ^{n },A \left(n +4\right)=-\frac{1}{2~\left(n +2\right)~\left(n +3\right)~\left(n +4\right)}\Bigg(-8~\left(n +2\right)~A \left(n +2\right)\\
			+\left(\Mapleoverset{n }{\Mapleunderset{k =1}{\sum }}\left(k +1\right)~\left(k +2\right)~\left(k +3\right)~A \left(k +3\right)~\left(n -k +2\right)~A \left(n -k +2\right)\right)\\
			+\left(\Mapleoverset{n }{\Mapleunderset{k =1}{\sum }}-4~\left(k +1\right)~A \left(k +1\right)~\left(n -k +2\right)~A \left(n -k +2\right)\right)\\
			+\left(\Mapleoverset{n }{\Mapleunderset{k =1}{\sum }}-2~\left(k +1\right)~\left(k +2\right)~A \left(k +2\right)~\left(n -k +2\right)~\left(n +3-k \right)~A \left(n +3-k \right)\right)\Bigg),\\
			0\le n ,\left[A \left(0\right)=0,A \left(1\right)=2,A \left(2\right)=0,A \left(3\right)=\frac{4}{3}\right]
		\end{multline}
	    \end{small}
	\end{maplelatex}
\end{maplegroup}

\noindent However in this case $f-g$ satisfies a trivial differential equation, which makes the computations much faster.
\begin{maplegroup}
	\begin{mapleinput}
		\mapleinline{active}{1d}{FPS[QDE](f-g,F(z))}{}
	\end{mapleinput}
	\mapleresult
	\begin{maplelatex}
		\[\displaystyle  {\frac {\rm d}{{\rm d}z}}F \left( z \right)  =0\]
	\end{maplelatex}
\end{maplegroup}
\begin{maplegroup}
\begin{mapleinput}
	\mapleinline{active}{1d}{CPUTime(FPS[FPS](f-g,z,n,fpstype=Quadratic))
	}{}
\end{mapleinput}
\mapleresult
\begin{maplelatex}
	\[\displaystyle 0.046, 0\]
\end{maplelatex}
\end{maplegroup}
\begin{maplegroup}
As with the previous example, this simplification is not accessible with the command \textit{simplify}.
\begin{mapleinput}
	\mapleinline{active}{1d}{simplify(f-g)
	}{}
\end{mapleinput}
\mapleresult
\begin{maplelatex}
	\[\displaystyle \ln  \left( {\frac {\cos \left( z \right) +\sin \left( z \right) }{\cos \left( z \right) -\sin \left( z \right) }} \right) -2\,{\rm arctanh} \left({\frac {\sin \left( 2\,z \right) }{\cos \left( 2\,z \right) +1}}\right)\]
\end{maplelatex}
\end{maplegroup}

\begin{acknowledgment}
This work summarizes the results of the first author's Ph.D. thesis
under the supervision of the second. The first author received a DAAD
STIBET and two DAAD Erasmus Plus scholarships to finish his work.
Therefore, both authors would like to thank DAAD and the Institute of
Mathematics of the University of Kassel for their important and valuable
support.
\end{acknowledgment}

\bibliographystyle{elsarticle-harv}

\end{document}